\documentclass[dvips,a4paper,titlepage]{article}

\usepackage{myreportpdf}
\usepackage{fixltx2e}
\usepackage{mathrsfs}
\usepackage[sort,noadjust]{cite}

\usepackage[normalem]{ulem}
\newcommand{\sF}{\mathscr{F}}

\newcommand{\B}{\mathscr{B}}
\newcommand{\U}{\mathscr{U}}
\newcommand{\E}{\mathscr{E}}
\newcommand{\bE}{\bar{\mathscr{E}}}

\newcommand{\bS}{\bar{\mathscr{S}}}

\newcommand{\sS}{{\mathscr{S}}}
\newcommand{\sG}{{\mathscr{G}}}
\newcommand{\eF}{{\EuS{F}}}

\newcommand{\Z}{\mathcal{Z}}
\renewcommand{\I}{\mathcal{I}}
\newcommand{\G}{\mathcal{G}}
\renewcommand{\L}{\mathcal{L}}
\renewcommand{\Q}{\mathcal{Q}}
\renewcommand{\V}{\mathcal{V}}

\newcommand{\rD}{\mathrm{D}}
\newcommand{\rd}{\mathrm{d}}

\newcommand{\e}{{\mathrm{e}\,}}

\newcommand{\snr}{{\textsf{snr}}}

\newcommand{\zf}{{\textsf{zf}}}
\newcommand{\bpsk}{{\textsf{bpsk}}}
\newcommand{\qpsk}{{\textsf{qpsk}}}

\newcommand{\crqpsk}{{\textsf{cr-qpsk}}}
\newcommand{\rs}{{\textsf{rs}}}
\newcommand{\rsb}{{\textsf{rsb1}}}

\newcommand{\sFyz}{\gimel_u(y,z)}
\newcommand{\tsFyz}{\tilde\gimel_u(y,z)}
\newcommand{\sFypz}{\gimel_u(\tilde y,z)}
\newcommand{\tsFyzup}{\tilde\gimel_{\upsilon}(y,z)}

\newcommand{\im}{{\textsf{im}}}
\newcommand{\re}{{\textsf{re}}}
\newcommand{\tot}{{\textsf{tot}}}
\newcommand{\dpc}{{\textsf{dpc}}}
\newcommand{\gthp}{{\textsf{gthp}}}

\DeclareMathOperator{\Unif}{\textsf{Unif}}

\floatstyle{ruled}%
\newfloat{floatequation}{thbp}{leq}
\floatname{floatequation}{}

\begin{document}

%!TEX TS-program = xelatex

\title{Vector Precoding for Gaussian MIMO Broadcast Channels: Impact of Replica Symmetry Breaking}
\author{Benjamin M. Zaidel   \and Ralf R. M\"{u}ller \and Aris L. Moustakas \and Rodrigo de Miguel
  \thanks{This work was supported in part by the Research Council of Norway under Grant No.\ 171133/V30, and by the European Commission under Grant ``Newcom++'' No. EU-IST-NoE-FP6-2007-216715. \newline
  \indent This work was presented in part at the 46$^{\rm th}$ Annual Allerton Conference on Communication, Control \& Computing, Monticello, IL, U.S.A, Sep.\ 2008 and the 10$^{\rm th}$ International Symposium on Spread Spectrum Techniques and Applications (ISSSTA), Bologna, Italy, Aug.\ 2008.\newline
\indent Benjamin M.\ Zaidel,  Ralf R.\ M\"{u}ller and Rodrigo de Miguel are with the Department of Electronics and Telecommunications, The Norwegian University of Science and Technology (NTNU), Trondheim, Norway. e-mail: zaidel@iet.ntnu.no, rodrigo@iet.ntnu.no, ralf@iet.ntnu.no.  Aris L.\ Moustakas is with the Physics Department, National and Kapodistrian University of Athens, Athens, Greece. e-mail: arislm@phys.uoa.gr}}
%, The Norwegian University of Science and Technology (NTNU), Trondheim, Norway.}}

\maketitle
%%% ----------------------------------------------------------------------

\begin{abstract}
The so-called ``replica method" of statistical physics is employed for the large system analysis of vector precoding for the Gaussian multiple-input multiple-output (MIMO) broadcast channel. 
The transmitter is assumed to comprise a linear front-end combined with nonlinear precoding, that minimizes the front-end imposed transmit energy penalty. 
Focusing on discrete complex input alphabets, the energy penalty is minimized by relaxing the input alphabet to a larger alphabet set prior to precoding.
For the common discrete lattice-based relaxation, the problem is found to violate the assumption of {\em replica symmetry} and a {\em replica symmetry breaking} ansatz is taken.
The limiting empirical distribution of the precoder's output, as well as the limiting energy penalty, are derived for one-step replica symmetry breaking. 
For convex relaxations, replica symmetry is found to hold and corresponding results are obtained for comparison. 
Particularizing to a ``zero-forcing" (ZF) linear front-end, and non-cooperative users, a decoupling result is derived according to which the channel observed by each of the individual receivers can be effectively characterized by the Markov chain $u$--$x$--$y$, where $u$, $x$, and $y$ are the channel input, the equivalent precoder output, and the channel output, respectively. 
For discrete lattice-based alphabet relaxation, the impact of replica symmetry breaking is demonstrated for the energy penalty at the transmitter. 
An analysis of spectral efficiency is provided to compare discrete lattice-based relaxations against convex relaxations, as well as linear ZF and Tomlinson-Harashima precoding (THP). Focusing on quaternary phase shift-keying (QPSK), significant performance gains of both lattice and convex relaxations are revealed compared to linear ZF precoding, for medium to high signal-to-noise ratios (SNRs). THP is shown to be outperformed as well. In addition, comparing certain lattice-based relaxations for QPSK against a convex counterpart, the latter is found to be superior for low and high SNRs but slightly inferior for medium SNRs in terms of spectral efficiency.
\end{abstract}

%%%%%%%%%%%%%%%
\section{Introduction}
\label{introduction}

The multiple-input multiple-output (MIMO) Gaussian broadcast channel (GBC) is the focus of many research activities, addressing the growing demand for higher throughput wireless systems, and in particular the increasing use of multiple-antenna systems in essentially all modern wireless standards (see, e.g., \cite{Caire-Shamai-Steinberg-Weingarten-2005,Astely-et-l-COMM-MAG-LTE-Review,Li-Lin-Zhang-Roh-COMM-MAG-WiMAX-Review-2009}). The capacity region of the MIMO GBC is the \emph{dirty paper coding}  (DPC) \cite{Costa-83} capacity region \cite{Weingarten-Steinberg-Shamai-2006}, and several attempts have been made in recent years to propose practically oriented approaches for implementing DPC, as e.g., \cite{Zamir-Shamai-Erez-2002,Erez-ten-Brink-2005,Bennatan-Brustein-Caire-Shamai-2006}. DPC still remains, however, a difficult, computationally demanding task, which motivates the search for more practical (suboptimum) precoding alternatives.

Since linear precoding, such as zero-forcing (ZF), leads to reduced performance (especially when the channel is ill-conditioned), much attention has been given to nonlinear precoding schemes.
In particular, lattice-based precoding approaches have often been investigated, as for example the \emph{vector perturbation} approach suggested in \cite{Hochwald-Peel-Swindlehurst-2005} (see also \cite{Fischer-book-2002} for a general framework). The vector perturbation approach  was inspired by the idea of Tomlinson-Harashima precoding (THP) \cite{Tomlinson-1971} \cite{Harashima-Miyakawa-1972}. In this scheme, a scaled \emph{complex integer} vector is added to each data vector, chosen to minimize the energy penalty imposed by a linear zero-forcing (ZF) front-end. A modulo function is employed at the receivers, uniquely determining the transmitted symbols in the absence of noise.  An analogous precoding scheme based on a linear minimum-mean-squared-error (MMSE) front-end was considered in \cite{Schmidt-Joham-Utschick-ETT-2008}.
An approach based on optimizing mutual information was taken in \cite{Payaro-Palomar-ISIT-2009}.
Vector perturbation is however still complex as it involves the solution of an NP-hard integer-lattice least squares problem (commonly implemented using the sphere-decoding algorithm \cite{Agrell-Eriksson-Vardy-Zeger-2002}). Addressing the complexity aspect of the method, related approaches can also be found, e.g., in \cite{Windpassinger-Fischer-Huber-COMM-2004,Taherzadeh-Mobasher-Khandani-2007} (see references therein for additional literature in this framework), where lattice-basis reduction techniques are employed.

The analytical performance analysis of such nonlinear precoding schemes is not at all trivial. It is common to consider, therefore, uncoded symbol error probabilities (via simulations), asymptotic capacity scaling laws and diversity orders (the asymptotic slope of the error probability in the high signal-to-noise ratio (SNR) regime), or to employ Monte-Carlo simulations to obtain information-theoretically achievable rates (see e.g., \cite{Hochwald-Peel-Swindlehurst-2005,Schmidt-Joham-Utschick-ETT-2008,Windpassinger-Fischer-Huber-COMM-2004,Taherzadeh-Mobasher-Khandani-2007,Windpassinger-Fischer-Vencel-Huber-TWC-2004,Stojnic-Vikalo-Hassibi-ICASSP-2006}, and also \cite{Boccardi-Tosato-Caire-IZS-2006} for a semi-tutorial review in this respect). The energy penalty induced by the linear front-end is another commonly addressed performance measure.
 %(essential for \emph{uncoded} communications). 
 A lower bound on the energy penalty based on lattice theoretic arguments can be found in \cite{Ryan-Collings-Clarkson-Heath-TRANS-COMM-2009}.
 % (see also \cite{Ryan-Collings-Clarkson-Heath-ICC-2008}). 
The optimum constellation shaping for a ZF front-end (in terms of the energy penalty), allowing for data to be independently decoded by the users, is investigated in  \cite{Mobasher-Khandani-CWIT-2007}, where a selective mapping technique is introduced based on random coding arguments, implementable using nested lattice coding in a trellis precoding framework (see also \cite{Mobasher-Maddah-Ali-Khandani-ISIT-2009} for a more recent study on selective mapping).

The energy penalty minimization was also investigated in \cite{Muller-Guo-Moustakas-JSAC-2008} where another nonlinear precoding approach in this framework was recently proposed. The transmitter comprises a linear front-end combined with nonlinear precoding. The nonlinear part relies on relaxation of the transmitted symbols' alphabets to larger alphabet sets. The idea is to optimize the vector of transmitted symbols over the extended alphabet sets, so as to minimize the energy penalty imposed by the linear front-end, which is essentially the idea behind vector perturbation. However, a notable feature of this precoding scheme is that it can also  be combined with \emph{convex} extended alphabet sets (in contrast to \cite{Hochwald-Peel-Swindlehurst-2005}), lending themselves to \emph{efficient} practical energy minimization algorithms. It can be considered in this sense as a generalization of the vector perturbation scheme (see also \cite{miguel-SP-09} in this respect).

Another interesting contribution of  \cite{Muller-Guo-Moustakas-JSAC-2008} is the harnessing of statistical physics tools for the analysis of the nonlinear precoding scheme, while considering the large system limit in which both the number of users and the number of transmit antennas grow large, while their ratio goes to some finite constant. One of the main objectives of statistical physics is the quantitative description of macroscopic properties of many-body systems while starting from the fundamental interactions between microscopic elements. In this framework, a general tool for the analysis of random (``disordered'') systems, referred to as the ``replica method", was originally invented for the analysis of spin glasses. The latter term describes a spin orientation that has similarity to the type of location of atoms in glasses, which are random in space but frozen in time \cite{Nishimori-Book-2001}. However, the replica method turns out to have a much wider range of applications (see, e.g., \cite{Nishimori-Book-2001,Mezard-Montanari-Book-Final-2009} for recent tutorial manuscripts). In recent years, in particular following Tanaka's pioneering work \cite{Tanaka-2002}, the method  has been successfully applied to various problems in wireless communications. The replica method has also been recognized by now as an important tool for information-theoretic analyses in cases where ``conventional" random matrix theory does not apply. Although the replica method is heuristic in nature, extensive simulations and exact analytical results in the literature suggest that the replica analysis generally yields excellent approximations in many cases of interest (see again \cite{Nishimori-Book-2001,Mezard-Montanari-Book-Final-2009}, and also, e.g., \cite{Tanaka-2002,Muller-Gerstacker-2004,Guo-Verdu-2005,Guo-Tanaka-Book-Chapter-2008} and references therein).

The replica analysis usually employs a number of underlying assumptions regarding the behavior of the quantities in concern in the large-system limit. One such fundamental assumption is the ``self-averaging" property, which relies on the expectation that macroscopic properties of large random systems converge to deterministic values as the system dimensions grow large. 
Self-averaging is a property of most physical systems with large (or infinite) degrees of freedom, and is the result of the high probability of occurrence of typical events or samples. Nevertheless, in the case of glassy systems this property has been not trivial to prove, since the underlying randomness of the interactions makes the system inherently non-ergodic.
The self-averaging property for the conventional so-called Sherrington-Kirkpatrick (SK) spin glass model \cite{Sherrington-Kirkpatrick-1972}  was first proven in \cite{pastur_shcherbina:91}. 
More recently, \cite{Guerra2002_ThermodynamicLimitSG,Guerra2003_infinite_volume} generalized it to a more general class of  spin glass models and, using an ingenious method, showed that the averages over the disorder actually do converge in the large system limit. 
Even more recently, code-division multiple-access (CDMA) systems were shown to be self-averaging \cite{Korada2010}.
Although the particular system we study is not explicitly covered by the above analysis, it can readily be proved to be self-averaging using the same method. For the sake of space, we will not cover the proof here, however, we will refer to self-averaging as a fundamental property of the large system limit rather than an assumption.
Another common assumption in replica analyses %carried out in the information-theoretic framework 
is that of \emph{replica symmetry} (RS) (see, e.g., \cite{Muller-Guo-Moustakas-JSAC-2008,Tanaka-2002,Guo-Verdu-2005}),
according to which it is assumed that the crosscorrelations between replicated microscopic system configurations are independent of the replica indices. The RS assumption, however, is known to produce incorrect conclusions for certain physical quantities such as, e.g., the minimum energy configuration. This led to the development of the \emph{replica symmetry breaking} (RSB) theory \cite{Nishimori-Book-2001,Talagrand-2006}.
Recently, the full RSB solution of the SK spin glass model, first proposed in \cite{Parisi-1980}, was shown to be an upper bound to the minimum energy configuration \cite{Guerra-2003} and later to be the exact solution of the model \cite{Talagrand-2006}. Apart from its general seminal importance, it is profoundly relevant in the context of vector precoding because the SK-model is a particular case of the more general models discussed in \cite{Muller-Guo-Moustakas-JSAC-2008} and in the sequel.

In this paper we consider a communication system setting in which RSB indeed occurs, and demonstrate the significant impact of the RSB treatment on the validity of the approximations produced by the replica analysis.
We focus here on a wireless MIMO broadcast channel (BC) setting, where the transmitter has $N$ transmit antennas and the $K$ users have single receive antennas. Full channel state information (CSI) is assumed available at the transmitter, while the receivers are cognizant of their own channels only (more on this later). No user-cooperation of any kind is assumed. The received signals are embedded in additive white Gaussian noise (AWGN).
The precoding approach considered in \cite{Muller-Guo-Moustakas-JSAC-2008} is revisited.
Note that the focus in  \cite{Muller-Guo-Moustakas-JSAC-2008} is mainly on presenting the method, and on the derivation of the energy penalty in the asymptotic regime, in which both the number of transmit antennas $N$ and the number of users $K$ go to infinity, while $K/N\rightarrow\alpha < \infty$ (commonly referred to as the \emph{system load}). Furthermore, the analysis in \cite{Muller-Guo-Moustakas-JSAC-2008} is based on the RS assumption. It turns out however that the RS assumption can only produce valid asymptotic approximations in this setting when the extended alphabets are convex sets (see, e.g., supportive simulation results in \cite{miguel-SP-09}). In contrast, for the non-convex alphabets considered in \cite{Muller-Guo-Moustakas-JSAC-2008} these approximations turn out to be rather loose, and produce overoptimistic results, especially as the system load gets close to unity. This behavior can be readily observed by comparing the RS based energy penalty to the asymptotic lower bound of \cite{Ryan-Collings-Clarkson-Heath-TRANS-COMM-2009}.

Here, an alternative analysis is provided based on what is referred to in the statistical physics literature as the \emph{one-step RSB (1RSB) ansatz}, which allows one to search for more general solutions than the RS ansatz, but does not cover the full complexity of solutions of full RSB. 
In addition to an energy penalty analysis, analogous to the one in \cite{Muller-Guo-Moustakas-JSAC-2008}, we complement the results by providing an information-theoretic perspective of the proposed precoding approach. Coded transmissions and achievable throughputs are considered.
The employed performance measure is the normalized spectral efficiency, defined as the total number of bits/sec/Hz \emph{per transmit antenna} that can be transmitted arbitrarily reliably through the broadcast channel.
The limiting marginal conditional distribution of the nonlinear precoder's output which is required for the calculation of spectral efficiency, as well as the limiting energy penalty, are \emph{analytically formulated}. Focusing on a ZF front-end, the spectral efficiency is expressed via the input-output mutual information of the equivalent single-user channel observed by each of the receivers. 
%(assumed aware of the channel's statistics by our main results). 
The analysis is applied next to a particular family of discrete extended alphabet sets (following \cite{Muller-Guo-Moustakas-JSAC-2008}), focusing on a QPSK input, demonstrating the RSB phenomenon. To complete the analysis we repeat the derivations while employing the RS ansatz, which is, as said, adequate for convex relaxation schemes, and the results are then applied to a convex alphabet example \cite{Muller-Guo-Moustakas-JSAC-2008}. For both extended alphabets, numerical spectral efficiency results  indicate \emph{significant} performance enhancement over linear ZF preprocessing for medium to high SNRs. Furthermore, performance enhancement is also revealed compared to a generalized THP approach (which is a popular practical nonlinear precoding alternative for such settings).
Comparison of the two types of extended alphabet examples leads to interesting conclusions regarding the performance vs.\ complexity tradeoff of precoding schemes of the kind considered here.

The remainder of this paper is organized as follows. Section \ref{sec: System Model} describes the system model. Section \ref{sec: Replica Analysis} provides an outline of the replica analysis and includes some general results. In particular, it clarifies the concept of RSB which later results are based upon. 
In order to analyze the mutual information and later the trade-off between spectral and power efficiency of various precoding schemes, we need to characterize the limiting conditional distribution of the precoder output. This task is solved in Section \ref{sec: Limiting Characterization of the Precoder Output} providing a set of nonlinear equations whose solutions characterize the desired distributions. Section \ref{sec: Zero-Forcing Front-End} particularizes to the ZF front-end and shows that the channel model can be represented as an equivalent concatenated single-user channel. Then, it derives the spectral efficiency of this equivalent concatenated channel. 
Section \ref{sec: Replica Symmetry Breaking Example} particularizes the results of the previous sections to a discrete lattice-based alphabet relaxation of QPSK. Numerical solutions of the analytical results are provided. Those based on RSB are shown to match simulation results while those based on RS are demonstrated to fail. Section \ref{sec: A Replica Symmetric Example} is the corresponding counterpart to Section \ref{sec: Replica Symmetry Breaking Example} for convex relaxation. Unlike Section \ref{sec: Replica Symmetry Breaking Example}, it finds the RS ansatz to provide accurate approximations. 
Section \ref{sec: Spectral Efficiency Comparison} presents a comparative analysis of the spectral efficiency of the two alphabet relaxation schemes against some other precoding approaches. Finally, Section \ref{sec: Concluding Remarks} ends this paper with some concluding remarks.

%%%%%%%%%%%%%%%%%%%%%%%%%%%%
\section{System Model}\label{sec: System Model}

Consider the following Gaussian MIMO broadcast channel
\begin{equation}\label{eq: Basic System Model}
\vct{r}=\Mat{H}\vct{t} + \vct{n} %\quad,
\end{equation}
where $\vct{r}_{[K\times 1]}$ is the vector of received signals, $\Mat{H} _{[K\times N]} $ is the (random) complex channel transfer matrix, assumed to be of unit expected row norm, $\vct{t}_{[N\times 1]}$  is the vector of transmitted signals, and $\vct{n}_{[K\times 1]}$ is the vector of i.i.d.\ zero mean proper complex AWGNs at the users' receivers. We denote the noises' spectral levels by $\sigma^2$ so that $\vct{n}\sim \N_c(\vct{0}, \sigma^2 \Mat{I})$.

\begin{figure}[!h]
\centering
\includegraphics[scale=0.35]{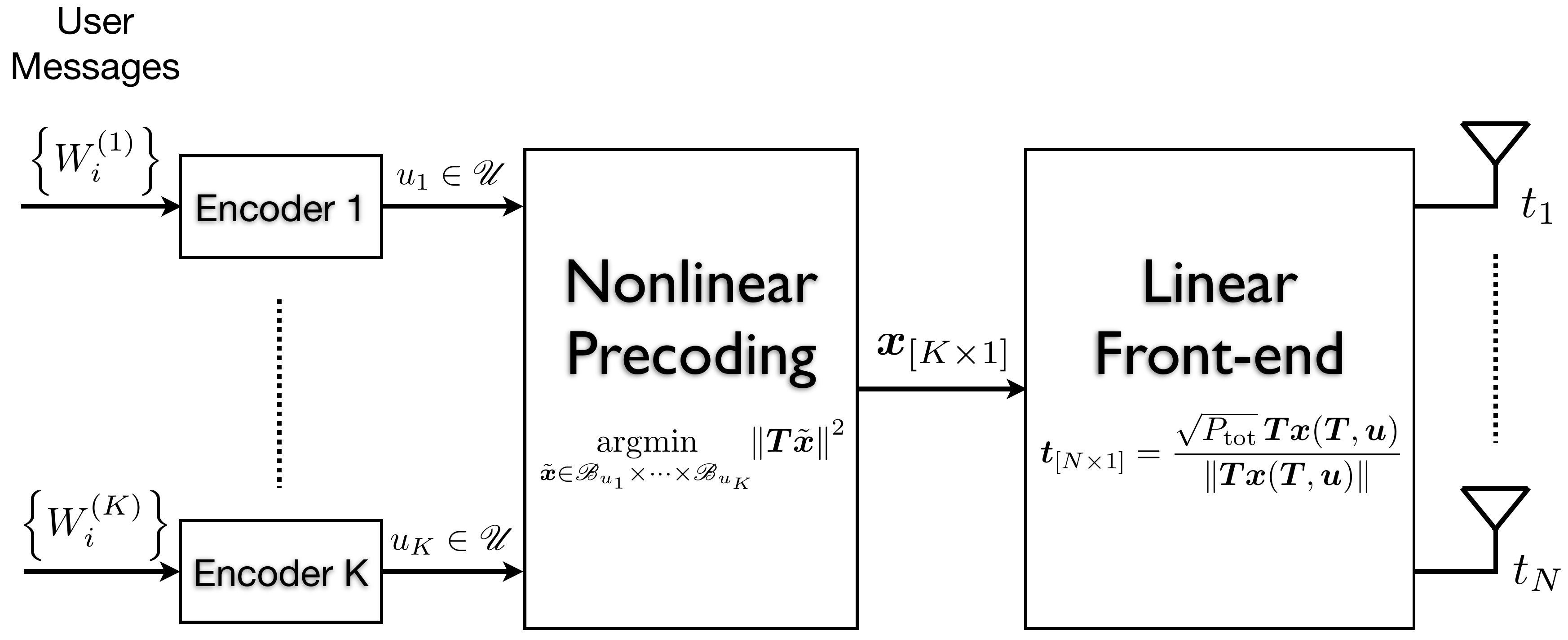}
\caption{Block diagram of the vector precoding scheme.}
\label{fig: System model}
\end{figure}

The precoding process at the transmitter is depicted in Figure \ref{fig: System model}. It is assumed that the users' messages are independently encoded, and that the encoders produce coded symbols $\set{u_k}_{k=1}^K$ taken from some \emph{discrete} alphabet $\U$. These symbols are treated as random variables, independent across users, and subject to the identical underlying discrete probability $P_U(\tilde{u})$, $\tilde{u}\in\U$. We use henceforth for convenience (as shall be made clear in Section \ref{sec: Zero-Forcing Front-End}), the following probability density function (pdf) formulation
\begin{equation}\label{eq: Definition of dP_u}
\rd F_U(u) \triangleq f_U(u) \, \rd u \triangleq \sum_{\tilde{u} \in \U} P_U(\tilde{u}) \delta(u-\tilde{u})\,  \rd u \quad .
\end{equation}
Let $\vct{u}_{[K\times 1]}$ denote the vector of the encoders' outputs, i.e., $\vct{u}=[u_1, \dots,u_K]^T \in \U^K$. The vector $\vct{u}$ is the input to a nonlinear precoding block that minimizes the energy penalty of the precoder through input alphabet relaxation (see below), and outputs a $K \times 1$ vector $\vct{x}$. The vector $\vct{x}$ is then taken as input to the linear front-end block where it is multiplied by the linear front-end matrix $\Mat{T}_{[N \times K]}$, which is, in general, a function of the channel transfer matrix $\Mat{H}$ (note that $\vct{x}$ depends on $\Mat{T}$ and, hence, we can use the functional notation $\vct{x}(\Mat{T},\vct{u})$). The result is then normalized so that the actually transmitted vector $\vct{t}$ satisfies an  instantaneous \emph{total} power (energy per symbol) constraint $P_\tot$, i.e.,
\begin{equation}
\label{eq: Relation of transmitted vector to the vector of coded symbols}
\vct{t}=\sqrt{P_\tot} \, \frac{\Mat{T}\vct{x}(\Mat{T},\vct{u})}{\norm{\Mat{T}\vct{x}(\Mat{T},\vct{u})}} \triangleq \sqrt{\frac{{P_\tot}}{{\mathscr{E}^\tot(\Mat{T},\vct{x})}}} \, \Mat{T}\vct{x}(\Mat{T},\vct{u}) \quad,
\end{equation}
where $\mathscr{E}^\tot(\Mat{T},\vct{x})$ denotes the energy penalty induced by the precoding matrix $\Mat{T}$, and the particular choice of $\vct{x}$, as well as the average symbol energy of the underlying alphabet $\U$ (the explicit dependence on the arguments is omitted henceforth for simplicity). Denoting by $P$ the individual power constraint \emph{per user} (taken as equal for all), so that $P_\tot=K P$, we define  the \emph{transmit} SNR as %$\snr \triangleq 1/(K \sigma^2)$.
\begin{equation}
\label{eq: Definition of snr}
\snr \triangleq \frac{P_\tot}{K \sigma^2} = \frac{P}{\sigma^2} \quad.
\end{equation}
The energy penalty minimization is performed in the following way. The original alphabet $\U$ is extended (``relaxed") to an alphabet $\B=\bigcup_{\tilde{u}\in\U} \B_{\tilde{u}}$, where the sets $\set{\B_{\tilde{u}}}$ are \emph{disjoint}. The idea here is that every coded symbol $u\in \U$ can be represented \emph{without ambiguity} using any element of $\B_u$  \cite{Muller-Guo-Moustakas-JSAC-2008}\footnote{
For practical purposes one would also like to impose additional properties such as a certain minimum distance, although, in principle, the underlying necessary condition is to avoid ambiguity. Note also that the normalization in \eqref{eq: Relation of transmitted vector to the vector of coded symbols} makes the system insensitive to any scaling of the underlying alphabet $\U$.
}. 
%Furthermore, the sets $\set{\B_{\mathsf u}}$ can be chosen in such a way that the minimum distance between signal points is preserved.
The vector $\vct{x}=[x_1,\dots,x_K]^T$ thus satisfies
\begin{equation}
\label{eq: Definition of x as argmin}
\vct{x} = \underset{\vct{\tilde{x}} \in \B_{u_1}\times \cdots \times \B_{u_K}}{\argmin} \norm{\Mat{T} \vct{\tilde{x}}}^2 \quad .
\end{equation}

We note at this point that as an alternative to the normalization taken in \eqref{eq: Relation of transmitted vector to the vector of coded symbols},  ensuring an \emph{instantaneous} transmit power constraint,  a weaker \emph{average} transmit power constraint can be applied, by simply replacing $\mathscr{E}^\tot$ with $E\set{\mathscr{E}^\tot}$, where $E\set{\cdot}$ denotes expectation. However, since we later concentrate on the energy penalty \emph{per symbol},
\begin{equation}
\label{eq: Definition of energy penalty per symbol}
\bE \triangleq \frac{\mathscr{E}^\tot}{K} \quad ,
\end{equation}
and in view of the self-averaging property of the large system limit (as shall be made clear in the following), the two types of energy constraints yield the same asymptotic results. We thus focus for convenience throughout this paper on the instantaneous power constraint (as implied by \eqref{eq: Relation of transmitted vector to the vector of coded symbols}). Note also that in order to differentiate between the energy penalty induced by the precoding scheme, and the effect of the underlying symbol energy of the input alphabet $\U$, one can alternatively represent the results in terms of what we refer to here as the \emph{precoding efficiency}, defined through
\begin{equation}
\label{eq: Factorization of the energy penalty with sig-u}
\zeta \triangleq \frac{\bE}{\sigma_u^2} \quad ,
\end{equation}
where $\sigma_u^2 = E\{\abs{u}^2\}$ (with the expectation taken with respect to \eqref{eq: Definition of dP_u}). 
%The precoding efficiency is completely equivalent to the \emph{inverse} multiuser efficiency \cite{Verdu-Book-98} for linear precoding.

%%%%%%%%%%%%%%%%%%%%%%%%%%%%%%%%%%%%
\section{Outline of the Replica Analysis}\label{sec: Replica Analysis}

In the following we describe the main ideas behind the replica analysis of the problem in hand, and provide a heuristic outline of the approach taken to derive the main results of this paper. The reader is referred to tutorial manuscripts such as \cite{Nishimori-Book-2001,Mezard-Montanari-Book-Final-2009,Guo-Tanaka-Book-Chapter-2008} for an elaborated background on the replica analysis. The fully detailed proofs are deferred to the appendices.

We start here by focusing on the energy penalty, and note that the task of the nonlinear precoding block at the transmitter (see Figure \ref{fig: System model}) can be described as follows. Its task is equivalent to the minimization of an objective function (called the {\em Hamiltonian} in physics literature) having the quadratic form \begin{equation}\label{eq: General definition of the minimization objective function}
\H(\vct{x}) =   \vct{x}^\dag \Mat{J} \vct{x} \quad,
\end{equation}
with $(\cdot)^\dag$ denoting transpose conjugation and $\Mat{J}$ being a random matrix of dimensions $K \times K$.
Thus, the minimum energy penalty per symbol can be expressed as 
\begin{equation}\label{eq: min_energy_penalty}
%\bE = 
\frac{1}{K} \min_{\vct{x}\in \B_{\vct{u}}} \H(\vct{x})  \quad ,
\end{equation}
where we use the shortened notation $\B_{\vct{u}} \triangleq \B_{u_1}\times \cdots \times \B_{u_K}$.
 Note also that to comply with \eqref{eq: Definition of x as argmin} one should take $\Mat{J} = \Mat{T}^\dag\Mat{T}$, however since the results derived in the sequel hold, at least in part, for a more general class of matrices, we retain the formulation as in \eqref{eq: General definition of the minimization objective function}.

To calculate the minimum of the objective function as defined in (\ref{eq: General definition of the minimization objective function}), it is convenient to introduce some notions from statistical physics (see, e.g., \cite{Guo-Tanaka-Book-Chapter-2008}). In particular, we define a discrete probability distribution on the set of state vectors $\set{\vct{x}}$, namely the \emph{Boltzmann distribution},  as
\begin{equation}\label{eq: Boltzmann distribution for the problem in hand}
%P_{\vct x}(\vct{x}) = \frac{1}{\Z} e^{-\beta \H(\vct{x})} 
P_{\mathcal{B}}(\vct{x}) = \frac{1}{\Z} e^{-\beta \H(\vct{x})} 
\end{equation}
where the parameter $\beta>0$ is referred to as the {\em inverse temperature} $\beta=1/T$, while the normalization factor $\Z$ is the so-called {\em partition function}, which is defined as
\begin{equation}\label{eq: definition of Z}
\Z = \sum_{\vct{x}\in \B_{\vct{u}}} e^{-\beta \H(\vct{x})} \ .
\end{equation}
The \emph{energy} of the system is given by
\begin{equation}
\label{eq: Thermodynamics definition of energy1}
%\mathcal{E} = \sum_{\vct{x}\in \B_{\vct{u}}} P_{\vct x}(\vct{x}) \H(\vct{x}) \quad ,
\mathcal{E} = \sum_{\vct{x}\in \B_{\vct{u}}} P_{\mathcal{B}}(\vct{x}) \H(\vct{x}) \quad ,
\end{equation}
and the \emph{entropy} (disorder) is defined as
\begin{equation}
\label{eq: Thermodynamics definition of entropy1}
%\mathcal{S} = - \sum_{\vct{x}\in \B_{\vct{u}}}P_{\vct x}(\vct{x}) \log P_{\vct x}(\vct{x}) \ .
\mathcal{S} = - \sum_{\vct{x}\in \B_{\vct{u}}}P_{\mathcal{B}}(\vct{x}) \log P_{\mathcal{B}}(\vct{x}) \ .
\end{equation}
The definitions above hold for both discrete and continuous alphabets $\B_{\vct u}$. The only difference is that for continuous alphabets the sums over $\vct x\in  \B_{\vct{u}}$ are replaced by integrals.

At thermal equilibrium, the energy of the system is preserved, while the second law of thermodynamics states that the entropy is the maximum possible. This is equivalent to minimizing the \emph{free energy} of the system
\begin{equation}
\label{eq: Thermodynamics definition of the free energy }
\mathcal{F}  \triangleq \mathcal{E} - \frac{\mathcal{S}}\beta \quad,
\end{equation}
where $\beta$, the inverse temperature, is in fact the Lagrange multiplier in the maximization of \eqref{eq: Thermodynamics definition of entropy1}, subject to the mean energy constraint. At equilibrium, the \emph{free energy} can be expressed as
\begin{equation}
\label{eq: Thermodynamics free energy at equilibrium}
\mathcal{F} = - \frac{1}{\beta} \log \Z \quad.
\end{equation}
Note that from Lagrangian duality the Boltzmann distribution \eqref{eq: Boltzmann distribution for the problem in hand} is also the solution to the problem of minimizing the energy for a given entropy.
%%%%%

All mean thermodynamic quantities can now be derived directly from the free energy. In particular, the \emph{energy} of the system is
\begin{eqnarray}
\label{eq: Thermodynamics definition of energy}
\mathcal{E} = \frac{{\rm d}(\beta \mathcal{F}(\beta))}{{\rm d}\beta} \quad ,
\end{eqnarray}
while its thermodynamic \emph{entropy} (disorder) is
\begin{eqnarray}
\label{eq: Thermodynamics definition of entropy}
\mathcal{S} =\beta^2 \frac{\rm d \mathcal{F}(\beta)}{\rm d \beta} \quad .
\end{eqnarray}
In addition to the above quantities we can use the  free energy and the partition function to obtain the empirical joint distribution of the precoder input $u$ and output $x$, which is defined for general $\beta$ as
\begin{eqnarray}\label{eq: joint empirical definition}
% P^{(K)}_{x,u}(x,u)=\frac{1}{K} \sum_{\vct{x}\in \B_{\vct{u}}}P_{\vct x}(\vct{x}) \sum_{k=1}^K  1\left\{(x_k,u_k)=(x,u)\right\} \quad .
% P^{(K)}_{x,u}(\xi,\upsilon)=\frac{1}{K} \sum_{\vct{x}\in \B_{\vct{u}}}P_{\mathcal{B}}(\vct{x}) \sum_{k=1}^K  1\left\{(x_k,u_k)=(\xi,\upsilon)\right\} \quad .
 P^{(K)}_{X,U}(\xi,\upsilon)=\frac{1}{K} \sum_{\vct{x}\in \B_{\vct{u}}}P_{\mathcal{B}}(\vct{x}) \sum_{k=1}^K  1\left\{(x_k,u_k)=(\xi,\upsilon)\right\} \quad .
\end{eqnarray}
Eqs.\ \eqref{eq: Boltzmann distribution for the problem in hand} to \eqref{eq: joint empirical definition} will be useful in deriving some of the results presented in the sequel.

The rationale behind the introduction of the Boltzmann distribution is that as $\beta\rightarrow \infty$, the partition function becomes dominated by the terms corresponding to the \emph{minimum} energy. Hence, taking the logarithm and further normalizing with respect to $\beta$, one gets the desired limiting quantity (energy, entropy or empirical distribution) at the minimum energy subspace of $\B_{\vct{u}}$. Note that even if the energy minimizing vector is not unique, or in fact even if the number of such vectors is exponential in $K$, one still gets the desired quantity when taking the limit $\beta \rightarrow \infty$. 

It is crucial to point out that in the above summation over the set of state-vectors $\B_{\vct{u}}$, both the input vector $\vct{u}$ and the matrix $\Mat{J}$ are {\em fixed}. These random variables are called {\em quenched}. Therefore, all the above manipulations still do not alleviate the difficulty of calculating the desired quantities. In particular, the main difficulty comes from the free energy being a random variable itself, which depends on the particular realizations of  $\Mat{J}$ and $\vct u$. 
As discussed in Section~\ref{introduction},  the proofs of the self-averaging property of the SK-model in \cite{pastur_shcherbina:91,Guerra2002_ThermodynamicLimitSG,Guerra2003_infinite_volume,Korada2010} can be generalized to apply to the form of $\H(\vct{x})$ analyzed here. This means that the free energy converges in probability at the asymptotic limit to a non-random quantity, i.e.,
\begin{equation}\label{eq: Convergence in probability of the free energy}
\lim_{K\rightarrow\infty} \Pr\left( \frac{1}{K} \left|\mathcal{F} - E\set{\mathcal{F}}\right| > \epsilon \right) = 0 \quad \forall \epsilon > 0 \quad ,
\end{equation}
where the expectation $E\{\cdot\}$ is over all realizations of $\Mat{J}$ and $\vct u$. As a result, all quantities that can be obtained from the free energy in an analytic manner, e.g., by differentiation of a parameter, are also self-averaging. 
The empirical joint distribution of the precoder input and output converges to a non-random distribution which is expressed by (\ref{eq: joint empirical as a function of partial derivative of log Z-h}). This self-averaging property makes the problem more straightforward to tackle, since we may now hope to get analytic results for the average of the free energy and its derivatives.

With that in mind, the limiting energy penalty (per symbol) can be represented as
\begin{equation}\label{eq: Representation of the energy penalty through F}
\begin{aligned}
%\bE &= \lim\limits_{K\to\infty}\frac{1}{K} \min_{\vct{x}\in \B_{\vct{u}}} \vct{x}^\dag \Mat{J} \vct{x} =
%  - \lim\limits_{K\to\infty}  \lim_{\beta \rightarrow \infty} \frac{1}{\beta K} E\set{\log \sum_{\vct{x} \in \B_{\vct{u}}} \e^{-\beta \vct{x}^\dag\Mat{J}\vct{x} }} \\ 
%  &
%  %= \lim\limits_{K\to\infty} \lim_{\beta \rightarrow \infty}\frac 1K\, E\set{ \frac{{\rm d}(\beta \mathcal{F}(\beta))}{{\rm d}\beta} }
%  =  \lim\limits_{K\to\infty} \lim_{\beta \rightarrow \infty} E\set{\frac{\mathcal{F}(\beta)}K}\quad .
\bE &= \lim\limits_{K\to\infty}\frac{1}{K} \min_{\vct{x}\in \B_{\vct{u}}} \vct{x}^\dag \Mat{J} \vct{x} =
  - \lim\limits_{K\to\infty}  \lim_{\beta \rightarrow \infty} \frac{1}{\beta K} E\set{\log \sum_{\vct{x} \in \B_{\vct{u}}} \e^{-\beta \vct{x}^\dag\Mat{J}\vct{x} }} \\ 
  &
  %= \lim\limits_{K\to\infty} \lim_{\beta \rightarrow \infty}\frac 1K\, E\set{ \frac{{\rm d}(\beta \mathcal{F}(\beta))}{{\rm d}\beta} }
  =  \lim\limits_{K\to\infty} \lim_{\beta \rightarrow \infty} E\set{\frac{\mathcal{F}(\beta)}K}\quad .
\end{aligned}
\end{equation}
To obtain the empirical joint distribution $P_{X,U}$, we follow a technique very common in the physics literature \cite{Mezard-Montanari-Book-Final-2009}, and 
introduce a dummy variable $h\in \mathbb R$, as well as the function 
\begin{equation}
\label{defsumV}
%V(h,\upsilon,\xi,\vct u,\vct x) = -h \sum_{k=1}^K  1\left\{(x_k,u_k)=(\xi,\upsilon)\right\} \quad.
V(h,\xi,\upsilon,\vct x,\vct u) = -h \sum_{k=1}^K  1\left\{(x_k,u_k)=(\xi,\upsilon)\right\} \quad.
\end{equation}
%where the ``magnetic fields'' $h_s$ are coupled to the indicator function $1\{\cdot\}$. 
If we add this term to $\H(\vct{x})$ in the exponent, the partition function gets modified to
\begin{equation}\label{eq: 1RSB: Definition of script-Z1}
\Z({h})= \sum_{\vct{x} \in \B_{\vct{u}}} \e^{-\beta \left(\H(\vct{x})+V(h)\right)}
\end{equation}
where we have dropped the explicit dependence of $\Z(h)$ and $V(h)$ on $\upsilon$ and $\xi$ (as well as $\vct{u}$ and $\vct{x}$) for the sake of notational compactness.
In the sequel, any dependence on $h$ shall implicitly also indicate a dependence on $\upsilon$ and $\xi$.
Using the above partition function we obtain a modified free energy using (\ref{eq: Thermodynamics free energy at equilibrium}). Upon differentiation with respect to $h$, setting $h= 0$, and letting $\beta\rightarrow\infty$ we get
\begin{align}\label{eq: joint empirical as a function of partial derivative of log Z-h}
%P_{x,u}(\xi,\upsilon) &= \lim\limits_{K\to\infty} P^{(K)}_{x,u}(\xi,\upsilon)\\ 
%&=  \lim\limits_{K\to\infty}\limits\frac{1}{K} \lim_{\beta\rightarrow\infty} E\set{\frac{\partial \mathcal{F}(\beta,{h})}{\partial h}\bigg|_{h= 0}} \quad
P_{X,U}(\xi,\upsilon) &= \lim\limits_{K\to\infty} P^{(K)}_{X,U}(\xi,\upsilon)\\ 
&=  \lim\limits_{K\to\infty}\limits\frac{1}{K} \lim_{\beta\rightarrow\infty} E\set{\frac{\partial \mathcal{F}(\beta,{h})}{\partial h}\bigg|_{h= 0}} \quad
\label{eq317}
\end{align}
where $\mathcal F(\beta, h)$ denotes the free energy for the modified partition function $\mathcal Z(h)$.\footnote{An alternative method for deriving the limiting empirical distribution, which relies on the limiting moments, can be found in \cite{Guo-Verdu-2005}, albeit with more restrictive assumptions on the limiting distribution.}

%%%%%%%%%
The next step in the analysis is to invoke some underlying assumptions. The first assumption is that the random matrix $\Mat{J}$ can be decomposed as
\begin{equation}\label{eq: Definition of the decomposability property}
\Mat{J} = \Mat{U}\Mat{D}\Mat{U}^\dag \quad,
\end{equation}
where $\Mat{D}$ is a diagonal matrix with diagonal elements being the eigenvalues of $\Mat{J}$, and $\Mat{U}$ is a unitary Haar distributed matrix \cite{Mehta-Book-1991}. It is further assumed that the empirical distribution of the diagonal elements of $\Mat{D}$ converges to a nonrandom distribution 
%with non-zero variance 
uniquely characterized by its $R$-transform\footnote{
%For an arbitrary probability distribution function $F(x)$, let $m(s) = \int \frac{dF(x)}{x-s}$ denote its Stieltjes transform. Then, the R-transform of $F(x)$ is $R(w) = m^{-1}(-w) - \frac{1}{w}$, where $m^{-1}(\cdot)$ denotes the inverse function of $m(\cdot)$.
For a definition of the $R$-transform, see Appendix \ref{prooflemma}.
} $R(\cdot)$, which is assumed to exist. 
 
%%%%%%%

Going back to the original communication system model, note that we are in fact interested in the normalized averages of most of the quantities described above, at the limit as $K\rightarrow\infty$. Therefore, to make a distinction, while retaining the relation between the quantities, we shall use henceforth the following notational convention
\begin{eqnarray}
\sF(\beta) &\triangleq& \lim_{K \rightarrow \infty} \frac{1}{K} E\set{\F} \quad , \label{eq: Definition of normalized free energy}\\
\sS(\beta) &\triangleq& \lim_{K \rightarrow \infty} \frac{1}{K} E\set{\mathcal{S}} \quad , \label{eq: Definition of normalized entropy}\\
\E(\beta) &\triangleq& \lim_{K \rightarrow \infty} \frac{1}{K} E\set{\mathcal{E}} \quad . \label{eq: Definition of normalized internal energy}
\end{eqnarray}

Calculating the expectation of a logarithm of a sum of exponents (see \eqref{eq: Representation of the energy penalty through F}) is a formidable task. The standard approach in statistical physics is to invoke the so-called replica ``trick''. The latter is based on the following identity\footnote{An equivalent representation often encountered in the literature is $E\set{\log \Z} = \lim_{n\rightarrow 0^+}\frac{E\set{\Z^n}-1}{n}$.}
\begin{equation}\label{eq: Replica identity}
%E\set{\log \Z} =  \frac{\partial}{\partial n} E\set{ \Z^n} |_{n=0^+}  =  \frac{\partial}{\partial n} \log E\set{ \Z^n} |_{n=0^+} \quad ,
E\set{\log \Z} = \lim_{n\rightarrow 0^+} \frac{\log E\set{ \Z^n}}{n}
\end{equation}
which holds in general for \emph{real} $n$. The ``trick" here relies on the \emph{assumption} that the right hand side (RHS) of \eqref{eq: Replica identity} can be evaluated for \emph{integer} $n$, and that the desired quantity can be found by analytic continuation in the vicinity of $n=0^+$. Although this ``trick" does not \emph{a priori} have any justified validity, its success in statistical physics, and more recently in communications theory, makes it a reasonable approach. Further assuming that the limits with respect to $K$ and $n$ can be interchanged (which is the common practice in replica analyses),  \eqref{eq: Representation of the energy penalty through F} can be rewritten as
\begin{equation}\label{eq: Representation of the energy penalty in terms of replicas}
\begin{aligned}
%\bE &= \lim_{\beta \rightarrow \infty} \E(\beta) \\
%&= -  \lim_{\beta \rightarrow \infty} \lim_{n\rightarrow 0^+} \frac{\partial}{\partial n} \lim_{K \rightarrow \infty} \frac{1}{\beta K} \log E\set{\left(\sum_{\vct{x} \in \B_{\vct{u}}} \e^{-\beta \vct{x}^\dag\Mat{J}\vct{x} } \right)^n} \\
%&= -  \lim_{\beta \rightarrow \infty} \frac{1}{\beta}  \lim_{n\rightarrow 0^+} \frac{\partial}{\partial n} \lim_{K \rightarrow \infty} \frac{1}{K}     \log E \set{  \sum_{\set{\vct{x}_a}}
% \e^{\sum_{a=1}^n -\beta \vct{x}_a^\dag\Mat{J}\vct{x}_a } } \\
%&=-  \lim_{\beta \rightarrow \infty} \frac{1}{\beta}  \lim_{n\rightarrow 0^+} \frac{\partial}{\partial n} \lim_{K \rightarrow \infty} \frac{1}{K} \log E \set{     \sum_{\set{\vct{x}_a}}   \e^{- \Tr (\beta \Mat{J} \sum_{a=1}^n \vct{x}_a \vct{x}_a^\dag )}  }  \quad ,
\bE &= \lim_{\beta \rightarrow \infty} \E(\beta) \\
&= -  \lim_{\beta \rightarrow \infty} \lim_{n\rightarrow 0^+} \frac{1}{n} \lim_{K \rightarrow \infty} \frac{1}{\beta K} \log E\set{\left(\sum_{\vct{x} \in \B_{\vct{u}}} \e^{-\beta \vct{x}^\dag\Mat{J}\vct{x} } \right)^n} \\
&= -  \lim_{\beta \rightarrow \infty} \frac{1}{\beta}  \lim_{n\rightarrow 0} \frac{1}{n} \lim_{K \rightarrow \infty} \frac{1}{K}     \log E \set{  \sum_{\set{\vct{x}_a}}
 \e^{\sum_{a=1}^n -\beta \vct{x}_a^\dag\Mat{J}\vct{x}_a } } \\
&=-  \lim_{\beta \rightarrow \infty} \frac{1}{\beta}  \lim_{n\rightarrow 0} \frac{1}{n} \lim_{K \rightarrow \infty} \frac{1}{K} \log E \set{     \sum_{\set{\vct{x}_a}}   \e^{- \Tr (\beta \Mat{J} \sum_{a=1}^n \vct{x}_a \vct{x}_a^\dag )}  }  \quad ,
\end{aligned}
\end{equation}
where we use the notation $\sum_{\set{\vct{x}_a}} = \sum_{\vct{x}_1 \in \B_{\vct{u}}} \cdots \sum_{\vct{x}_n \in \B_{\vct{u}}}$, and $\Tr(\cdot)$ denotes the trace operator.

The summation over the replicated precoder output vectors $\set{\vct{x}_a}_{a=1}^n$ in \eqref{eq: Representation of the energy penalty in terms of replicas} is performed by splitting the replicas into subshells,  defined through an $n \times n$ matrix $\Mat{Q}$
\begin{equation}\label{eq: Definition of the subshell S-Q}
S(\Mat{Q}) \triangleq 
\set{\vct{x}_1,\dots,\vct{x}_n \big| \vct{x}_a^\dag \vct{x}_b = KQ_{ab}} .
\end{equation}
The limit  $K\to \infty$ allows us to perform the following derivations by saddle point integration.
This first yields the following general result.
\begin{prop}
\label{prop:general_saddle_point_free_energy}
For any inverse temperature $\beta$, any structure of $\Mat Q$ consistent with \eqref{eq: Definition of the subshell S-Q}, and any $R$-transform $R(\cdot)$ such that $R(\Mat Q)$ is well-defined\footnote{\emph{
Note that if $R(\cdot)$ has a series expansion, $R(\Mat Q)$ is well-defined. Since $R(\cdot)$ is the free cumulant generating function, $R(\Mat Q)$ is well-defined, if all moments of the asymptotic eigenvalue distribution of $\Mat J$ exist.
}}, the energy  is given by
\begin{equation}
\label{eq: General expression for energy}
\E(\beta) =  \lim\limits_{n\to 0}\frac{1}{n} \Tr \left[\Mat{Q} \, R(-\beta \Mat{Q})\right] \quad ,
\end{equation}
where $\Mat Q$ is the solution to the saddle point equation
\begin{equation}
\Mat Q = \int \frac{\sum\limits_{{\bf x}\in \B_{u}^n} {\bf xx^\dagger}{\rm e}^{\,-\beta{\bf x}^\dagger R(-\beta\Mat Q)\bf x}}
{\sum\limits_{{\bf x}\in \B_{u}^n} {\rm e}^{\,-\beta {\bf x}^\dagger R(-\beta\Mat Q)\bf x}}
\,{\rm d} F_U(u) \quad 
\label{prop31b}
\end{equation}
with $\B_{u}^n$ denoting the $n$-fold Cartesian product of $\B_u$.
\end{prop}
\begin{proof}
See Appendix \ref{app: Proof of Proposition on Free Energy}.
\end{proof}

With the help of Proposition \ref{prop:general_saddle_point_free_energy}, the energy can be written as 
\begin{align}
\label{RT1}
\E(\beta) &= \lim\limits_{K\to \infty}\frac{1}{K} \frac{\sum\limits_{\Mat x\in \B_{\Mat u}} \Mat x^\dagger\Mat{Jx} \, {\rm e}^{\,-\beta \Mat x^\dagger\Mat{Jx}}}
  {\sum\limits_{\Mat x\in \B_{\Mat u}}  {\rm e}^{\,-\beta \Mat x^\dagger\Mat{Jx}}}\\
   &= \lim\limits_{n\to 0}\frac{1}{n}  \int \frac{\sum\limits_{{\bf x}\in \B_{u}^n} {\bf x}^\dagger R(-\beta\Mat Q){\bf x}\, {\rm e}^{\,-\beta{\bf x}^\dagger R(-\beta\Mat Q)\bf x}}
{\sum\limits_{{\bf x}\in \B_{u}^n} {\rm e}^{\,-\beta {\bf x}^\dagger R(-\beta\Mat Q)\bf x}}
\,{\rm d} F_U(u) 
\label{RT2}
\end{align}
with $\Mat Q$ given by \eqref{prop31b}.
In \eqref{RT1}, $\Mat x$ is a $K$-dimensional vector and its components represent users. The contributions of the users to the energy arise due to the inner product $\Mat x^\dagger\Mat{Jx}$ and are coupled, unless $\Mat J$ is diagonal. 
In \eqref{RT2}, $\bf x$ is an $n$-dimensional vector and its components represent replicas of the \emph{same} user. The contributions of the users to the energy arise due to integration over the distribution $F_U(u)$, and are decoupled and additive. 
This is just another incarnation of the decoupling principle that, under the assumption of replica symmetry, was addressed in \cite{Guo-Verdu-2005}.
Here, we find that it holds for the energy of general (also replica symmetry breaking) spin glass systems and their equivalents in communication theory.

Another interesting observation is the following. In \cite{Tse-Free-Probability-Allerton-99}, an analogy between the $R$-transform and effective interference in linear MMSE detection was discovered, and the additivity of the effective interference of coupled users was explained based on the additivity of the $R$-transforms of free random variables. 
Relying on the code symbols of different users $\set{u_k}$ being i.i.d., we can rewrite \eqref{RT2} as
\begin{align}
\E(\beta) = & \lim\limits_{K\to \infty}\frac{1}{K} \sum_{k=1}^K  \underbrace{\lim\limits_{n\to 0}\frac{1}{n} \frac{\sum\limits_{{\bf x}\in \B_{u_k}^n} {\bf x}^\dagger R(-\beta\Mat Q){\bf x}\, {\rm e}^{\,-\beta{\bf x}^\dagger R(-\beta\Mat Q)\bf x}}
{\sum\limits_{{\bf x}\in \B_{u_k}^n} {\rm e}^{\,-\beta {\bf x}^\dagger R(-\beta\Mat Q)\bf x}}}_{\E_k(\beta)}
\label{RT3}
\end{align}
and interpret $\E_k(\beta)$ as the effective energy of user $k$. 
Like the effective interference in \cite{Tse-Free-Probability-Allerton-99}, it depends only on the signal constellation of user $k$ and the $R$-transform, and it is additive among users.
In contrast to \cite{Tse-Free-Probability-Allerton-99}, \eqref{RT3} is more general and neither constrained to linear detectors nor to Gaussian symbol alphabets. 

To produce explicit results, note that the limit $n\to0$ of the $n\times n$ matrix $\Mat Q$ can only be defined imposing a certain structure onto the crosscorrelation matrix $\Mat{Q}$ at the saddle point, unless the summations in \eqref{prop31b} can be evaluated explicitly, e.g., for $\B_u=\mathbb{C}$. The simplest structure is that of replica symmetry, which in the current setting boils down to
\begin{equation}\label{eq: RS: Definition of the RS structure of Q}
\Mat{Q} = q_0 \Mat{1}_{n \times n} + \frac{\chi_0}{\beta} \Mat{I}_{n \times n} \quad ,
\end{equation}
for some constants $\set{q_0, \chi_0}$.
The 1RSB assumption leads to a more involved structure, formulated as
\begin{equation}
\label{eq: 1RSB: Definition of the 1RSB structure of Q}
\Mat{Q} = q_1 \Mat{1}_{n \times n} + p_1 \Mat{I}_{\frac{n\beta}{\mu_1} \times \frac{n\beta}{\mu_1}}\otimes  \Mat{1}_{\frac{\mu_1}{\beta} \times \frac{\mu_1}{\beta}} + \frac{\chi_1}{\beta} \Mat{I}_{n \times n} \quad ,
\end{equation}
using the constants $\set{q_1,p_1,\chi_1,\mu_1}$.
The above constants (i.e., $\set{q_0,\chi_0}$ for RS, and $\set{q_1,p_1,\chi_1,\mu_1}$ for 1RSB) are referred to as \emph{macroscopic} parameters, and obtained from the corresponding saddle point equations. The limiting energy penalty can then be expressed in terms of these macroscopic parameters, as shown in the following sections. An analogous procedure can be employed to obtain the limiting empirical joint distribution of the precoder input and output using \eqref{eq: joint empirical as a function of partial derivative of log Z-h}.
%

%Replica symmetry breaking is not limited to one step, though we will only present precoding results up to the accuracy of 1RSB.
%For the interested reader and the sake of completeness, we include general results on multiple-step RSB in Appendix~\ref{rstepRSB}.

Replica symmetry breaking is not limited to one step, and in fact in order to \emph{exactly} characterize the limiting energy penalty and precoder output statistics, we would eventually need to consider full RSB, as discussed in Section \ref{introduction}. However, we will only present here precoding results up to the accuracy of 1RSB for purposes of analytical tractability. For the interested reader and the sake of completeness, we include general results on multiple-step RSB in Appendix~\ref{rstepRSB}.

%%%%%%%%%%%%%%%%%%%%%%%%%%%%%%%%%%
\section{Limiting Characterization of the Precoder Output}
\label{sec: Limiting Characterization of the Precoder Output}

%In order to \emph{exactly} characterize the precoder output statistics, we would eventually need to consider full RSB, as discussed in Section \ref{introduction}.
%For the sake of analytical tractability, however, we restrict ourselves to 1RSB in the following.
%As demonstrated in the sequel, when compared to simulation results at finite numbers of antennas, 1RSB gives quite accurate approximations for the quantities of interest, while the RS ansatz does so only in special cases.

We restrict ourselves in the following to 1RSB analysis of the limiting characteristics of the precoder output.
As demonstrated in the sequel, when compared to simulation results at finite numbers of antennas, 1RSB gives quite accurate approximations for the quantities of interest, while the RS ansatz does so only in special cases.

\subsection{The 1RSB Solution}\label{subsec: The 1RSB Solution}
Applying the 1RSB ansatz, as outlined in Section \ref{sec: Replica Analysis} (see in particular \eqref{eq: 1RSB: Definition of the 1RSB structure of Q}), the limiting properties of the precoder output are characterized by means of four macroscopic parameters $q_1, p_1, \chi_1, \mu_1 \in (0,\infty) $, which are determined as specified below. Let $\Mat{J}$ be  a $K \times K$ random matrix satisfying the decomposability property \eqref{eq: Definition of the decomposability property}, and let $R(\cdot)$ denote the $R$-transform of its limiting eigenvalue distribution. Consider now the following function of complex arguments
\begin{equation}
\begin{aligned}
\sFyz &\triangleq   \e^{- \mu_1\underset{x \in \B_{u}}{\min} \varepsilon_1  \abs{x}^2 -2  \Re\set{x (f_1 z^* + g_1y^*)}   } \quad, \quad (y, z) \in \mathbb{C}^2 \quad ,
\end{aligned}
\label{eq: 1RSB: Definition of the function sF}
\end{equation}
where $\Re\set{\cdot}$  takes the real part of the argument, and the parameters $\varepsilon_1$, $g_1$ and $f_1$ are defined as
\begin{eqnarray}
\varepsilon_1 & = &   R(-\chi_1) \label{eq: 1RSB: Expression for epsilion after taking derivative wrt q p and chi} \ ,\\
g_1 & = &
\sqrt{\frac{R(-\chi_1) - R(-\chi_1 - \mu_1 p_1)}{\mu_1}}
\label{eq: 1RSB: Expression for g after taking derivative wrt q p and chi} \ ,\\
f_1 & = & \sqrt{q_1 R'(-\chi_1-\mu_1 p_1)} \ .
\label{eq: 1RSB: Expression for f after taking derivative wrt q p and chi}
\end{eqnarray}
Furthermore, denote its normalized version by
\begin{equation}
\tsFyz = \frac{\sFyz}{\int_\mathbb{C} \sFypz {\rm d}\tilde y}
\end{equation}
to compact notation.
Then, using the shortened notations (for $z_\re,z_\im \in \mathbb{R}$)
\begin{equation}\label{eq: Shortened notation for the Dz integral}
\int_\mathbb{C} (\cdot) \, \rD z \triangleq \int\limits_{-\infty}^\infty \int\limits_{-\infty}^\infty (\cdot) \, \frac{\e^{-\abs{z}^2}}{\pi} \, \rd z_\re \, \rd z_\im \quad , \quad  z\triangleq z_\re+j z_\im \in \mathbb{C} \quad ,
\end{equation}
and
\begin{equation}\label{eq: Shortened notation for the DyDz integral}
\int_{\mathbb{C}^2} (\cdot) \, \rD y \, \rD z \triangleq \int_\mathbb{C}\int_\mathbb{C} (\cdot) \, \rD y \, \rD z ,
\end{equation}
the parameters $\set{q_1, p_1, \chi_1, \mu_1}$ are given by the solutions to the four coupled equations\footnote{In general these coupled equations have multiple solutions and one needs to choose the solution that minimizes the energy penalty.}
\begin{align}
\chi_1  +  p_1\mu_1 &=
 \frac{1}{f_1} \iint_{\mathbb{C}^2}  \Re\set{z^*\underset{x \in \B_{u}}{\argmin} \abs{f_1 z + g_1 y - \varepsilon_1 x}}  %\hfill \cdot
  \tsFyz  \rD y
 \, \rD z \, \rd F_U(u) \ , \label{eq: 1RSB: Equation after partial derivative wrt f - beta inf h zero}\\
 \chi_1 + (q_1 + p_1)\mu_1 &=
\frac1{g_1}  \iint_{\mathbb{C}^2}  \Re\set{y^*\underset{x \in \B_{u}}{\argmin} \abs{f_1 z + g_1 y - \varepsilon_1 x}}   \tsFyz   \rD y
 \, \rD z \, \rd F_U(u) \ , \label{eq: 1RSB: Equation after partial derivative wrt g  - beta inf h zero}
\end{align}
\begin{equation}
q_1+p_1 =
\iint_{\mathbb{C}^2} \abs{\underset{x \in \B_{u}}{\argmin} \abs{f_1 z + g_1 y - \varepsilon_1 x}}^2
\tsFyz  \rD y
 \, \rD z \, \rd F_U(u) \ , \label{eq: 1RSB: Equation after partial derivative wrt epsilon - beta inf h zero}
\end{equation}
and
%%\small
\begin{multline}
 \int\limits_{\chi_1}^{\chi_1+\mu_1 p_1}
 R(-w)\, \rd w =
  \iint_{\mathbb{C}} \log\left( \int_{\mathbb{C}}  \sFyz  \, \rD y \right)\, \rD z \, \rd F_U(u)  \\ -2\chi_1R(-\chi_1) +(\mu_1 q_1  +2 \chi_1 +2\mu_1 p_1) R(-\chi_1 -\mu_1p_1) \\  - 2\mu_1 q_1(\chi_1 + \mu_1 p_1)R'(-\chi_1 - \mu_1p_1) \ .\label{eq: 1RSB: Equation after partial derivative wrt mu - beta inf and h zero}
\end{multline}

The limiting properties of the precoder outputs can now be summarized by means of the following two propositions. The detailed proofs are provided in Appendices \ref{proof41} and \ref{proof42}, respectively.
\begin{prop}
\label{prop: 1RSB: Limiting Energy Penalty}
Suppose the random matrix $\Mat{J}$ satisfies the decomposability property \eqref{eq: Definition of the decomposability property}. Then under some technical assumptions, including in particular \emph{one-step replica symmetry breaking}, 
%and a self averaging property (\ref{eq: Convergence in probability of the free energy}), 
the effective energy penalty per symbol $ \E^\tot / K$ converges in probability as $K,N \rightarrow\infty$, $K/N\rightarrow\alpha<\infty$, to
\begin{equation}
\bE_\rsb \triangleq
  \left(q_1+p_1+\frac{\chi_1}{\mu_1}\right)R(-\chi_1-\mu_1p_1) - \frac{\chi_1}{\mu_1}R(-\chi_1)  - q_1(\chi_1 + \mu_1 p_1) R'(-\chi_1 - \mu_1 p_1) \ .
 \label{eq: 1RSB: Limit of the effective energy penalty}
\end{equation}
\end{prop}
The conditional limiting empirical distribution of the precoder's outputs is specified next.
\begin{prop}\label{prop: 1RSB: Marginal limiting conditional distribution of x given u}
With the same underlying assumptions as in Proposition \ref{prop: 1RSB: Limiting Energy Penalty}, the limiting conditional empirical distribution of the nonlinear precoder's outputs given an input symbol $u$ satisfies
\begin{equation}
%P_{x|u}(x|u) = \int_{\mathbb{C}^2} 1{\set{x = \underset{\tilde{x} \in \B_{u}}{\argmin} \abs{f_1 z + g_1 y - \varepsilon_1 \tilde{x}}}}
%\tsFyz  \rD y \, \rD z  \ , \label{eq: 1RSB: Equation for the limiting conditional distribution - Infinite beta}
P_{X|U}(\xi|\upsilon) = \int_{\mathbb{C}^2} 1{\set{\xi = \underset{{x} \in \B_{\upsilon}}{\argmin} \abs{f_1 z + g_1 y - \varepsilon_1{x}}}}
\tsFyzup  \rD y \, \rD z  \ , \label{eq: 1RSB: Equation for the limiting conditional distribution - Infinite beta}
\end{equation}
where $1\! \set{\cdot}$ denotes the indicator function.
\end{prop}

%%%%%%%%%%%%%%%%%%%%
\subsection{A Replica Symmetric Reduction}\label{subsec: A Replica Symmetric Reduction}

Although the 1RSB solution of the replica analysis leads in principle to a more accurate description of the large system limit, corresponding results can also be derived using the simplifying assumption that the system exhibits a replica symmetric behavior (see \eqref{eq: RS: Definition of the RS structure of Q}). These results shall be used in the sequel to demonstrate the impact of replica symmetry breaking. However they can also be extremely useful for more conveniently analyzing settings that do exhibit replica symmetric properties, such as the case of convex extended alphabet sets addressed in \cite{miguel-SP-09}. A convex alphabet example is discussed in Section~\ref{sec: A Replica Symmetric Example}.

The limiting energy penalty under the RS assumption was in fact already derived in \cite{Muller-Guo-Moustakas-JSAC-2008}, and the result is recalled in the following proposition. The result is given in terms of the two macroscopic parameters $q_0,\chi_0 \in (0,\infty)$, which are obtained through the solution of the two coupled equations
\begin{equation}\label{eq: Fixed-point equation for q-0}
q_0 = \int  \int_{\mathbb{C}} \abs{\underset{x\in \B_u}{\argmin} \abs{z - \frac{R(-\chi_0)\, x }{\sqrt{q_0 R'(-\chi_0)}}}}^2 \, \rD z \, \rd F_U(u) \ ,
\end{equation}
and
\begin{equation} \label{eq: Fixed-point equation for chi-0}
\chi_0 = \frac{\int  \, \Re \bigl\{\int_{\mathbb{C}} \underset{x\in \B_u}{\argmin} \abs {z - \frac{R(-\chi_0)\, x }{\sqrt{q_0 R'(-\chi_0)}}} z^* \, \rD z \bigr\} \, \rd F_U(u) }{\sqrt{q_0 R'(-\chi_0)}} \ .
\end{equation}

\begin{prop}[\cite{Muller-Guo-Moustakas-JSAC-2008}, Proposition 1]
\label{prop: RS: MGM Limiting Energy Penalty}
Suppose the random matrix $\Mat{J}$ satisfies the decomposability property \eqref{eq: Definition of the decomposability property}. Then under some technical assumptions, including in particular \emph{replica symmetry}, 
%and a self averaging property (\ref{eq: Convergence in probability of the free energy}), 
the effective energy penalty per symbol $ \E^\tot / K$ converges in probability as $K,N \rightarrow\infty$, $K/N\rightarrow\alpha<\infty$, to\footnote{\em In \cite{Muller-Guo-Moustakas-JSAC-2008}, the self-averaging property was stated as an assumption, since the authors were not aware of \cite{Guerra2002_ThermodynamicLimitSG}.}
\begin{equation}
\bE_\rs \triangleq q_0 [R(-\chi_0) - \chi_0 R'(-\chi_0)]
   \ .
 \label{eq: RS: Limit of the effective energy penalty}
\end{equation}
\end{prop}

The limiting conditional distribution of the precoder outputs can also be characterized under the RS assumption, in an analogous manner to Proposition \ref{prop: 1RSB: Marginal limiting conditional distribution of x given u}.
\begin{prop}\label{prop: RS Marginal limiting conditional distribution of x given u}
With the same underlying assumptions as in Proposition \ref{prop: RS: MGM Limiting Energy Penalty}, the limiting conditional empirical distribution of the nonlinear precoder's outputs given an input symbol $u$ satisfies
\begin{equation} \label{eq: RS: Equation for the limiting conditional distribution - Infinite beta}
%P_{x|u}(x|u) = \int_{\mathbb{C}} 1{\set{x = \underset{\tilde{x}\in\B_u}{\argmin} \abs{z - \frac{R(-\chi_0) \tilde{x}}{\sqrt{q_0 R'(-\chi_0)}} }}}  \, \rD z \ .
P_{X|U}(\xi|\upsilon) = \int_{\mathbb{C}} 1{\set{\xi = \underset{{x}\in\B_\upsilon}{\argmin} \abs{z - \frac{R(-\chi_0) {x}}{\sqrt{q_0 R'(-\chi_0)}} }}}  \, \rD z \ .
\end{equation}
This is the measure of the corresponding Voronoi region in the scaled \emph{conditional} signal constellation $\B_\upsilon$, with respect to the (complex) Gaussian probability measure.
\end{prop}
\begin{proof} The proof follows the same steps as in the proof of Proposition \ref{prop: 1RSB: Marginal limiting conditional distribution of x given u}, while replacing \eqref{eq: 1RSB: Definition of the 1RSB structure of Q} with \eqref{eq: RS: Definition of the RS structure of Q}.
\end{proof}

%%%%%%%%%%%%%%%%%%
\subsection{Zero-Temperature Entropy}

One way to demonstrate the degree of consistency of the RS and 1RSB solutions is to look at their limiting (thermodynamic) zero-temperature entropy defined as $\bS=\lim_{\beta \to \infty}\sS(\beta)$.
It can also be obtained in a manner similar to Propositions~\ref{prop: 1RSB: Limiting Energy Penalty} and \ref{prop: RS: MGM Limiting Energy Penalty}.
In Appendix~\ref{proof43}, we show:
\begin{prop}\label{prop: 1RSB entropy result}\label{Cor: RS entropy result}
With the same underlying assumptions as in Propositions \ref{prop: 1RSB: Limiting Energy Penalty} or \ref{prop: RS: MGM Limiting Energy Penalty}, the limiting entropy per symbol converges to
\begin{equation}
\label{eq: 1RSB Limiting Entropy per symbol}
%\bS_\rsb \triangleq 
\bS = \chi R(-\chi) - \int_0^{\chi} R(-w) \, \rd w \quad
\end{equation}
with $\chi$ denoting $\chi_1$ and $\chi_0$ for 1RSB and RS, respectively.
\end{prop}
Proposition \ref{prop: 1RSB entropy result}  gives rise to the conjecture that the entropy for general $r$-step RSB is given by $\chi_r R(-\chi_r) - \int_0^{\chi_r} R(-w) \, \rd w$ (see \eqref{eq: rRSB: Definition of the rRSB structure of Q} for the definition of the macroscopic parameters in the general setting).

In any stable thermodynamic system the entropy is non-negative for all temperatures. However, one of the main pitfalls of the RS solution of the original SK-model is that its zero-temperature entropy is negative, indicating an instability \cite{Nishimori-Book-2001}. 
For all $R$-transforms that are strictly increasing functions of negative real arguments, Proposition \ref{Cor: RS entropy result} clearly implies that the entropy is always negative, becoming zero only when the zero temperature value of $\chi_1$, respectively $\chi_0$, approaches zero.
While the full RSB solution has been shown to have vanishing entropy at zero temperature and corresponds to the correct solution, the following lemma proven in Appendix \ref{prooflemma}, indicates that negative entropy is a rather common effect for finite RSB steps.

\begin{lem}
\label{incRtrafo}
The R-transform, wherever its derivative with respect to a real argument exists, is an increasing function.
If the probability distribution  is different from a single mass point, the R-transform is strictly increasing.
\end{lem} 

Note that the above argument for the entropy holds only for discrete state variables. In the case of continuous alphabets, the (then differential) entropy of a system can in fact be negative. Therefore, a negative zero-temperature entropy is not an alarm bell \emph{per se}.
For discrete state variables, the zero-temperature entropy serves as a measure of accuracy: the closer it is to zero, the better the approximation.

%%%%%%%%%%%%%%%%%%%%%%%%%%%%%%%%%%%%%%%%%
\section{Zero-Forcing Front-End}\label{sec: Zero-Forcing Front-End}

To gain more insight into the impact  on system performance of the nonlinear precoding scheme under investigation, we now particularize to a specific linear front-end, namely the ZF front-end.
The precoding matrix $\Mat{T}$ in this case is given by the pseudo-inverse of the channel transfer matrix, which we write here as
\begin{equation}
\label{eq: Definition of the zero-forcing precoder}
\Mat{T} = \Mat{H}^+ = \lim\limits_{\epsilon\to0}\Mat{H}^\dag\left(\Mat{H}\Mat{H}^\dag+\epsilon\mbox{\bf I}\right)^{-1} \quad .
\end{equation}

The underlying assumptions are that $N \ge K$ and that the matrix $\Mat{H}\Mat{H}^\dag$ is
almost surely (a.s.)  positive definite\footnote{In Section~\ref{speceff}, we will also allow for $N<K$ following the treatment in \cite{miguel-TWC-10}.}.
Focusing on the asymptotic regime for $K/N \rightarrow \alpha \le 1$,
then using \eqref{eq: Basic System Model}, \eqref{eq: Relation of transmitted vector to the vector of coded symbols},
and Proposition \ref{prop: 1RSB: Limiting Energy Penalty}, the equivalent single-user channel observed by user $i$ is%can be represented as
\begin{equation}
\label{eq: Normalized Received signal of the i-th user with zero-forcing - large system limit}
\check{r}_i \approx  x_i + \check{n}_i , \quad K \gg 1 \ ,
\end{equation}
where $\check{n}_i$ is a zero mean circularly symmetric complex Gaussian noise with variance $\frac{1}{\rho}$,
\begin{equation}\label{eq: Definition of rho}
\rho \triangleq \frac{\snr}{\bE_\rsb}
\end{equation}
denotes the effective received SNR, and $\bE_\rsb$ is given by \eqref{eq: 1RSB: Limit of the effective energy penalty}.

\begin{prop}
\label{prop: Decoupling result for Zero-Forcing}
Employing the same underlying assumptions  as in Proposition \ref{prop: 1RSB: Limiting Energy Penalty}, then with a ZF front-end the channel observed by a randomly chosen user is equivalent in the large system limit to a concatenated single-user channel, with input $u \in \U$, intermediate output $x \in \B_u$, and final output $y \in \mathbb{C}$, specified by the Markov chain $u$--$x$--$y$ as shown in Figure \ref{fig: ZF equivalent SU channel model}. This Markov chain is defined by the following joint probability density function
%\begin{equation}
%\label{eq: Joint probability in decoupling proposition}
%dP_{u,x,y} = dP_u(u) dP_{x|u}(x|u) p_{y|x}(y|x) \quad,
%\end{equation}
\begin{equation}
\label{eq: Joint probability in decoupling proposition}
%f_{uxy}(u,x,y) = f_u(u) f_{x|u}(x|u) f_{y|x}(y|x) \quad,
f_{UXY}(u,x,y) = f_U(u) f_{X|U}(x|u) f_{Y|X}(y|x) \quad,
\end{equation}
where %$dP_u(u)=\sum_{u_i \in \U} P_u(u_i) \delta(u-u_i)$,
%$f_{x|u}(x|u)=\sum_{\mathsf{x}_i \in \B_u} P_{x|u}(\mathsf{x}_i|u) \delta(x-\mathsf{x}_i)$,
\begin{eqnarray}\label{eq: Definition of conditional probability in decoupling proposition}
%f_{x|u}(x|u) &=& \sum_{\mathsf{x} \in \B_u} P_{x|u}(\mathsf{x}|u) \delta(x-\mathsf{x}) \quad ,\\
f_{X|U}(x|u) &=& \sum_{\tilde{x} \in \B_u} P_{X|U}(\tilde{x}|u) \delta(x-\tilde{x}) \quad ,
\end{eqnarray}
with $P_{X|U}(\mathsf{x}|u)$ given by \eqref{eq: 1RSB: Equation for the limiting conditional distribution - Infinite beta}, and
%$ f_{y|x}(y|x) = \frac{\rho}{\pi} \e^{-\abs{y-x}^2\rho}$
\begin{eqnarray}
%f_{y|x}(y|x) = \frac{\rho}{\pi} \e^{-\abs{y-x}^2\rho}
f_{Y|X}(y|x) = \frac{\rho}{\pi} \e^{-\abs{y-x}^2\rho}
\end{eqnarray}
is the (complex) Gaussian density with mean $x$ and variance $1/\rho$.
\end{prop}
\begin{proof}
The Proposition follows straightforwardly from Proposition \ref{prop: 1RSB: Marginal limiting conditional distribution of x given u} and \eqref{eq: Normalized Received signal of the i-th user with zero-forcing - large system limit}.
\end{proof}
Note that the RS reduction of the above result is readily obtained from Propositions \ref{prop: RS: MGM Limiting Energy Penalty} and \ref{prop: RS Marginal limiting conditional distribution of x given u}, by replacing $\bE_\rsb$ in \eqref{eq: Definition of rho} with $\bE_\rs$ of \eqref{eq: RS: Limit of the effective energy penalty}, and taking \eqref{eq: RS: Equation for the limiting conditional distribution - Infinite beta} for $P_{X|U}(\tilde{x}|u)$ in \eqref{eq: Definition of conditional probability in decoupling proposition}.

\begin{figure}[!t]
\centering
\includegraphics[scale=0.25]{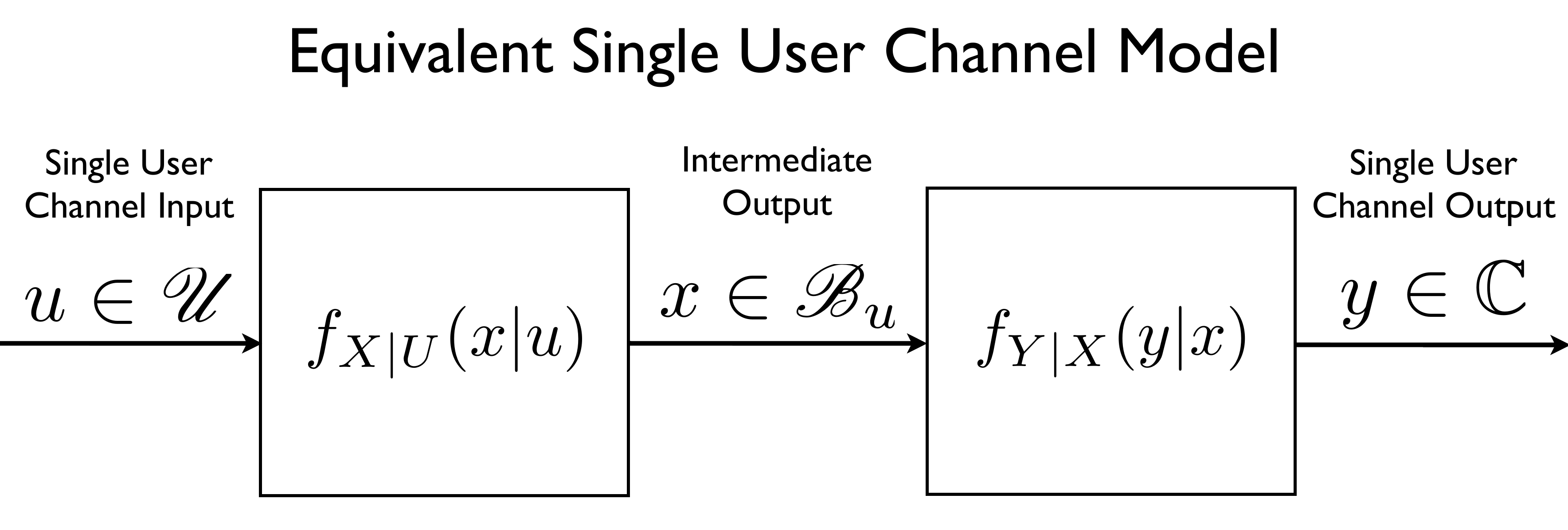}
\caption{Schematic description of the equivalent single user channel model for a ZF front-end.}
\label{fig: ZF equivalent SU channel model}
\end{figure}

The achievable throughput of the nonlinear precoding  scheme can be derived from the equivalent single-user channel model using Proposition \ref{prop: Decoupling result for Zero-Forcing}.
Accordingly, the achievable rate of a randomly chosen user is given by the mutual information\footnote{Note that in the large-system limit the receivers only need information about the state of their own channel, but not about the states of the other channels due to the self-averaging property which makes the impact of the other users' channels and data deterministic.}
%\footnote{The underlying assumption here is that the receiver is fully aware of the channel statistics in view of the limiting results of Propositions \ref{prop: 1RSB: Limiting Energy Penalty} -- \ref{prop: 1RSB: Marginal limiting conditional distribution of x given u}.} 
between the input $u$ and received signal $y$, i.e.,
\begin{equation}
\label{eq: ZF Achievable rate in terms of mutual information}
\begin{aligned}
R &= {\rm I}(u\,;y) = {\rm h}(y) - {\rm h}(y|u) \quad ,
\end{aligned}
\end{equation}
where ${\rm h}(\cdot)$ and ${\rm h}(\cdot|\cdot)$ denote differential entropy and conditional differential entropy, respectively (which can be readily calculated using Proposition \ref{prop: Decoupling result for Zero-Forcing}).
The \emph{normalized} spectral efficiency is then given  by
\begin{equation}
\label{eq: Definition of the spectral efficiency}
C\approx \frac{K}{N} R \xrightarrow[K\rightarrow\infty]{} \alpha R \quad ,
\end{equation}
and it is functionally dependent on the system average $\ebno$ through the relation \cite{Shamai-Verdu-2001-fading}
\begin{equation}\label{eq: Relation between SNR and EbNo}
\snr = \frac{1}{\alpha} C \febno \quad .
\end{equation}

To get a better insight into the impact of the nonlinear precoding scheme, it is useful to compare the results to the spectral efficiency of DPC with Gaussian input (specifying the ultimate performance), as well as to the spectral efficiency of \emph{linear} ZF (for both Gaussian and discrete alphabet input). Another interesting comparison is to the spectral efficiency of generalized THP (GTHP), which is a popular practical nonlinear precoding alternative to the scheme considered here (see, e.g., \cite{Boccardi-Tosato-Caire-IZS-2006}).
For the sake of comparison we further particularize henceforth to the case in which the entries of the channel transfer matrix $\Mat{H}$ are i.i.d.\ zero-mean circularly symmetric complex Gaussian random variables, with variance $1/N$ (``a Gaussian $\Mat{H}$"). Note that in this case the $R$-transform of the limiting eigenvalue distribution of the random matrix $\Mat{J}=(\Mat{H}\Mat{H}^\dag)^{-1}$, and its derivative, simplifiy to \cite{Muller-Guo-Moustakas-JSAC-2008}
\begin{eqnarray}
R(w) &=& \frac{1-\alpha - \sqrt{(1-\alpha)^2-4\alpha w}}{2\alpha w} \quad , \label{eq: R-transform for Gaussian H} \\
R'(w) &=& \frac{\left( 1 - \alpha -  \sqrt{(1-\alpha)^2-4\alpha w}    \right)^2}{4\alpha w^2 \sqrt{(1-\alpha)^2-4\alpha w}} \quad . \label{eq: Derivative of R-transform for Gaussian H}
\end{eqnarray}

Starting with DPC, the limiting spectral efficiency in this setting coincides with the corresponding spectral efficiency of the \emph{dual uplink} channel with uniform power distribution \cite{Shamai-Verdu-2001-fading}. This follows from the limiting conclusion in \cite{Viswanath-Tse-Anatharam-IT-2001-Waterfilling}, and by observing that the
optimization problem over diagonal input covariance matrices, that specifies the maximum achievable sum-rate (see \cite{Caire-Shamai-Steinberg-Weingarten-2005} and references therein), is solved by a uniform power distribution \cite{Zaidel-Shamai-Verdu-JSAC-2001-short}. The spectral efficiency of DPC is hence given by \cite{Shamai-Verdu-2001-fading}
\begin{equation}
\label{eq: Spectral efficiency of DPC - Gaussian H}
C^{\dpc}(\snr) = \alpha \log_2\left(1+\snr - \frac{1}{4} \eF(\snr,\alpha)\right)   + \log_2 \left(1 + \alpha \, \snr - \frac{1}{4} \eF(\snr,\alpha)\right) - \frac{\log_2 \e}{4\snr}\eF(\snr,\alpha) \ ,
\end{equation}
where $\eF(\snr,\alpha)$ is defined as \cite{Verdu-Shamai-paper-3-99}
\begin{equation}
\label{eq: Definition of F-x-z}
%\eF(x,z)\triangleq \left( \sqrt{x \, (1+\sqrt{z})^2+1} - \sqrt{x \, (1-\sqrt{z})^2+1}  \right)^2 \quad .
\eF(\snr,\alpha)\triangleq \left( \sqrt{\snr \, (1+\sqrt{\alpha})^2+1} - \sqrt{\snr \, (1-\sqrt{\alpha})^2+1}  \right)^2 \quad .
\end{equation}

%{\cha RALF: Chose a different notation for the function $\F$, as we use that symbol for the free energy}.

Regarding linear ZF, we restrict the discussion to the case in which the active user population can only be controlled through the system load $\alpha$, as is in fact assumed for the nonlinear precoding scheme (see also the discussion in Section \ref{sec: Spectral Efficiency Comparison}). In this setting, as shown, e.g., in \cite{Verdu-Book-98,Muller-IT-2001}, the induced precoding efficiency \eqref{eq: Factorization of the energy penalty with sig-u} (equivalent here to the inverse multiuser efficiency) converges in the large system limit to
\begin{equation}
\label{eq:ralf-1997}
\zeta_\zf = \frac{1}{1-\alpha} \quad ,
\end{equation}
and again for Gaussian input the spectral efficiency coincides with the corresponding result in \cite{Muller-IT-2001} (see also \cite{Caire-Shamai-2003,Verdu-Shamai-paper-3-99,Shamai-Verdu-2001-fading})
\begin{equation}
\label{eq: ZF Spectral Efficiency with Gaussian input}
C^\zf(\snr) =\alpha \log_2 \left(1+(1-\alpha)\, \snr \right) \quad .
\end{equation}
The corresponding spectral efficiency with discrete input alphabet can be derived, e.g., following the guidelines in \cite{Verdu-paper-low-snr-regime-02}. Considering the particular case of binary phase shift keying (BPSK) input, one obtains
\begin{equation}
\label{eq: ZF BPSK Spectral Efficiency}
C^{\zf,\bpsk}(\snr) =
\alpha\left(1 - \int_{-\infty}^\infty \sqrt{\frac{(1-\alpha)\, \snr}{\pi}} \e^{-(1-\alpha)\, \snr \,(s-1)^2} \log_2\left( 1 + \e^{-4 (1-\alpha)\, \snr \, s}\right) \, \rd s \right) \quad .
\end{equation}
The spectral efficiency of linear ZF precoding combined with QPSK input is obtained via the relation \cite{Verdu-paper-low-snr-regime-02}
\begin{equation}
\label{eq: ZF QPSK Spectral Efficiency}
C^{\zf,\qpsk}(\snr) = 2 C^{\zf,\bpsk}\left(\frac{\snr}{2}\right) \quad,
\end{equation}
yielding (through \eqref{eq: Relation between SNR and EbNo}) $C^{\zf,\qpsk}(\febno) = 2 C^{\zf,\bpsk}\left(\febno\right)$. The spectral efficiency of GTHP for the corresponding setting is derived in Appendix \ref{App: Spectral Efficiency of Generalized Tomlinson-Harashima Precoding}.

%%%%%%%%%%%%%%%%%%%%

\section{Lattice Precoding: An RSB Example}\label{sec: Replica Symmetry Breaking Example}

Adhering to \cite{Muller-Guo-Moustakas-JSAC-2008}, we consider in the following a particular example of a discrete relaxed alphabet set for QPSK signaling, which exhibits replica symmetry breaking. The original QPSK constellation alphabet is represented by the set
\begin{equation}
\U=\set{1+j,-1+j,-1-j,1-j} \quad ,
\end{equation}
and quadrature symmetric transmissions are assumed (note that the above definition induces $\sigma_u^2=2$). The relaxed alphabets in this particular example can be represented as points from the extended lattice
\begin{equation}
\B_u = \frac{u}{1+j} ((4\mathbb{Z}+1)\times(4\mathbb{Z}+1)) \quad , \quad \forall u \in \U \quad .
\end{equation}

\begin{figure}[!t]\begin{center}
\includegraphics[scale=0.5]{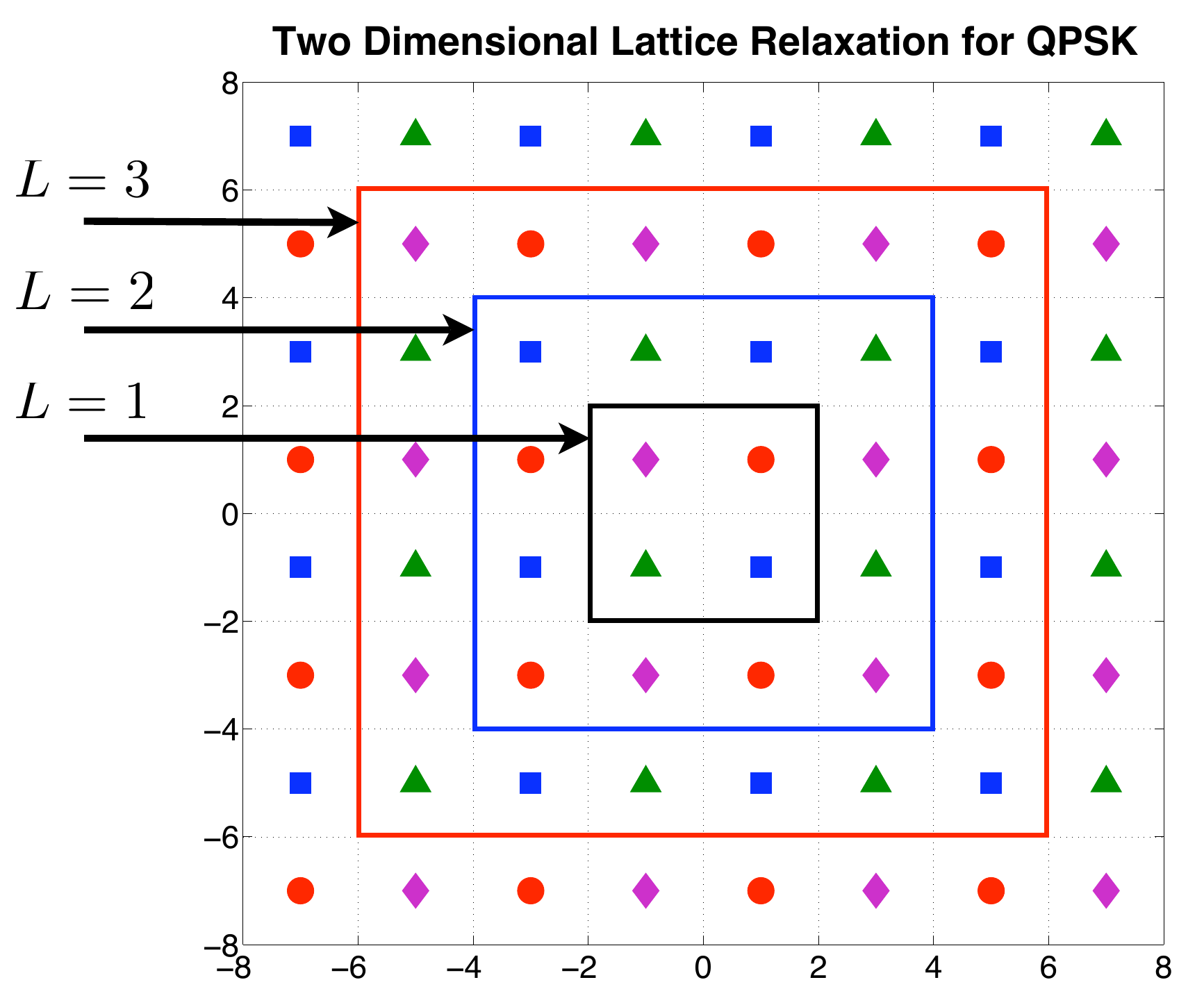}
\end{center}
\caption{Two dimensional lattice based relaxation for QPSK input.} \label{fig: Two dimensional lattice based relaxation for QPSK input}
\end{figure}

More specifically,  we take
\begin{equation}
\B_{\pm1 \pm j}=\pm\set{c_1,c_2,\dots,c_L} \pm j\set{c_1,c_2,\dots,c_L} \quad ,
\end{equation}
where it is assumed that $-\infty =c_0 < c_1 < \cdots < c_L <  c_{L+1} = \infty$. The parameter $L$ thus specifies the number of lattice points used in the extended alphabet in each dimension, and we particularize here to the set $\set{+1,-3,+5,-7,+9,\dots}$. The alphabet relaxation scheme is depicted in Figure \ref{fig: Two dimensional lattice based relaxation for QPSK input}.
Due to the complete quadrature symmetry of this setting, all QPSK constellation points and their corresponding relaxed alphabet subsets are completely equivalent, and we focus in the following, for notational convenience, on the QPSK constellation point represented by $u= 1+j$, and $\B_{1 + j}$.

The first step in the analysis is to rewrite \eqref{eq: 1RSB: Equation after partial derivative wrt f - beta inf h zero}--\eqref{eq: 1RSB: Equation after partial derivative wrt mu - beta inf and h zero} and obtain the four macroscopic parameters $\set{q_1, p_1,\chi_1,\mu_1}$ for the current example.
Denoting the real and imaginary parts of an arbitrary point $s\in \mathbb{C}$ as  $s_\re \triangleq \Re \set{s}$ and $s_\im \triangleq \Im\set{s}$, the Voronoi region of the lattice point $x=c_m + j c_n$ is the region in the complex plane for which %$s_\re  \in  (v_m,v_{m+1})$ and $s_\im \in  (v_n,v_{n+1})$,
\begin{equation}
s_\re  \in  (v_m,v_{m+1}) \quad , \quad
s_\im \in  (v_n,v_{n+1}) \quad ,
\end{equation}
where the boundaries of the Voronoi regions are $\set{v_i= \frac{c_i+c_{i-1}}{2}}$.
%\begin{equation}
%v_i= \frac{c_i+c_{i-1}}{2} \quad .
%\end{equation}
Now considering \eqref{eq: 1RSB: Definition of the function sF}, recall that for any given $y,z \in \mathbb{C}$ the lattice point that maximizes the exponent therein is given by $\underset{x \in \B_{u}}{\argmin} \abs{f_1 z + g_1y - \varepsilon_1  x}$. This implies that a lattice point $x=c_m + j c_n$ is the solution to the above minimization problem whenever
\begin{equation}\label{eq: 1RSB QPSK: Definition of the integration region for the 2D discrete lattice}
y_\re \in \left(\psi_m(z_\re) ,\psi_{m+1}(z_\re)  \right) \ , \quad
y_\im \in  \left(\psi_n(z_\im) ,\psi_{n+1}(z_\im)  \right) \ ,
\end{equation}
where we introduced the real argument function
\begin{equation}\label{eq: 1RSB QPSK: Definition of psi-xi}
\psi_k(\xi) \triangleq \frac{\varepsilon_1 v_k - f_1 \xi}{g_1} \quad .
\end{equation}

Applying this observation to \eqref{eq: 1RSB: Equation after partial derivative wrt f - beta inf h zero}--\eqref{eq: 1RSB: Equation after partial derivative wrt mu - beta inf and h zero}, and exploiting the quadrature symmetry property, the derivation simplifies considerably by noticing that the inner integrals therein can be represented as sums of separate integrals over the regions specified by \eqref{eq: 1RSB QPSK: Definition of the integration region for the 2D discrete lattice}. Accordingly, consider the two real argument functions
\begin{equation}
\Theta_k(\xi)
\triangleq \e^{\mu_1 c_k \left[ (\mu_1 g_1^2 - \varepsilon_1) c_k+ 2  f_1 \xi\right]} \biggl[ Q\left( \sqrt{2} (\psi_k(\xi) - \mu_1  g_1 c_k) \right)  - Q\left( \sqrt{2} (\psi_{k+1}(\xi) - \mu_1  g_1 c_k) \right)\biggr]  \  , \label{eq: 1RSB QPSK: Definition of Theta-xi}
\end{equation}
\begin{equation}
\Psi_k(\xi) \triangleq \frac{1}{2 \sqrt{\pi}} \e^{\mu_1 c_k \left[ (\mu_1 g_1^2 - \varepsilon_1) c_k+ 2  f_1 \xi\right]}
\biggl[   \e^{-( \psi_k(\xi) - \mu_1  g_1 c_k)  )^2}  - \e^{-( \psi_{k+1}(\xi) - \mu_1  g_1 c_k)  )^2}  \biggr]  \label{eq: 1RSB QPSK: Definition of Psi-xi} \ .
\end{equation}
Then, following some tedious algebra, it can be shown from \eqref{eq: 1RSB: Equation after partial derivative wrt f - beta inf h zero}--\eqref{eq: 1RSB: Equation after partial derivative wrt mu - beta inf and h zero} that the parameters $\set{q_1, p_1, \chi_1, \mu_1}$ are the solutions to the coupled equations
\begin{eqnarray}
q_1 &=& 2   \int_{-\infty}^\infty \frac{ \sum_{m=1}^L c_m^2  \Theta_m(\xi) }{\sum_{k=1}^L \Theta_k(\xi) }  \e^{-\xi^2} \, \frac{\rd \xi}{\sqrt{\pi}} - p_1 \quad , \label{eq: 1RSB QPSK: Fixed point equation for q} \\
 p_1 &=& \frac{2}{f_1\mu_1}    \int_{-\infty}^\infty \frac{\sum_{m=1}^L c_m \Theta_m(\xi)}{\sum_{k=1}^L \Theta_k(\xi)}
\xi \e^{-\xi^2} \, \frac{\rd \xi}{\sqrt{\pi}} - \frac{\chi_1}{\mu_1}  \ , \label{eq: 1RSB QPSK: Fixed point equation for p} \\
\chi_1  &=&
\frac{2}{g_1}  \int_{-\infty}^\infty \frac{\sum_{m=1}^L c_m \Psi_m(\xi)  }{\sum_{k=1}^L \Theta_k(\xi)}   \e^{-\xi^2} \, \frac{\rd \xi }{\sqrt{\pi}} \ , \label{eq: 1RSB QPSK: Fixed point equation for chi}
\end{eqnarray}
and
\begin{multline}
\mu_1   =  \left[2 q_1(\chi_1 + \mu_1 p_1)R'(-\chi_1 - \mu_1p_1)\right]^{-1}  \cdot \biggl[2 \int_{-\infty}^\infty \log \left(\sum_{m=1}^L \Theta_m(\xi) \right) \e^{-\xi^2} \, \frac{\rd \xi}{\sqrt{\pi}} \\
-\int_{\chi_1}^{\chi_1+\mu_1 p_1} R(-w)\, \rd w - 2\mu_1 \chi_1 g_1^2  + \mu_1 ( q_1  +2 p_1) R(-\chi_1 -\mu_1p_1)    \biggr]     \ . \label{eq: 1RSB QPSK: Fixed point equation for mu}
\end{multline}

The corresponding energy penalty is obtained by plugging the four solutions into \eqref{eq: 1RSB: Limit of the effective energy penalty}. Applying the same approach to Proposition \ref{prop: 1RSB: Marginal limiting conditional distribution of x given u}, the limiting conditional probability of the precoder output being $x=c_m+j c_n \in \B_{1+j}$ is given by
\begin{equation}\label{eq: 1RSB QPSK: Conditional distribution for c-m j c-n}
\begin{aligned}
\Pr\set{x= c_m + j c_n \in \B_u | u=1+j}  &= \int_{-\infty}^\infty \frac{\Theta_m(\xi)}{ \sum_{k=1}^L \Theta_k(\xi)} \, \e^{-\xi^2} \, \frac{\rd \xi}{\sqrt{\pi}}  \cdot
\int_{-\infty}^\infty \frac{\Theta_n(\zeta)}{\sum_{\ell=1}^L \Theta_\ell(\zeta) } \, \e^{-\zeta^2} \, \frac{\rd \zeta}{\sqrt{\pi}} \\
&= \Pr\set{\Re\set{x}= c_m | u=1+j} \cdot \Pr\set{\Im\set{x} = c_n | u=1+j} \ .\end{aligned}
\end{equation}
The limiting conditional probabilities that correspond to the rest of the QPSK constellation points are readily obtained from \eqref{eq: 1RSB QPSK: Conditional distribution for c-m j c-n} by symmetry considerations.
Note also that \eqref{eq: 1RSB QPSK: Conditional distribution for c-m j c-n} implies that the real part and the imaginary part of the precoder output $x$ behave as independent random variables.

Numerical results for the limiting energy penalty of the discrete lattice relaxation scheme are plotted in Figure \ref{fig: Energy penalty comparison}. 
%\begin{figure}[tb]
%\begin{center}%\centering
%\includegraphics[scale=0.4]{complexfigures/Journal_Paper_Discrete_Lattice_Energy_Plot_1RSB_2}
%\caption{\small The energy penalty per symbol, as a function of the system load $\alpha$, for the two dimensional discrete alphabet relaxation scheme for QPSK input.
%\label{fig: Energy penalty comparison}}
%\end{center}
%\end{figure}
\begin{figure}[tb]
\begin{center}%\centering
\includegraphics[scale=0.45]{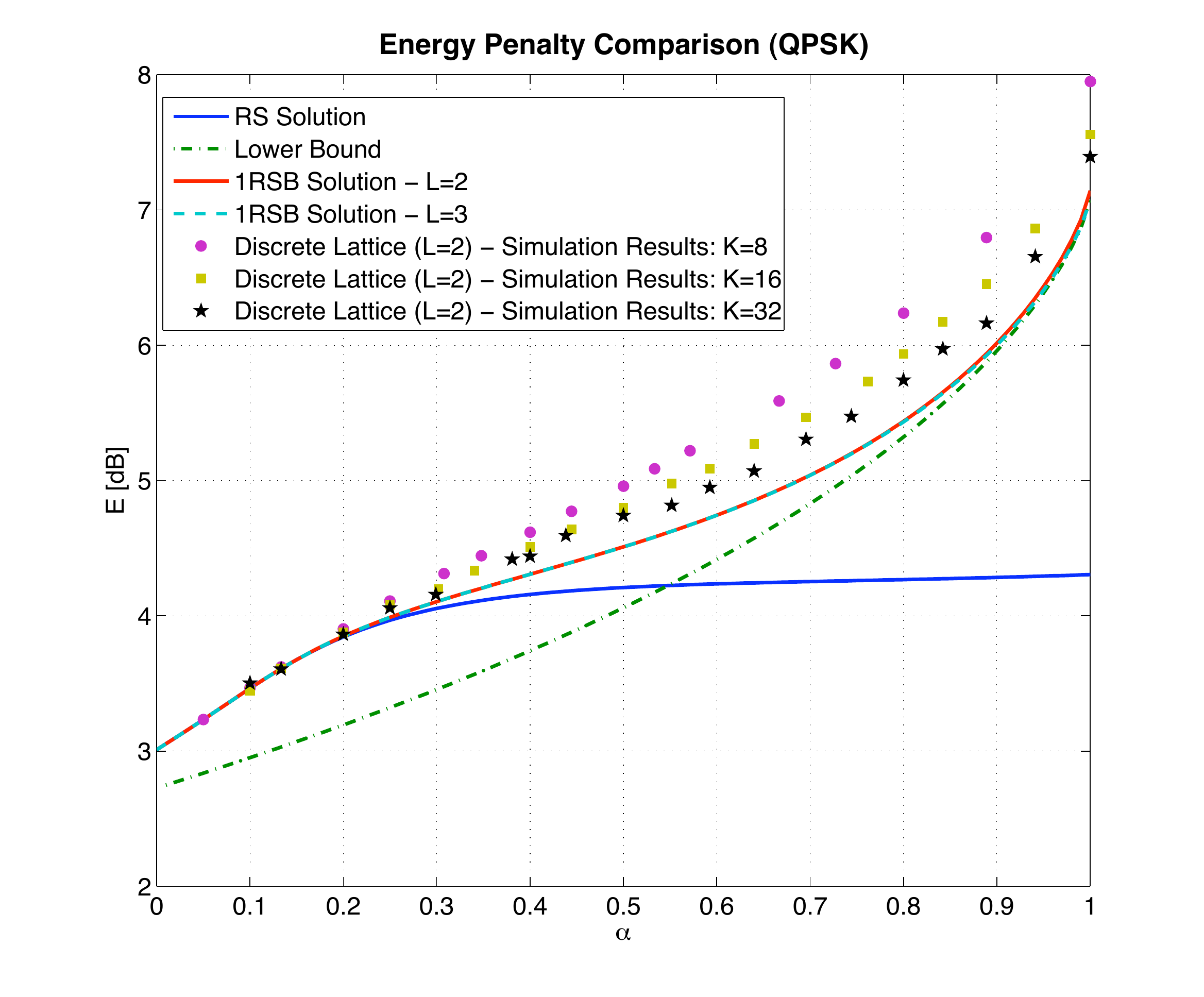}
\caption{\small The energy penalty per symbol, as a function of the system load $\alpha$, for the two dimensional discrete alphabet relaxation scheme for QPSK input.
\label{fig: Energy penalty comparison}}
\end{center}
\end{figure}
The figure shows the limiting energy penalty (in dB) as a function of the system load $\alpha$, for the particular case of a Gaussian $\Mat{H}$ and a ZF front-end.
Since $\sigma_u^2 = 2$, the corresponding precoding efficiency \eqref{eq: Factorization of the energy penalty with sig-u} can be immediately obtained by subtracting $3\rm dB$ from the energy penalty shown in the figure. The results in Figure \ref{fig: Energy penalty comparison} correspond to alphabet relaxations with $L=2$ and $L=3$. Note that the two curves are essentially indistinguishable and the energy penalty with $L=2$ becomes only negligibly larger as $\alpha$ gets close to unity. 
This implies that increasing $L$ beyond $2$ in this setting provides diminishing returns. 
Empirical energy penalties obtained through Monte Carlo simulations are  also included in the figure. 
The results are for systems in which the \emph{number of users} is fixed to $K=8$, $K=16$, and $K=32$ (averaged over $10^4$, $10^3$, and $10^2$ channel realizations, respectively).
 The energy penalty is shown to decrease with the system size, and the simulation results exhibit a good match to the limiting energy penalty predicted by the 1RSB replica analysis. The lower bound for the limiting energy penalty obtained in \cite{Ryan-Collings-Clarkson-Heath-TRANS-COMM-2009} is also plotted in this figure which, with appropriate scaling to match the current setting, is given by
\begin{equation}\label{eq: Ryan et al lower bound for the limiting energy}
\bE_{LB}  = \frac{16}{\pi} (1-\alpha)^{\frac{1}{\alpha}-1} \quad .
\end{equation}
Figure \ref{fig: Energy penalty comparison} shows that the 1RSB prediction approaches the lower bound as the load approaches unity. Note however that the 1RSB result stays strictly higher than the lower bound. In fact, a careful numerical examination of the limiting 1RSB energy penalty at $\alpha=1$ shows that it hits the value of $7.0744$ dB for $L \ge4$, while the lower bound in this case is $\frac{16}{\pi}\approx 7.0697$ dB. The numerical analysis of the limiting energy penalty is considerably simplified in this region of $\alpha $ by the (numerical) observation that the macroscopic parameter $\chi_1$ approaches $0$ as $\alpha \rightarrow 1$ (although it stays strictly positive). The small $\chi_1$ approximation of the equations employed to calculate the limiting 1RSB energy penalty is shortly discussed in Appendix \ref {app: chi-zero approximation of the 1RSB equations in the vicinity of unit load}.
The RSB phenomena is demonstrated by considering the limiting energy penalty obtained via the RS approximation, as stated by Proposition \ref{prop: RS: MGM Limiting Energy Penalty} (the explicit expression for the current example is given in \cite[Eq.\ (26)]{Muller-Guo-Moustakas-JSAC-2008}). As shown in Figure \ref{fig: Energy penalty comparison}, the RS approximation fails to predict the limiting energy penalty for $\alpha > 0.3$, and in fact it even violates the lower bound \eqref{eq: Ryan et al lower bound for the limiting energy} for $\alpha > 0.55$.

The better accuracy of 1RSB is also visible looking at the zero-temperature entropy.
We can analytically evaluate Proposition \ref{Cor: RS entropy result} in the case of a Gaussian $\Mat{H}$, which becomes 
\begin{equation}\label{eq: Explicit entropy expression - Gaussian H}
{\bS}= \frac{1-\alpha - \sqrt{(1-\alpha)^2 + 4 \alpha \chi}}{2 \alpha} +
   \frac{1-\alpha}{\alpha}  \log \left(\frac{1-\alpha +\sqrt{(1-\alpha)^2 + 4 \alpha \chi}}{2(1-\alpha)}\right) \quad .
\end{equation}
The entropy for both the RS and 1RSB approximations for a relaxation level $L=2$ are shown in Figure \ref{fig: Zero temperature entropy}.
Although the 1RSB solution of the above model also has negative zero-temperature entropy, it is much closer to zero, corresponding to a much weaker instability, and approaches zero as $\alpha \rightarrow 1$. In contrast, the RS entropy drifts away from zero as $\alpha\rightarrow 1$.

\begin{figure}[t]
\begin{center}
%\begin{minipage}{7cm}
%\begin{center}
%\includegraphics[scale=0.225]{Journal_Zero_Temperature_Entropy_RS}
%\caption{Zero temperature entropy of the RS solution as a function of the system load $\alpha$. \label{fig: Zero temperature entropy - RS}}
%\end{center}
%\end{minipage}
%\hfill
%\begin{minipage}{7cm}
%\begin{center}
%\includegraphics[scale=0.235]{Journal_Zero_Temperature_Entropy_1RSB}
%\caption{Zero temperature entropy of the 1RSB solution as a function of the system load $\alpha$. \label{fig: Zero temperature entropy - 1RSB}}
%\end{center}
%\end{minipage}
%\centerline{\includegraphics[scale=0.5]{plot_entropy_corr}}
\centerline{\includegraphics[scale=0.5]{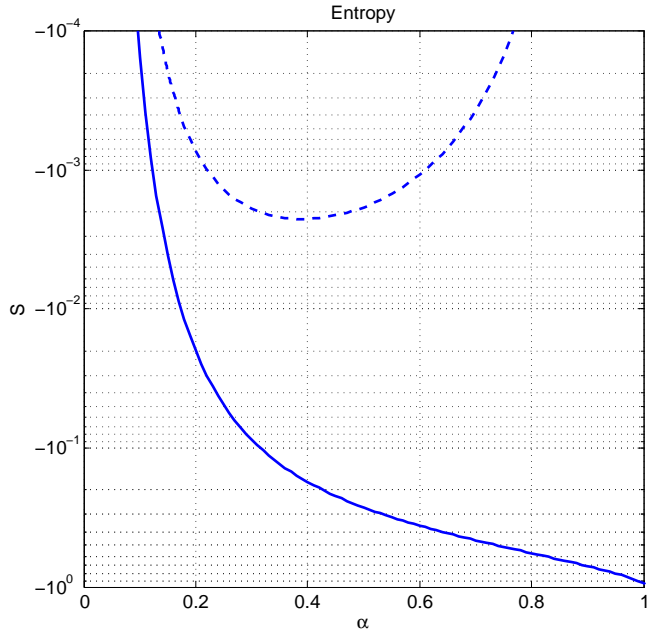}}
\caption{Zero temperature entropy of the RS solution (solid), and 1RSB solution (dashed), as a function of the system load $\alpha$ (corresponding to \eqref{eq: Explicit entropy expression - Gaussian H}).
\label{fig: Zero temperature entropy}}
\end{center}
\end{figure}

The limiting conditional probabilities of \eqref{eq: 1RSB QPSK: Conditional distribution for c-m j c-n} are plotted in Figure  \ref{fig: Discrete Lattice - Conditional probabilities},
\begin{figure}[tb]
\begin{center}
\includegraphics[scale=0.45]{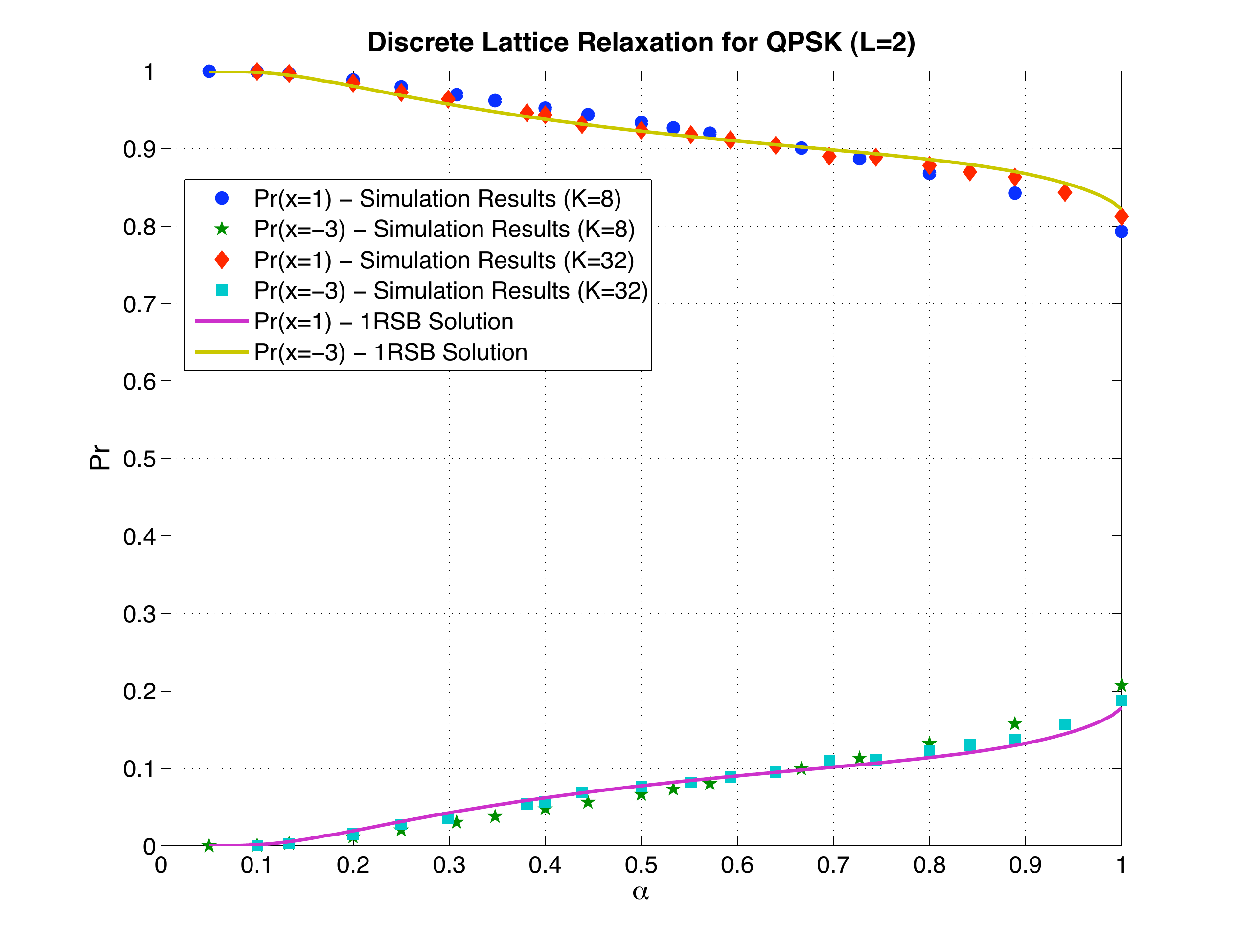}
\end{center}
\caption{Conditional probabilities of the real part of the precoder output for the two-dimensional discrete lattice relaxation scheme, given that $\Re\set{u}=1$. \label{fig: Discrete Lattice - Conditional probabilities}}
\end{figure}
as well as the empirical conditional probabilities, based on the Monte Carlo simulations employed to produce the energy penalties of Figure  \ref{fig: Energy penalty comparison}. The results correspond to a relaxation level of $L=2$, and focus on the \emph{real part} of the extended alphabet points, given that the real part of the original QPSK constellation point satisfies $\Re\set{u}=1$ (recall the decoupling of the real and imaginary parts implied by \eqref{eq: 1RSB QPSK: Conditional distribution for c-m j c-n}). The simulation results exhibit again a good match to the limiting analytical 1RSB prediction. It is also clearly demonstrated that, when the system load $\alpha$ is low, hardly any relaxation is required, while the probability of using symbols from the extended alphabet set increases as $\alpha$ approaches unity.

%%%%%%%%%%%%%%%%%%%%%%%%
\section{Convex Precoding: An RS Example}\label{sec: A Replica Symmetric Example}

This section is devoted to another alphabet relaxation scheme, also introduced in
\cite{Muller-Guo-Moustakas-JSAC-2008} for QPSK signaling. The key feature of this relaxation scheme is that the extended alphabet set is continuous and \emph{convex}, allowing for an efficient solution to the corresponding quadratic programming problem of minimizing the energy penalty.
Convex optimization problems are generally believed not to exhibit replica symmetry breaking \cite{mezard}. In certain special cases this has been shown explicitly \cite{moustakas-2007}.
Furthermore, as will be demonstrated in the sequel,
the \emph{replica symmetric} solution for this alphabet relaxation scheme agrees well with numerical simulations and thus
considerably simplifies the analysis of the limiting regime.

Denoting
\begin{equation}
\label{eq: CR-QPSK - Formal definition}
\B_{1+j} = \set{z \in \mathbb{C} : \Re\set{z} \ge 1, \Im\set{z} \ge 1} \quad ,
\end{equation}
the relaxed alphabet subsets are defined by
\begin{equation}
\B_u = \frac{u}{1+j} \B_{1+j} \quad , \quad u \in \set{1+j,-1+j,-1-j,1-j} \quad.
\end{equation}
The alphabet relaxation scheme is depicted in Figure \ref{fig: Illustration of CR-QPSK},
%These four alphabets are independent (assuming independent input symbols), and \emph{convex}.
and it is referred to henceforth as \emph{convex relaxation for QPSK (CR-QPSK)}.

\begin{figure}[tb]
\begin{center}
\includegraphics[scale=0.3]{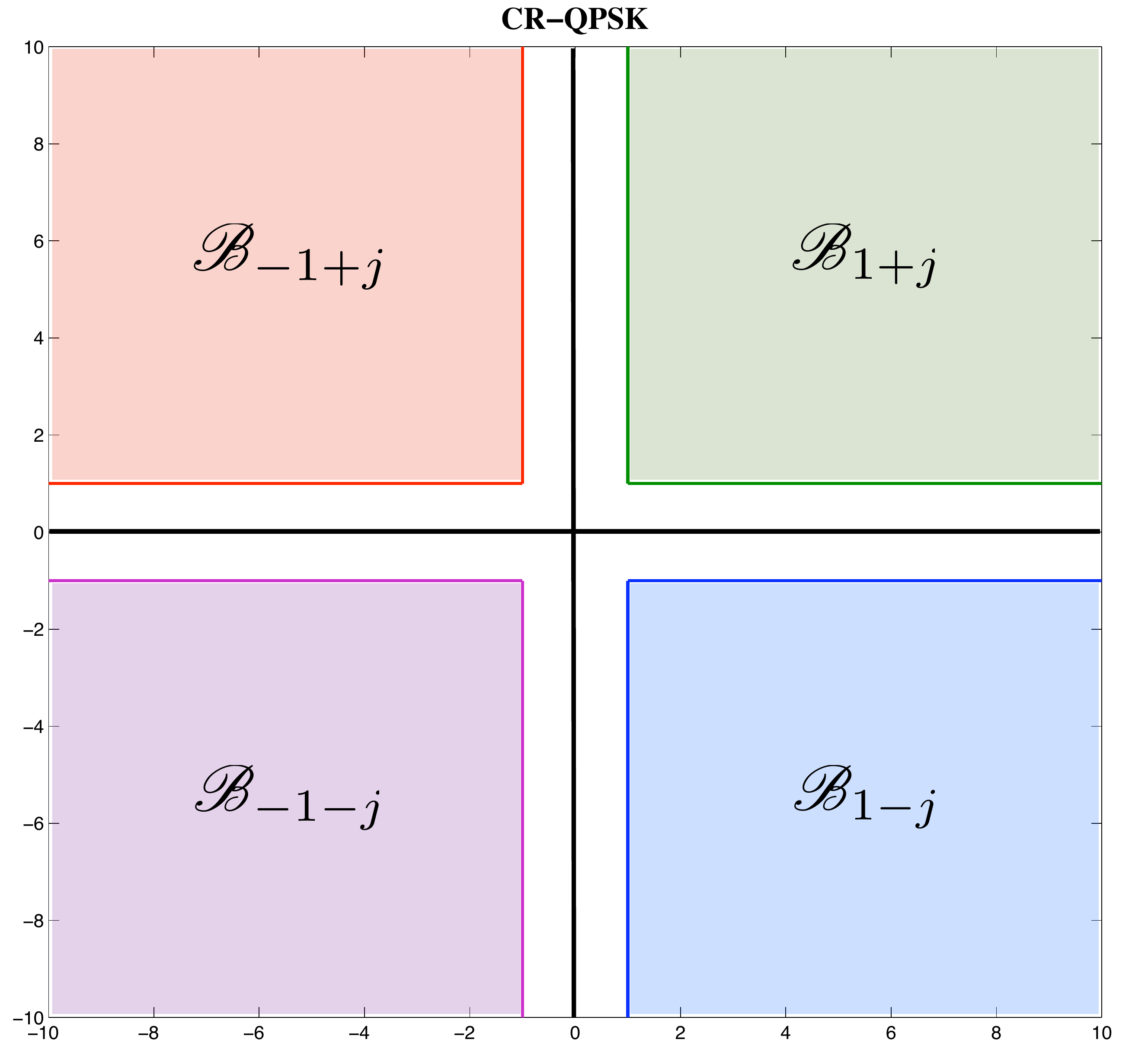}
\caption{Extended alphabet sets for the convex relaxation precoding scheme for QPSK signaling.}
\label{fig: Illustration of CR-QPSK}
\end{center}
\end{figure}

The RS approximation for the limiting energy penalty with the CR-QPSK relaxation scheme is obtained through Proposition \ref{prop: RS: MGM Limiting Energy Penalty}, and it is given by the solution to the following fixed point equation \cite[Eq.\ (30)]{Muller-Guo-Moustakas-JSAC-2008}
\begin{equation}
\label{eq: CR-QPSK energy penalty equation}
Q\left(\sqrt{\frac{2}{\alpha\bE}}\right) = \frac{2 +(\alpha-1)\bE +\sqrt{\frac{\alpha\bE}{\pi}} \e^{-\frac{1}{\alpha\bE}} }{2+\alpha\bE} \quad.
\end{equation}
Note that \eqref{eq: CR-QPSK energy penalty equation} yields finite energy penalties for all loads $0\le\alpha<2$. Although loads greater than unity imply that the matrix $\vct{HH}^\dagger$ in \eqref{eq: Definition of the zero-forcing precoder} is singular, this does not lead to interference at the receivers in the large system limit, as shown rigorously in \cite{miguel-TWC-10}. 

Numerical results for the limiting energy penalty of CR-QPSK are plotted in Figure \ref{fig: Energy penalty of CR-QPSK}. 
\begin{figure}[tb]
\begin{center}
\includegraphics[scale=0.45]{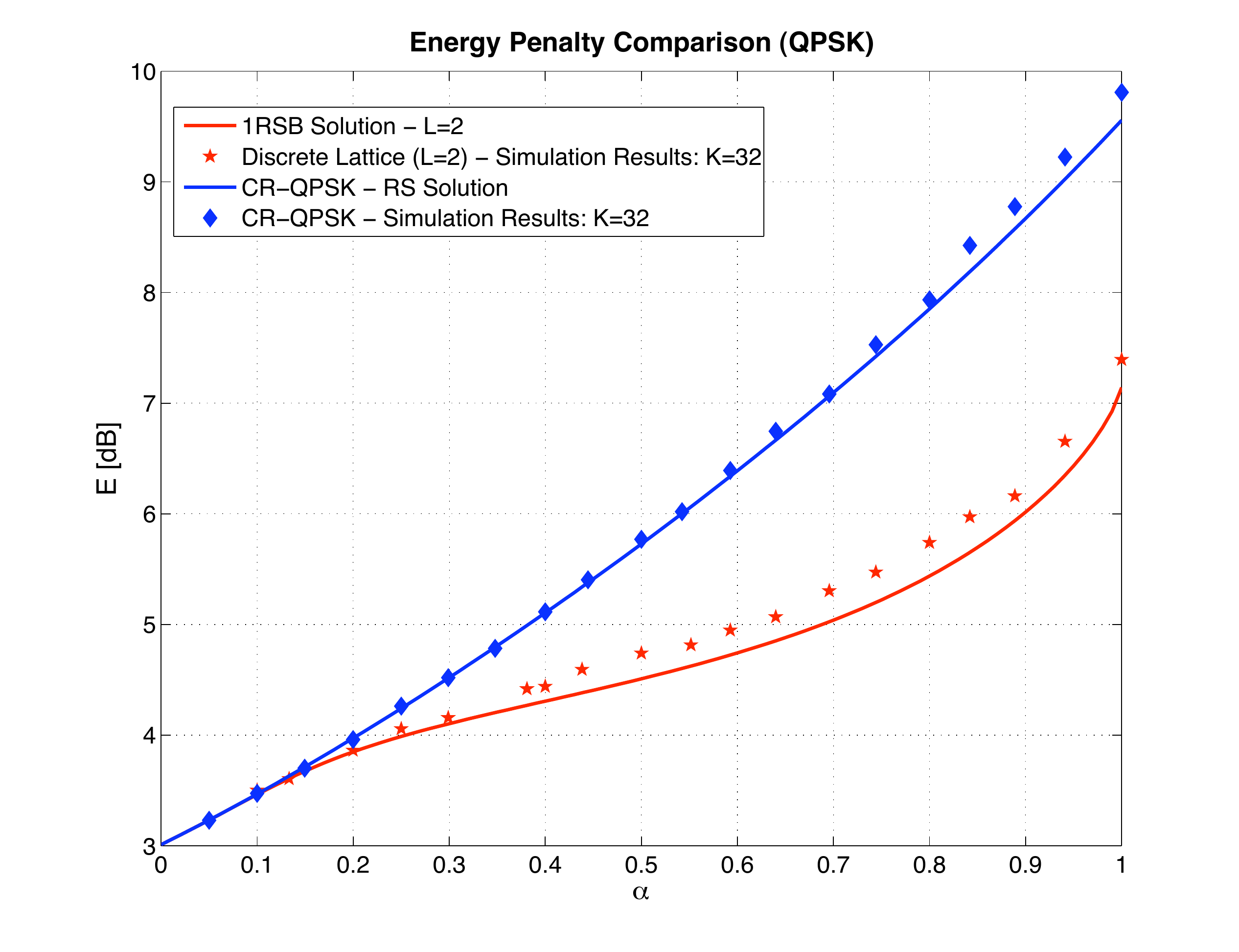}
\caption{The energy penalty per symbol as a function of the system load $\alpha$, for the CR-QPSK alphabet relaxation scheme. Corresponding results for the discrete alphabet relaxation scheme of Section \ref{sec: Replica Symmetry Breaking Example} are provided for comparison.}
\label{fig: Energy penalty of CR-QPSK}
\end{center}
\end{figure}
Empirical results based on Monte Carlo simulations are provided as well. These results were obtained by fixing the number of users to $K=32$, and averaging over 1000 channel realizations. The results exhibit an excellent match to the limiting RS analytical results, thus supporting the validity of the RS approximation. The corresponding results for the discrete lattice-based alphabet relaxation scheme of Section \ref{sec: Replica Symmetry Breaking Example} are also provided for the sake of comparison, and it is clearly observed that in terms of the limiting energy penalty, the discrete scheme is superior to the CR-QPSK scheme for all $\alpha\in (0,1]$. The limiting energy penalty difference approaches its maximum value of $2.41\,$dB at $\alpha=1$.  As will be shown in Section \ref{sec: Spectral Efficiency Comparison}, however, the comparison becomes more subtle when spectral efficiency is investigated.

\begin{figure}[tb]
\begin{center}
\includegraphics[scale=0.45]{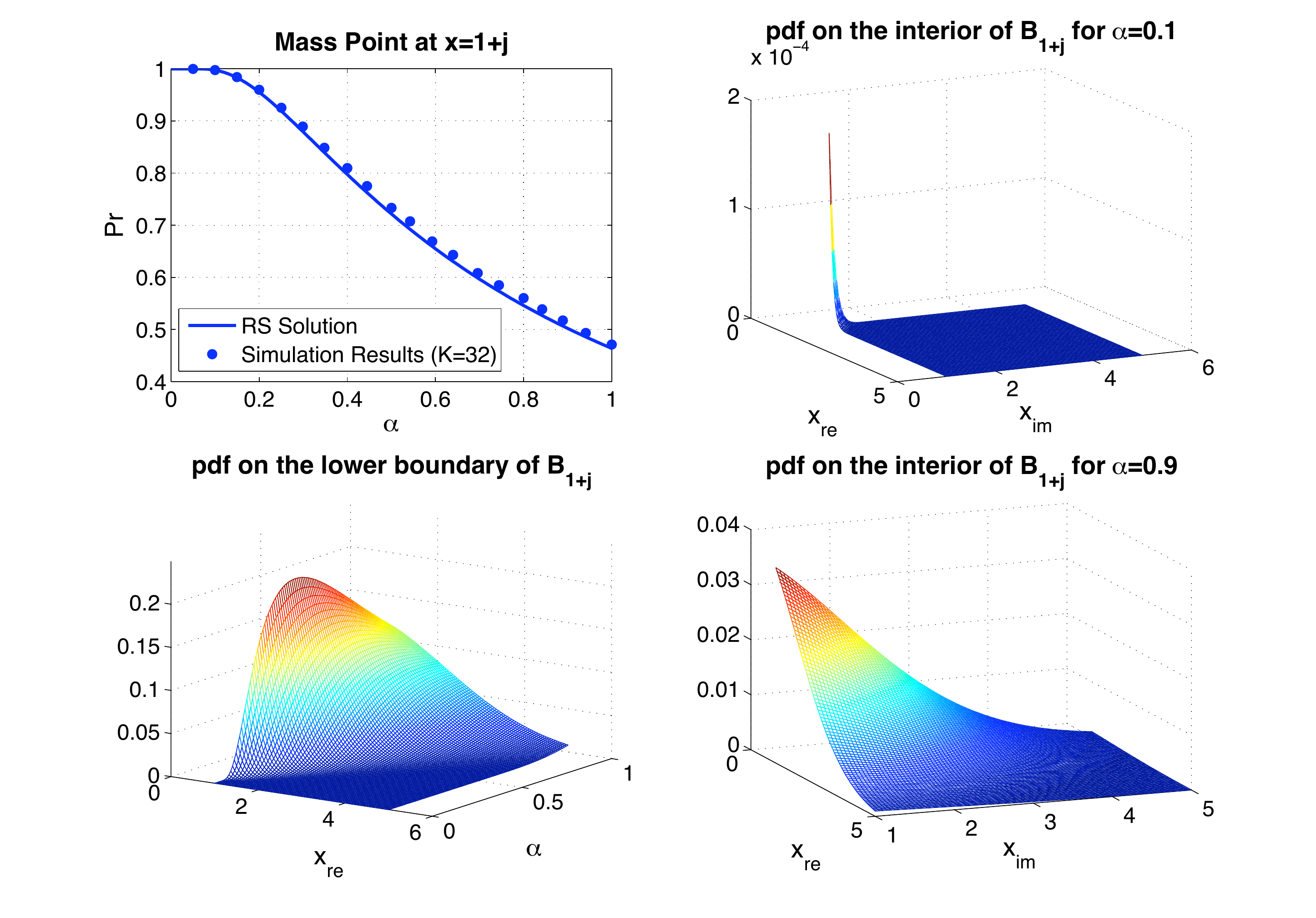}
\caption{The limiting conditional distribution of the precoder output for the CR-QPSK scheme, given that $u=1+j$.}
\label{fig: Conditional densities for CR-QPSK}
\end{center}
\end{figure}

The RS approximation of the limiting conditional distribution of the precoder outputs is obtained using Proposition \ref{prop: RS Marginal limiting conditional distribution of x given u}. The idea here is to start from a discretized version of the continuous CR-QPSK relaxed alphabet set, and obtain the limiting conditional distribution of each relaxed alphabet point using \eqref{eq: RS: Equation for the limiting conditional distribution - Infinite beta}. The final step is then to take the limit as the areas of the Voronoi cells corresponding to each such point vanish. Using this approach, while restricting the discussion to a Gaussian $\Mat{H}$ and focusing for convenience on the QPSK constellation point $u=1+j$, one gets the corresponding conditional probability \emph{density} function (pdf) 
\begin{equation}
\label{eq: CR-QPSK conditional distribution for 1+j}
\begin{aligned}
%f^{\crqpsk}_{x|u}(x|u=1+j)
f^{\crqpsk}_{X|U}(x|u=1+j)
&=Q_1^2  \delta(x_\re-1)  \delta(x_\im-1)
\\ &\phantom{=} + Q_1 \frac{1}{\sqrt{\pi\alpha\bE^{\crqpsk}}} \e^{-\frac{x_\im^2}{\alpha\bE^{\crqpsk}}} \delta(x_\re-1) \mathcal{U}(x_\im-1) \\
&\phantom{=}+ Q_1 \frac{1}{\sqrt{\pi\alpha\bE^{\crqpsk}}} \e^{-\frac{x_\re^2}{\alpha\bE^{\crqpsk}}} \delta(x_\im-1) \mathcal{U}(x_\re-1)  \\
&\phantom{=}+ \frac{1}{\pi\alpha\bE^{\crqpsk}} \e^{-\frac{\abs{x}^2}{\alpha\bE^{\crqpsk}}} \mathcal{U}(x_\re-1) \mathcal{U}(x_\im-1) \ , \quad x_\re, x_\im \in \mathbb{R} \quad ,
\end{aligned}
\end{equation}
%where $\Im\set{\cdot}$ takes the imaginary part of the argument,
where we decompose the complex argument as $x\triangleq x_\re + j x_\im$, $\mathcal{U}(x)$ denotes the unit step function, $\bE^{\crqpsk}$ denotes the limiting energy penalty of the CR-QPSK scheme obtained from \eqref{eq: CR-QPSK energy penalty equation}, and the constant $Q_1$ is defined as
\begin{equation}
\label{eq: Definition of Q_1}
Q_1 \triangleq Q\left(-\sqrt{\frac{2}{\alpha\bE^{\crqpsk}}}\right) \quad .
\end{equation}
The conditional pdf given the rest of the QPSK constellation points (i.e.,  $u \in \{-1+j,$ $-1-j,1-j\}$) is obtained in an analogous manner, while considering the full symmetry of the extended constellation.

Returning to \eqref{eq: CR-QPSK conditional distribution for 1+j}, note that this pdf contains masses on the boundaries of $\B_{1+j}$, and in particular a mass point at the original QPSK constellation point (i.e., $x=1+j$). Plots that demonstrate this behavior of the pdf as a function of $\alpha$ are provided in Figure \ref{fig: Conditional densities for CR-QPSK}. The upper left plot shows the weight of the mass point at $x=1+j$, as a function of $\alpha$ (corresponding to $Q_1^2$). The lower left plot shows the pdf mass on the lower boundary of the extended alphabet subset $\B_{1+j}$ (i.e., when the imaginary part of the precoder's output is fixed to $x_\im=j$). The plots on the right show the pdf on the interior of $\B_{1+j}$, for $\alpha=0.1$ (upper right) and $\alpha=0.9$ (lower right). The increase in probability of using extended alphabet points as the system load increases, is clearly demonstrated in the figure. 

Additional numerical results comparing the analytical RS approximation for the pdf to empirical simulation results are shown in the upper left plot of Figure \ref{fig: Conditional densities for CR-QPSK} and in Figure \ref{fig: CR-QPSK: CDF comparison - empirical vs RS - Prob x_re lt x}. 
%\begin{figure}[tb]
%\begin{center}
%\includegraphics[scale=0.35]{complexfigures/Journal_CR_Probability_Comparison_Mass_at_1}
%\caption{The probability that at the real part of the precoder's output $x_\re=1$ for the CR-QPSK scheme, given that $u=1+j$.}
%\label{fig: CR-QPSK: CDF comparison - empirical vs RS - Mass point at 1}
%\end{center}
%\end{figure}
\begin{figure}[tb]
\begin{center}
\includegraphics[scale=0.45]{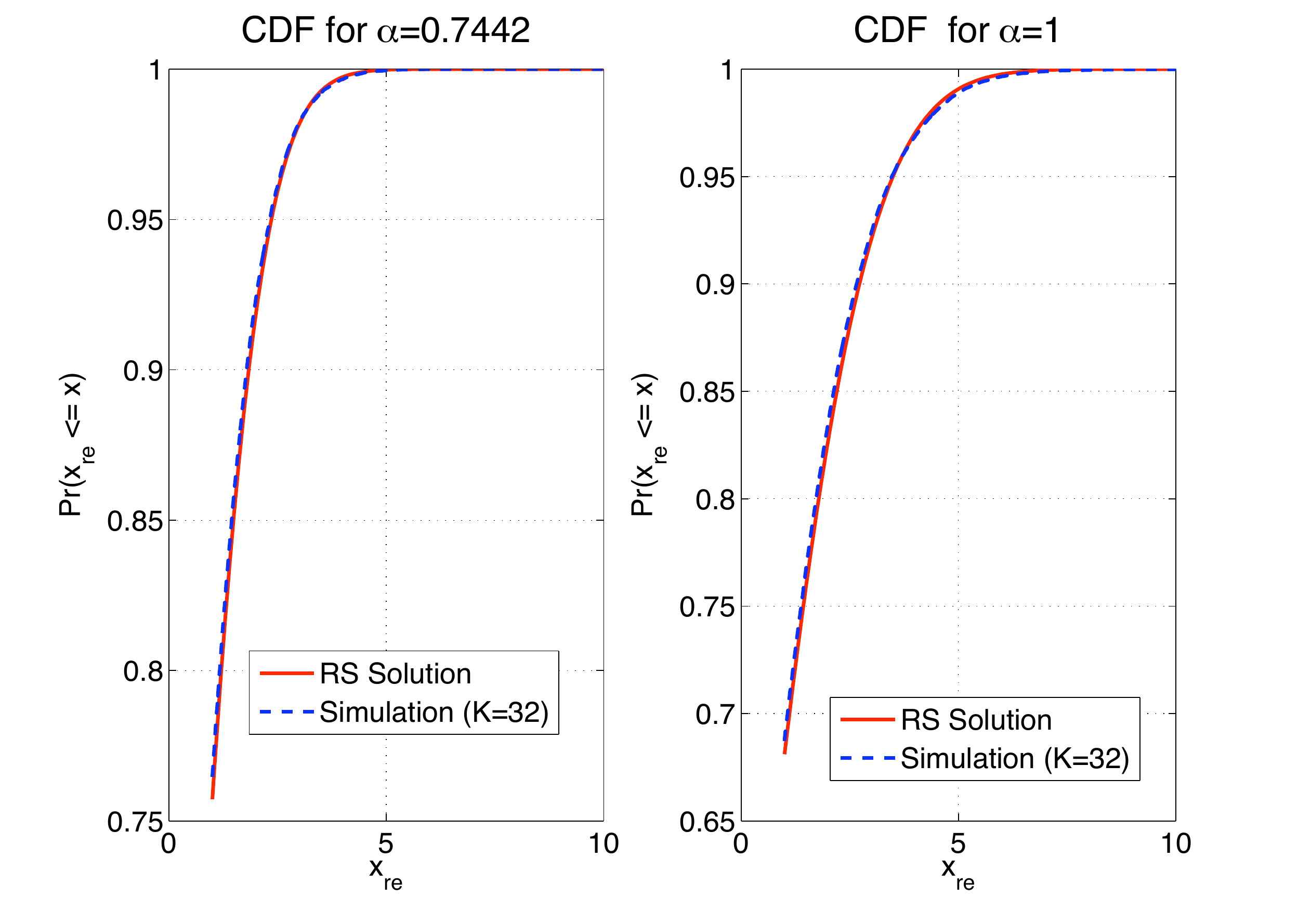}
\caption{CDF of the \emph{real} part of the precoder's output for the CR-QPSK scheme, given that $\Re\set{u}=1$.}
\label{fig: CR-QPSK: CDF comparison - empirical vs RS - Prob x_re lt x}
\end{center}
\end{figure}
The upper left plot of Figure \ref{fig: Conditional densities for CR-QPSK} compares the probability mass at $x=1+j$ to corresponding simulation results for $K=32$ (averaged over 1000 channel realizations). The corresponding comparison for the cumulative distribution function (CDF) of $x_\re$, given that $\Re\set{u}=1$, is shown in Figure \ref{fig: CR-QPSK: CDF comparison - empirical vs RS - Prob x_re lt x}. The left plot shows the CDF for $\alpha=0.7442$ (i.e., for $N=43$), while the right plot shows the results for unit load. As observed, all empirical results exhibit a very good match to the analytical RS approximation, further supporting the validity of the RS analysis for the CR-QPSK scheme.

%%%%%%%%%%%%%%%%%%%%%%%
\section{Spectral Efficiency Comparison}\label{sec: Spectral Efficiency Comparison}
\label{speceff}

The two previous sections focused on the \emph{transmitting} end of the system, and investigated the limiting behavior of the precoder output while employing two particular alphabet relaxation schemes. In the following we turn to investigate the limiting behavior of the system as a whole, by considering the normalized spectral efficiency in view of the analysis of Section \ref{sec: Zero-Forcing Front-End}. Accordingly, we restrict the discussion to a ZF front-end and a Gaussian $\Mat{H}$, and apply Proposition \ref{prop: Decoupling result for Zero-Forcing} to obtain the spectral efficiencies of the discrete lattice-based alphabet relaxation scheme and of CR-QPSK.

Starting with the discrete scheme, the spectral efficiency is obtained by incorporating \eqref{eq: 1RSB: Limit of the effective energy penalty} and \eqref{eq: 1RSB QPSK: Conditional distribution for c-m j c-n} into \eqref{eq: Joint probability in decoupling proposition}--\eqref{eq: Definition of the spectral efficiency}.
The observation made in Section \ref{sec: Replica Symmetry Breaking Example} regarding the independence of the real and imaginary parts of the precoder's output,  leads to the following conclusion.
The achievable rate in \eqref{eq: ZF Achievable rate in terms of mutual information} for QPSK input can  be obtained by treating QPSK signaling as two independent corresponding BPSK signaling settings.
Accordingly, the conditional precoder output probabilities, given a real BPSK input of $u=1$, are given by (cf.\ \eqref{eq: 1RSB QPSK: Conditional distribution for c-m j c-n})
\begin{eqnarray}
%P^{\bpsk}_{x|u=1}(c_k) &\triangleq& \int_{-\infty}^\infty \frac{\Theta_k(\xi)}{ \sum_{m=1}^L \Theta_m(\xi) } \, \e^{-\xi^2} \, \frac{\rd \xi}{\sqrt{\pi}} \ ,  \label{eq: 1RSB: QPSK: Conditional probabilities for equivalent BPSK setting}
P^{\bpsk}_{X|u=1}(c_k) &\triangleq& \int_{-\infty}^\infty \frac{\Theta_k(\xi)}{ \sum_{m=1}^L \Theta_m(\xi) } \, \e^{-\xi^2} \, \frac{\rd \xi}{\sqrt{\pi}} \ ,  \label{eq: 1RSB: QPSK: Conditional probabilities for equivalent BPSK setting}
\end{eqnarray}
where we set $\B_1=\set{c_k}_{k=1}^L$, and the conditional probabilities given $u=-1$ can be immediately obtained from symmetry considerations. It is then straightforward to show that the corresponding spectral efficiency is given by
\begin{equation}
%C^{\bpsk}(\rho) = \alpha \Biggl[1 - \int_{-\infty}^\infty  \sum_{k=1}^L P^{\bpsk}_{x|u=1}(c_k) \sqrt{\frac{\rho}{\pi}}  \e^{-\rho(\xi-c_k)^2}  \log_2\left(1 + \frac{\sum_{k=1}^L P^{\bpsk}_{x|u=1}(c_k)   \e^{-\rho(\xi+c_k)^2}}{\sum_{k=1}^L  P^{\bpsk}_{x|u=1}(c_k)  \e^{-\rho(\xi-c_k)^2}} \right) \, \rd \xi     \Biggr]  .
C^{\bpsk}(\rho) = \alpha \Biggl[1 - \int_{-\infty}^\infty  \sum_{k=1}^L P^{\bpsk}_{X|u=1}(c_k) \sqrt{\frac{\rho}{\pi}}  \e^{-\rho(\xi-c_k)^2}  \log_2\left(1 + \frac{\sum_{k=1}^L P^{\bpsk}_{X|u=1}(c_k)   \e^{-\rho(\xi+c_k)^2}}{\sum_{k=1}^L  P^{\bpsk}_{X|u=1}(c_k)  \e^{-\rho(\xi-c_k)^2}} \right) \, \rd \xi     \Biggr]  .
\label{eq: 1RSB: QPSK: Spectral efficiency for equivalent BPSK setting}
\end{equation}
The spectral efficiency with QPSK input is then obtained through the relation $C^\qpsk(\febno)=2C^{\bpsk}(\febno)$, while substituting
\begin{equation}
\rho  =  \frac{C \febno}{\alpha \bE_\rsb} \quad .
\end{equation}

Turning to the CR-QPSK scheme, and applying Proposition \ref{prop: Decoupling result for Zero-Forcing}, it can be shown that the conditional pdf of the equivalent single user channel output $y$, given an input $u$, is equal to\footnote{In the following notation the $\pm$ signs are designated with adherence to the corresponding signs of the real and imaginary parts of $u$. For example, for $u=1+j$ one should substitute $Q_2(y_\re)$, $\e^{-(y_\re - 1)^2\rho}$, $Q_2(y_\im)$, and $\e^{-(y_\im - 1)^2\rho}$ in the corresponding terms in \eqref{eq: CR-QPSK Conditional distribution of y given u}.
}
\begin{equation}
\label{eq: CR-QPSK Conditional distribution of y given u}
\begin{aligned}
%f^\crqpsk_{y|u}(y|u=\pm1 \pm j)
%&= \sqrt{\frac{\rho}{\pi}} \left[  \frac{1}{\sqrt{1+\rho\alpha\bE^{\crqpsk}}}  Q_2(\pm y_\re) \e^{-\frac{y_\re^2 \rho}{1+\rho\alpha\bE^{\crqpsk}}} + Q_1 \e^{-(y_\re \mp 1)^2\rho}  \right] \\
%&\phantom{=} \cdot \sqrt{\frac{\rho}{\pi}} \left[  \frac{1}{\sqrt{1+\rho\alpha\bE^{\crqpsk}}}  Q_2(\pm y_\im) \e^{-\frac{y_\im^2 \rho}{1+\rho\alpha\bE^{\crqpsk}}} + Q_1 \e^{-(y_\im \mp 1)^2\rho}  \right] \ ,
f^\crqpsk_{Y|U}(y|u=\pm1 \pm j)
&= \sqrt{\frac{\rho}{\pi}} \left[  \frac{1}{\sqrt{1+\rho\alpha\bE^{\crqpsk}}}  Q_2(\pm y_\re) \e^{-\frac{y_\re^2 \rho}{1+\rho\alpha\bE^{\crqpsk}}} + Q_1 \e^{-(y_\re \mp 1)^2\rho}  \right] \\
&\phantom{=} \cdot \sqrt{\frac{\rho}{\pi}} \left[  \frac{1}{\sqrt{1+\rho\alpha\bE^{\crqpsk}}}  Q_2(\pm y_\im) \e^{-\frac{y_\im^2 \rho}{1+\rho\alpha\bE^{\crqpsk}}} + Q_1 \e^{-(y_\im \mp 1)^2\rho}  \right] \ ,
\end{aligned}
\end{equation}
where we decomposed the complex argument as $y\triangleq y_\re + j y_\im$, and the real argument function $Q_2(\xi)$ is defined as
\begin{equation}
\label{eq: Definition of Q_2}
Q_2(\xi) \triangleq Q\left(\sqrt{\frac{2}{\alpha\bE^{\crqpsk}}}\frac{\rho\alpha\bE^{\crqpsk}(1-\xi)+1}{\sqrt{1+\rho\alpha\bE^{\crqpsk}}}\right) \quad, \quad \xi \in \mathbb{R} \quad .
\end{equation}
The marginal distribution of the equivalent single user channel output $y$ is given by
\begin{equation}
\label{eq: CR-QPSK Marginal distribution of y}
\begin{aligned}
%f^\crqpsk_y(y)
f^\crqpsk_Y(y)
&= \frac{1}{2}  \sqrt{\frac{\rho}{\pi}} \left[  \frac{Q_2(y_\re) +  Q_2(- y_\re)}{\sqrt{1+\rho\alpha\bE^{\crqpsk}}}  \e^{-\frac{y_\re^2 \rho}{1+\rho\alpha\bE^{\crqpsk}}} + Q_1 \left( \e^{-(y_\re - 1)^2\rho}  +  \e^{-(y_\re + 1)^2\rho}   \right) \right] \\
& \quad \cdot \frac{1}{2}  \sqrt{\frac{\rho}{\pi}} \left[  \frac{Q_2(y_\im) +  Q_2(- y_\im)}{\sqrt{1+\rho\alpha\bE^{\crqpsk}}} \e^{-\frac{y_\im^2 \rho}{1+\rho\alpha\bE^{\crqpsk}}} + Q_1 \left( \e^{-(y_\im - 1)^2\rho}  +  \e^{-(y_\im + 1)^2\rho}   \right) \right]  .
\end{aligned}
\end{equation}
Finally, following \eqref{eq: Definition of the spectral efficiency} and accounting for the inherent symmetry in \eqref{eq: CR-QPSK Conditional distribution of y given u} and \eqref{eq: CR-QPSK Marginal distribution of y}, the spectral efficiency of the CR-QPSK scheme is given  by
\begin{multline}
\label{eq: CR-QPSK Spectral Efficiency}
C^{\crqpsk}(\rho)
= 2\alpha \Biggl( 1 -   \int_{-\infty}^\infty
\sqrt{\frac{\rho}{\pi}} \left[  \frac{1}{\sqrt{1+\rho\alpha\bE^{\crqpsk}}}  Q_2(s) \e^{-\frac{s^2 \rho}{1+\rho\alpha\bE^{\crqpsk}}} + Q_1 \e^{-(s - 1)^2\rho}  \right] \\
 \cdot \log_2 \left(1 + \frac{    Q_2(- s)  \e^{-\frac{s^2 \rho}{1+\rho\alpha\bE^{\crqpsk}}} + \sqrt{1+\rho\alpha\bE^{\crqpsk}} Q_1 \e^{-(s + 1)^2\rho}  } { Q_2(s) \e^{-\frac{s^2 \rho}{1+\rho\alpha\bE^{\crqpsk}}} + \sqrt{1+\rho\alpha\bE^{\crqpsk}} Q_1 \e^{-(s - 1)^2\rho}}\right) \, \rd s \Biggr) \ .
\end{multline}

\begin{figure}[tb]
\begin{center}
\includegraphics[scale=0.4]{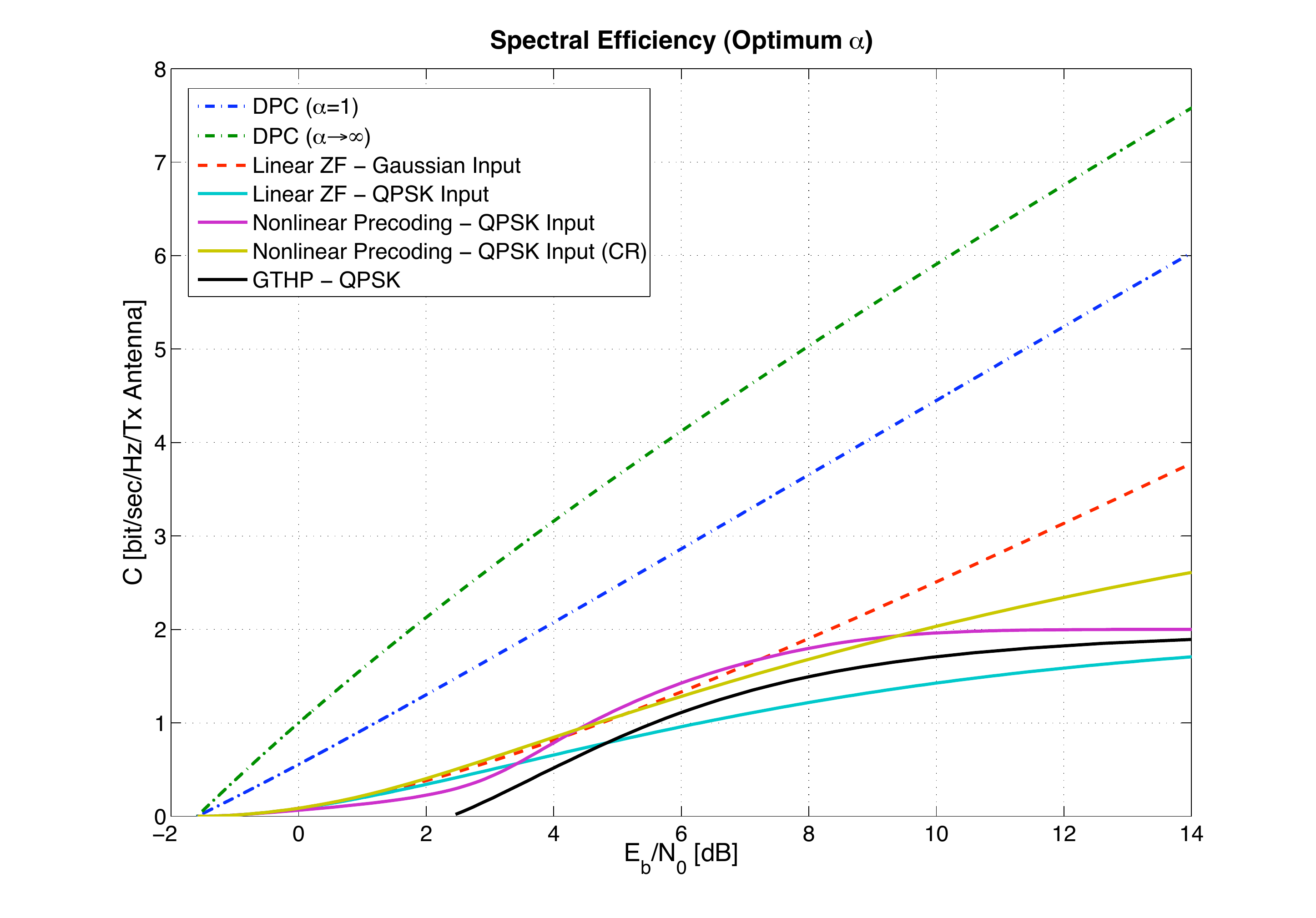}
\caption{\small Spectral efficiency results optimized with respect to the load $\alpha$.}
\label{fig: Spectral Efficiency Comparison - Optimum alpha}
\end{center}
\end{figure}

Comparative numerical spectral efficiency results are plotted in Figure \ref{fig: Spectral Efficiency Comparison - Optimum alpha}. The figure shows the spectral efficiencies of the discrete extended alphabet relaxation scheme (while taking $L=2$), and of the CR-QPSK scheme, as well as the spectral efficiency of linear ZF precoding for Gaussian and QPSK input (see \eqref{eq: ZF Spectral Efficiency with Gaussian input} and \eqref{eq: ZF QPSK Spectral Efficiency}, respectively), and the spectral efficiency of GTHP with QPSK input (following \eqref{App: eq: rate function for BPSK input and GTHP}--\eqref{App: Relation between spectral efficiency of QPSK and BSPK for GTHP}). The spectral efficiencies were evaluated for the \emph{optimum} choice of the system load $\alpha$. The optimum load is a function of $\febno$ and shown in Figure \ref{fig: Optimum Load}. 
\begin{figure}[tb]
\begin{center}
\includegraphics[scale=0.45]{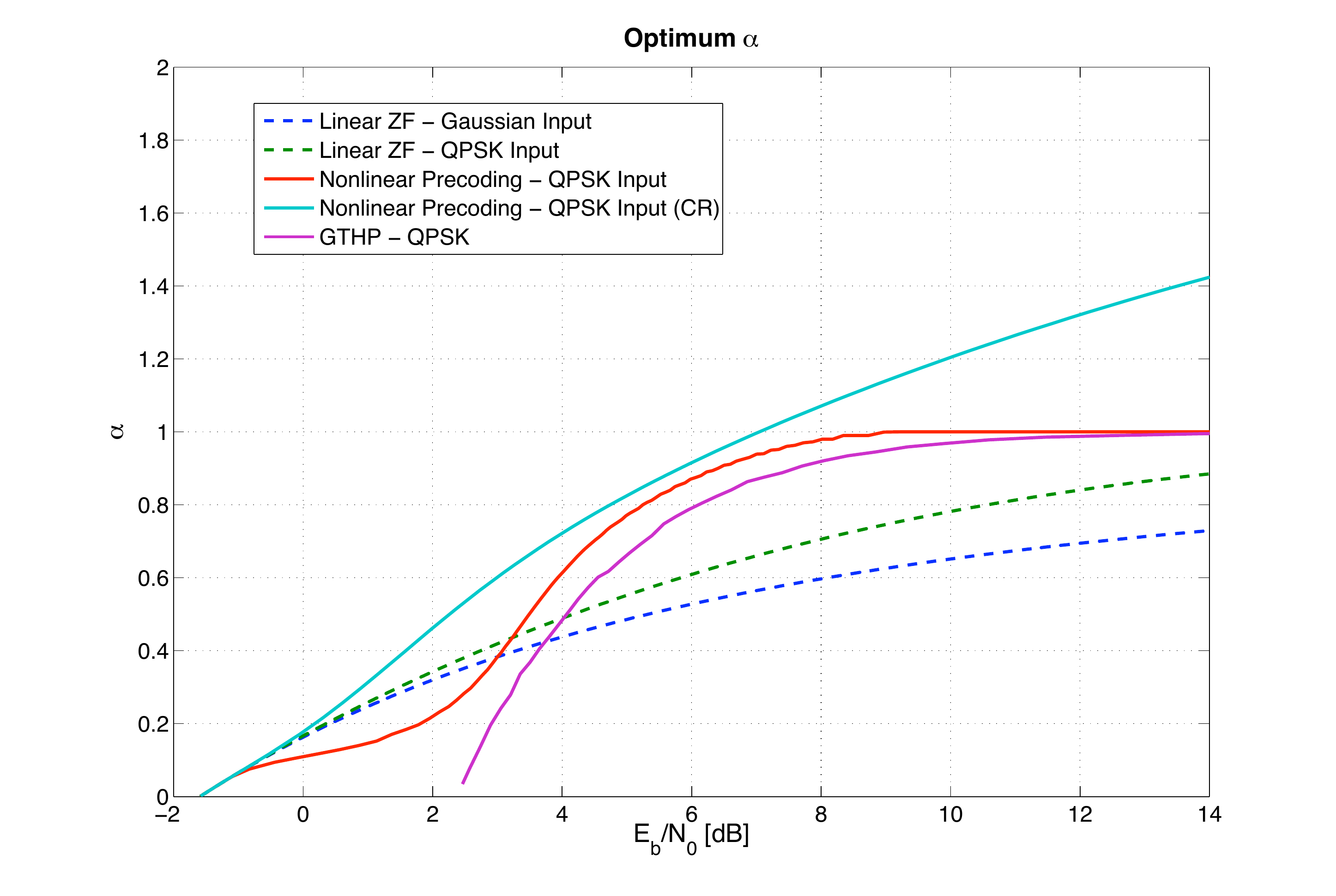}
%\caption{\small System load that maximizes spectral efficiency as a function of normalized signal-to-noise ratio $\febno$.}
\caption{\small System load that maximizes spectral efficiency as a function of $\febno$.}
\label{fig: Optimum Load}
\end{center}
\end{figure}
In Figure~\ref{fig: Spectral Efficiency Comparison - Optimum alpha}, the DPC spectral efficiency \eqref{eq: Spectral efficiency of DPC - Gaussian H} is also provided for comparison, evaluated both for $\alpha=1$, and for $\alpha \rightarrow \infty$ (specifying the ultimate performance). The optimization with respect to $\alpha$ emphasizes its role as a crucial system design parameter, facilitating the proper working point for each transmission scheme, per each $\febno$. It also naturally translates to a practical scheduling scheme, specifying the desired number of {\em simultaneously} active scheduled users per transmit antenna (see, e.g., \cite{Zaidel-Shamai-Verdu-JSAC-2001-short}).

The results indicate that nonlinear precoding can provide significant performance enhancement for medium to high $\febno$ values. The discrete lattice-based relaxation scheme is shown to outperform linear ZF with QPSK input for $\febno > 3.43\, \text{dB}$. The beneficial effect of the lattice relaxation scheme becomes more pronounced, the more the spectral efficiency approaches the upper limit of $2$ bits/sec/Hz per transmit antenna. For example, a spectral efficiency of $1.75$ bits/sec/Hz can be obtained with lattice relaxation already at $\febno \approx 7.66\,\text{dB}$, whereas linear ZF requires additional  $7.26\, \text{dB}$ for the same spectral efficiency. In fact, the QPSK-based lattice precoding scheme is shown to marginally outperform linear ZF with \emph{Gaussian} input for $4.19\,\text{dB} < \febno < 7.26\, \text{dB}$.
The lattice relaxation scheme also outperforms GTHP for all $\febno$ values, becoming more effective for medium to high $\febno$ (for example, GTHP needs $2.93\, \text{dB}$ more energy per bit to achieve $1.75$ bits/sec/Hz)\footnote{Note that in general the modulo-receiver employed by GTHP induces poor performance in the low spectral efficiency region (see Appendix \ref{App: Spectral Efficiency of Generalized Tomlinson-Harashima Precoding}).}. The gap from the DPC upper bound is, however, still essentially retained ($4.49\, \text{dB}$ at $1.75$ bits/sec/Hz, considering DPC with $\alpha=1$, to make a fairer comparison).

As for the CR-QPSK scheme, Figure \ref{fig: Spectral Efficiency Comparison - Optimum alpha} shows that it also provides a considerable performance enhancement over linear ZF with QPSK input.
It is  outperformed by the lattice relaxation scheme for 4.38 dB $< \febno <$ 9.40 dB.
It performs better at low values of $\febno$, and in fact it even negligibly outperforms linear ZF with \emph{Gaussian} input in the low $\febno$ region. Moreover, unlike the discrete scheme, CR-QPSK outperforms linear ZF precoding (with QPSK input) for \emph{all} $\febno$ values. 
Furthermore, it outperforms lattice relaxation in the high $\febno$ region, since it allows for loads up to $\alpha<2$ and therefore its spectral efficiency is no longer upper bounded by 2 bits/sec/Hz, but rather by 4 bits/sec/Hz per transmit antenna. Though, the convergence to the limiting spectral efficiency of 4 bits/s/Hz at high $\febno$ is rather slow. CR-QPSK also outperforms GTHP for all $\febno$ values, but the advantage is more significant for high $\febno$, where overloading is employed.
These results are of particular interest since the CR-QPSK scheme lends itself to efficient implementation, whereas the discrete relaxation scheme involves the solution of an NP-hard optimization problem. It is also important to note that, as shown in Figure \ref{fig: Energy penalty of CR-QPSK}, the CR-QPSK scheme is always inferior to the lattice relaxation scheme in terms of the limiting energy penalty. Hence, in view of the observations made here, one can conclude that restricting the analysis to the energy penalty alone provides only limited insight into the behavior of large \emph{coded} systems, as it essentially focuses only on the transmitter, while ignoring the impact of the nonlinear precoding scheme on the receiver.

%%%%%%%%%%%%%%%%%%%%

\section{Concluding Remarks}\label{sec: Concluding Remarks}

The replica symmetry breaking ansatz of statistical physics was employed in this paper to investigate the large system limit behavior of nonlinear precoding for the MIMO Gaussian broadcast channel based on linear zero-forcing and alphabet relaxation.
For lattice relaxations, the replica symmetric ansatz was shown to yield misleading results for system loads greater than approximately 0.3 while the one-step replica symmetry breaking ansatz provides sensible results for any load.
For exact results, however, multiple-step replica symmetry breaking must be considered. 
%which has been shown to be the correct solution for the more restricted SK model \cite{Guerra2002_ThermodynamicLimitSG,Talagrand-2006}.

Introducing a nonlinear superchannel comprising the actual channel and the precoder, allows for a Markov chain description of an individual user's channel. This enables the calculation of mutual information and spectral efficiency in the large system limit.
While convex QPSK relaxations are significantly outperformed by lattice relaxations in terms of transmitted energy per bit, they are very competitive when combined with strong error-correction coding as shown by the spectral efficiency analysis. Except for medium signal-to-noise ratios, they are superior to lattice relaxations. 
Both schemes were shown to outperform Tomlinson-Harashima precoding with QPSK input for all signal-to-noise ratios.

The combination of polynomial complexity and high spectral efficiency makes convex alphabet relaxation schemes, as introduced in \cite{Muller-Guo-Moustakas-JSAC-2008}, 
%a promising alternative to the NP-hard lattice relaxations (in particular when compared to the relative simplicity of the Tomlinson-Harashima precoding alternative). 
a promising alternative to the NP-hard lattice relaxations due to their polynomial complexity, and to Tomlinson-Harashima precoding due to their superior performance.
The results %on spectral efficiency 
motivate the search for convex %alternatives 
schemes amenable to efficient implementation. Additional examples for extended alphabets are currently investigated, see \cite{miguel-SP-09} for preliminary results. Note, however, that the problem of finding the optimum precoding scheme that maximizes the spectral efficiency in this framework is not at all trivial, as the corresponding equivalent channel statistics depend, in this setting, on the choice of input distribution and extended alphabet sets.
%\newpage

\appendix

\section{Higher RSB Orders}
\label{rstepRSB}

The $r$-step RSB ansatz reads
\begin{equation}
\label{eq: rRSB: Definition of the rRSB structure of Q}
\Mat{Q} = q_r \Mat{1}_{n \times n} + \sum\limits_{i=1}^{r} p_r^{(i)} \Mat{I}_{\frac{n\beta}{\mu_r^{(i)}} \times \frac{n\beta}{\mu_r^{(i)}}}\otimes  \Mat{1}_{\frac{\mu_r^{(i)}}{\beta} \times \frac{\mu_r^{(i)}}{\beta}} + \frac{\chi_r}{\beta} \Mat{I}_{n \times n} \quad ,
\end{equation}
using the constants $\set{q_r,p_r^{(1)},\dots,p_r^{(r)},\chi_r,\mu_r^{(1)},\dots,\mu_r^{(r)}}$. The limit as $r\to \infty$ is called \emph{full} RSB and gives the \emph{exact} solution to the problem \cite{Talagrand-2006}.
The particular temperature-dependent scaling of some parts of $\Mat Q$ is used to evaluate the free energy at zero temperature without getting divergent terms.
If a finite temperature is of interest a different scaling may be considered.
Plugging (\ref{eq: rRSB: Definition of the rRSB structure of Q}) into (\ref{eq: General expression for energy}), while exploiting the particular structure of $\Mat{Q}$, we find:

\begin{prop}\label{prop: r-RSB solution for energy entropy and free energy}
For any temperature, the energy for $r$-step RSB is
\begin{align}
\E(\beta) =& -q_r  \left[\chi_r+\sum_{i=1}^r\mu_r^{(i)}p_r^{(i)}\right]R^\prime \! \left(-\chi_r-\sum_{i=1}^r\mu_r^{(i)}p_r^{(i)}\right)+\nonumber\\
& + \left[q_r+\frac{\chi_r+\sum_{i=1}^r\mu_r^{(i)}p_r^{(i)}}{\mu_r^{(1)}}\right]R \! \left(-\chi_r-\sum_{i=1}^r\mu_r^{(i)}p_r^{(i)}\right)+\nonumber\\
& + \sum\limits_{j=2}^r\left(\frac1{\mu_r^{(j)}}-\frac1{\mu_r^{(j-1)}}\right)\left[{\chi_r+\sum_{i=j}^r\mu_r^{(i)}p_r^{(i)}}\right]R \! \left(-\chi_r-\sum_{i=j}^r\mu_r^{(i)}p_r^{(i)}\right)\nonumber+\\
& + \left(\frac{\chi_r}\beta -\frac{\chi_r}{\mu_r^{(r)}}\right) R (-\chi_r)
\label{eq:energy of r-RSB} \quad ,
\end{align}
where $R^\prime(\cdot)$ denotes the derivative of the function $R(\cdot)$.
\end{prop}

In order to proceed to full RSB, the limit $r\to\infty$ must be taken.
Naively, one might think this would make the sums in (\ref{eq:energy of r-RSB}) diverge.
However, the macroscopic parameters are determined by the saddle point equations, which guarantee that the sums stay finite through decreasing the macroscopic parameters.
Thus, we introduce a continuum of macroscopic parameters $\mu(x)$ and $p(x)$, taken over $0\le x\le 1$, such that
\begin{align}
q&=q_r \ ,\\
\chi&=\chi_r \ , \\
p(i/r)&=p_r^{(i)} \ , \\
\mu(i/r)&=\mu_r^{(i)} \ , \\
p(0)&=q_r \ , \\
\mu(0)&=1 \ ,
\end{align}
and the function
\begin{equation}
\sG(x) \triangleq -  \chi - \int\limits_x^1\mu(y)p(y){\rm d}y \quad .
\end{equation}
Accordingly, we find for the energy in the limit $r\to\infty$
\begin{multline}
\E(\beta) = q\,\sG(0) \,R^\prime [\sG(0)]
+ \left[q-\frac{\sG(0)}{\mu(0)}\right] R[\sG(0)] \\
 +\int_0^1 \sG(x) \, R[\sG(x)] \frac{{\rm d}\mu(x)}{\mu^2(x)}
 + \left(\frac{\sG(1)}{\mu(1)}+\frac\chi\beta\right) R[\sG(1)] \quad .
 \label{fullrsb}
 \end{multline}
Using integration by parts, (\ref{fullrsb}) simplifies to
 \begin{equation}
\E(\beta) = \left. \left[\sG(x)\,R[\sG(x)]\right]^\prime\right|_{x=0}
+ \frac\chi\beta\, R[\sG(1)] + \int\limits_0^1 \frac{{\rm d}\!\left[\sG(x)\,R[\sG(x)]\right]}{\mu(x)} \quad.
 \end{equation}
The functions $p(x)$ and $\mu(x)$
must be determined by the respective saddle point equations.

%%%%%%%%%%%%%%%%%%%%%%%%%
\section{Proofs for 1RSB}
\label{app: 1RSB: Proof of conditional probabilities proposition}

%%%%%%%%%%%%%%%%%%%%
\subsection{Proposition \ref{prop: 1RSB: Limiting Energy Penalty}}
\label{proof41}

The joint distribution of the entries of the vector $\vct{x}$, conditioned on both the input vector $\vct{u}$ and the channel transfer matrix $\Mat{H}$, is given for a non-zero temperature by the Boltzmann distribution
\begin{equation}
\label{eq: 1RSB: Boltzmann Distribution of tilde-u}
%P(\vct{x}|\Mat{H},\vct{u}) 
P_{\mathcal{B}}(\vct{x}|\Mat{H},\vct{u})= \frac{1}{\Z} \e^{-\beta \vct{x}^\dag\Mat{J}\vct{x} } \quad ,
\end{equation}
where $\Z$ is the partition function defined in \eqref{eq: definition of Z}. Taking the limit $\beta \rightarrow \infty$ (zero temperature), the denominator in \eqref{eq: 1RSB: Boltzmann Distribution of tilde-u} is dominated by its maximum value term, and the limiting \emph{joint} distribution of the entries of $\vct{x}$, \emph{conditioned on all inputs}, converges to the Dirac measure at $\argmin_{\vct{x}\in \B_{\vct u}} \vct{x}^\dag\Mat{J}\vct{x}$, corresponding to the minimum normalized energy penalty, as given by Proposition \ref{prop: 1RSB: Limiting Energy Penalty}. 
%In this appendix we focus, however, on the \emph{marginal} conditional distribution of the entries of $\vct{x}$.

To prove Proposition \ref{prop: 1RSB: Limiting Energy Penalty}, we will need to evaluate the free energy averaged over all realizations of 
 $\vct{u}$ and $\Mat{H}$. For future convenience, we also include the dummy variable $h$ and the function $V(\cdot)$ defined in \eqref{defsumV} and rewrite the free energy as   
\begin{equation}\label{eq: 1RSB: Definition of free energy in the replica analysis}
\begin{aligned}
-\beta\sF(\beta) &=  \lim_{K\rightarrow\infty} \frac{1}{ K} E_{\vct{u},\Mat{H}}\set{ \log  \Z({h};\vct{u},\Mat{H})}  \\
&= \lim_{K\rightarrow\infty} \sum_{\vct{u} \in \U^K} 
%P_{\vct u}(\vct{u}) 
P_{U^K}(\vct{u}) 
\left(  \frac{1}{ K}  E_{\Mat{H}}\set{ \log  \Z({h};\vct{u},\Mat{H})}  \right)
\quad ,
\end{aligned}
\end{equation}
where $\Z({h};\vct{u},\Mat{H})$ is given by \eqref{eq: 1RSB: Definition of script-Z1}.
The second equality is a manifestation of the underlying assumption that the coded symbols of all users are drawn randomly
and independently of the channel transfer matrix $\Mat{H}$. 
In view of this formulation, we consider now the limit of the term in the parentheses above
\begin{equation}\label{eq: 1RSB: Limit to be considered - Expectation wrt H}
 \lim_{K\rightarrow\infty} \frac{1}{ K} E_{\Mat{H}}\set{ \log  \Z(h;\vct{u},\Mat{H})} \quad .
\end{equation}
As shown later on, this inner limit is a deterministic quantity, for almost every realization of the input vector $\vct{u}$. It will hence be concluded that in fact
\begin{equation}
\label{eq: 1RSB: Claim - Free energy is given by the inner limit alone}
-\beta\sF(\beta) =  \lim_{K\rightarrow\infty} \frac{1}{ K} E_{\Mat{H}}\set{ \log  \Z(h;\vct{u},\Mat{H})} \quad .
\end{equation}
As indicated earlier, the key tool in the derivation of the above quantity is the replica method of statistical physics, using the identity
\begin{equation}\label{eq: 1RSB: Identity used for replica analysis}
%E_{\Mat{H}}\set{ \log  \Z(h;\vct{u},\Mat{H})} = \lim_{n\rightarrow 0} \frac{\partial}{\partial n} \log  E_{\Mat{H}} \set{[\Z(h;\vct{u},\Mat{H})]^n}
E_{\Mat{H}}\set{ \log  \Z(h;\vct{u},\Mat{H})} = \lim_{n\rightarrow 0} \frac{1}{n} \log  E_{\Mat{H}} \set{[\Z(h;\vct{u},\Mat{H})]^n} \quad ,
\end{equation}
and following the outline in Section \ref{sec: Replica Analysis}. With that in mind, the quantity $[\Z(h;\vct{u},\Mat{H})]^n$ is regarded as consisting of $n$ identical replicas of the original (unnormalized) probability model in the following way \cite{Tanaka-2002}
\begin{equation}
\label{eq: 1RSB: Definition of script-Z-n for replica calculations}
\begin{aligned}
{[\Z(h;\vct{u},\Mat{H})]}^{n} &= \left(\sum_{\vct{x}\in \B_{\vct{u}}} \e^{-\beta V(h,\xi,\upsilon,\vct x,\vct u))}\e^{-\beta \vct{x}^\dag\Mat{J}\vct{x} }\right)^n \\
&=\sum_{\set{\vct{x}_a}}  \e^{-\beta \sum_{a=1}^n V(h,\xi,\upsilon,\vct x_a,\vct u)}
 \e^{\sum_{a=1}^n -\beta \vct{x}_a^\dag\Mat{J}\vct{x}_a } \quad .
\end{aligned}
\end{equation}
Interchanging the limits of $K \rightarrow \infty$ and $n\rightarrow 0$, the focus is first on the derivation of the limit
\begin{equation}
\label{eq: 1RSB: Definition of Xi-n}
\begin{aligned}
\Xi_n & \triangleq \lim_{K\rightarrow\infty} \frac{1}{K} \log E_{\Mat{H}} \set{{[\Z(h;\vct{u},\Mat{H})]}^{n}} \\
%&= \lim_{K\rightarrow\infty} \frac{1}{K}      \log E_{\Mat{H}} \set{  \sum_{\set{\vct{x}_a}}  \e^{-\beta  \sum_{a=1}^n V(h,\upsilon,\xi,\vct u,\vct x_a)}  \e^{\sum_{a=1}^n -\beta \vct{x}_a^\dag\Mat{J}\vct{x}_a } } \\
&= \lim_{K\rightarrow\infty} \frac{1}{K} \log E_{\Mat{H}} \set{     \sum_{\set{\vct{x}_a}}  \e^{-\beta  \sum_{a=1}^n V(h,\xi,\upsilon,\vct x_a,\vct u)}  \e^{- \Tr (\beta \Mat{J} \sum_{a=1}^n \vct{x}_a \vct{x}_a^\dag )}  } \quad .
\end{aligned}
\end{equation}
%Representing $\Xi_n$ as
%\begin{equation}
%\label{eq: Representation of Xi-n with trace}
%\Xi_n = \lim_{K\rightarrow\infty} \frac{1}{K} \log E_{\Mat{H}} \set{     \sum_{\set{\vct{x}_a}}  \e^{\beta h \sum_{a=1}^n \phil(\vct{x}_a)}  \e^{- \Tr (\beta \Mat{J} \sum_{a=1}^n \vct{x}_a \vct{x}_a^\dag )}  } \quad,
%\end{equation}
Since the first exponential term within the expectation is independent of the channel transfer matrix $\Mat{H}$, $\Xi_n$ can be rewritten as
\begin{equation}
\label{eq: 1RSB: Representation of Xi-n with expectation inside the sum}
\Xi_n = \lim_{K\rightarrow\infty} \frac{1}{K} \log \left(    \sum_{\set{\vct{x}_a}}  \e^{-\beta \sum_{a=1}^n V(h,\xi,\upsilon,\vct x_a,\vct u)} E_{\Mat{H}}\set{ \e^{- \Tr (\beta \Mat{J} \sum_{a=1}^n \vct{x}_a \vct{x}_a^\dag )} } \right) \quad.
\end{equation}
The inner expectation in \eqref{eq: 1RSB: Representation of Xi-n with expectation inside the sum} is the Harish-Chandra-Itzykson-Zuber integral (see \cite{Muller-Guo-Moustakas-JSAC-2008} and references therein), and the objective here is its evaluation for fixed-rank matrices $\sum_{a=1}^n \vct{x}_a \vct{x}_a^\dag$, in the large $K$ limit. This problem was recently considered in \cite{Guionnet-Maida-2005}, and invoking Theorem 1.7 therein, \eqref{eq: 1RSB: Representation of Xi-n with expectation inside the sum} can be represented for large $K$ as\footnote{$o(K)$ is used here to denote quantities that satisfy $\lim_{K \to \infty} o(K)/K = 0$.}
\begin{equation}
\label{eq: 1RSB: Representation of Xi-n with the Guionnet Maida result}
\Xi_n = \lim_{K\rightarrow\infty} \frac{1}{K} \log  \left(\sum_{\set{\vct{x}_a}}  \e^{-\beta  \sum_{a=1}^n V(h,\xi,\upsilon,\vct x_a,\vct u)}  \e^{- K \sum_{a=1}^n \int_0^{\lambda_a} R(-w)\, \rd w  + o(K)}  \right) \quad,
\end{equation}
where $R(w)$ is the $R$-transform of the limiting eigenvalue distribution of the matrix $\Mat{J}$, and $\set{\lambda_a}$ denote the eigenvalues of the $n \times n$ matrix $\beta \Mat{Q}$ with $\Mat{Q}$ defined through\footnote{
Here \cite[Theorem 1.7]{Guionnet-Maida-2005} is applied individually for all \emph{given}
vectors $\set{\vct{x}_a}$.
}
\begin{equation}
\label{eq: 1RSB: Definition of the entries of the matrix Q}
Q_{ab} = \frac{1}{K} \vct{x}^\dag_a \vct{x}_b \triangleq  \frac{1}{K} \sum_{k=1}^K x_{ak}^* x_{bk} \quad .
\end{equation}
Since additive exponential terms of order $o(K)$ have no effect on the results in the limiting regime as $K\rightarrow \infty$, due to the $\frac{1}{K}$ factor outside the logarithm in \eqref{eq: 1RSB: Representation of Xi-n with the Guionnet Maida result} (this shall become clear in the derivation to follow), any such terms are dropped henceforth for notational simplicity.

In order to calculate the summation in \eqref{eq: 1RSB: Representation of Xi-n with the Guionnet Maida result}, the procedure employed in \cite{Muller-Guo-Moustakas-JSAC-2008} is repeated here, and the $Kn$-dimensional space spanned by the replicas is split into subshells by means of \eqref{eq: Definition of the subshell S-Q}.
%\begin{equation}\label{eq: 1RSB: Definition of the subshell S-Q}
%S(\Mat{Q}) \triangleq \set{\vct{x}_1,\dots,\vct{x}_n \big| \vct{x}_a^\dag \vct{x}_b = KQ_{ab}} .
%\end{equation}
Assuming $\Mat{Q}^\dag=\Mat{Q}$, $\Xi_n$ can be represented as
\begin{equation}
\label{eq: 1RSB: Representation of Xi-n as integral with I and G}
\Xi_n = \lim_{K\rightarrow\infty} \frac{1}{K} \log \left(  \int \e^{K\L} \e^{K \I(\Mat{Q})} \e^{-K \G(\Mat{Q})}\, \D \Mat{Q} \right) \quad ,
\end{equation}
where
\begin{equation}
\label{eq: 1RSB: Definition of DQ}
\D\Mat{Q} = \prod_{a=1}^n \rd Q_{aa} \prod_{b=a+1}^n \rd \Re(Q_{ab}) \, \rd \Im (Q_{ab})
\end{equation}
is the integration measure,
\begin{eqnarray}
\label{eq: 1RSB: Definition of G-Q}
\G(\Mat{Q}) &=&\sum_{a=1}^n \int\limits_0^{\beta \lambda_a(\Mat{Q})} R(-w)\, \rd w\\
&=&\sum_{a=1}^n \int\limits_0^{\beta} \lambda_a(\Mat{Q}) R(-w\lambda_a(\Mat{Q}))\, \rd w\\
 &=& \int\limits_0^{\beta}\, \Tr\left[\Mat{Q} R(-w\Mat{Q})\right] \, \rd w \quad , \label{GQint}
 \end{eqnarray}
since the trace is the sum of the eigenvalues,
\begin{equation}
\label{eq: 1RSB: Definition of L-Q}
\L = -\frac{\beta}{K}  \sum_{a=1}^n V(h,\xi,\upsilon,\vct x_a,\vct u) \quad ,
\end{equation}
and
\begin{equation}
\label{eq: 1RSB: Definition of int I-Q}
\e^{K\I(\Mat{Q})} = \sum_{\set{\vct{x}_a}} \prod_{a=1}^n \delta(\vct{x}_a^\dag \vct{x}_a - K Q_{aa}) \prod_{b=a+1}^n \delta(\Re[\vct{x}_a^\dag \vct{x}_b - KQ_{ab}]) \, \delta(\Im[\vct{x}_a^\dag \vct{x}_b - KQ_{ab}])
\end{equation}
is the probability weight of the subshell.

Starting with $\e^{K\I(\Mat{Q})}\e^{K\L}$ we follow \cite{Muller-Guo-Moustakas-JSAC-2008} and represent the Dirac measures using the inverse Laplace transform. This is performed by introducing the \emph{complex} variables $\set{\tilde{Q}_{ab}^{(I)}}$, $1\le a \le b \le n$, and $\set{\tilde{Q}_{ab}^{(Q)}}$, $1\le a \le b \le n$, and defining the matrix $\Mat{\tilde{Q}}$ with elements (taking $a<b$)
\begin{eqnarray}
\tilde{Q}_{aa} &=& \tilde{Q}_{aa}^{(I)} \quad , \label{eq: 1RSB: Definition of tQ-aa}\\
\tilde{Q}_{ab} &=& \frac{1}{2} \left( \tilde{Q}_{ab}^{(I)} - j\tilde{Q}_{ab}^{(Q)} \right)  \quad , \\
\tilde{Q}_{ba} &=& \frac{1}{2} \left( \tilde{Q}_{ab}^{(I)} + j\tilde{Q}_{ab}^{(Q)} \right)   \quad .
\end{eqnarray}
Denoting by $\Mat{P}$ the Hermitian matrix with elements $P_{ab} = \vct{x}_a^\dag \vct{x}_b - KQ_{ab}$, this yields
\begin{eqnarray}
\delta(P_{aa}) &=& \int_{\mathcal{J}} \e^{\tilde{Q}_{aa}P_{aa}} \, \frac{d\tilde{Q}_{aa}^{(I)}}{2\pi j} \quad , \\
\delta(\Re\set{P_{ab}})\, \delta(\Im\set{P_{ab}})
 &=& \int_{\mathcal{J}^{2}} \e^{\tilde{Q}_{ab}^{(I)}\Re\set{P_{ab}} - \tilde{Q}_{ab}^{(Q)}\Im\set{P_{ab}}} \, \frac{d \tilde{Q}_{ab}^{(I)} \, d\tilde{Q}_{ab}^{(Q)}}{(2\pi j)^2}  \\
&=& \int_{\mathcal{J}^{2}} \e^{\tilde{Q}_{ab}P_{ba} + \tilde{Q}_{ba}P_{ab}} \, \frac{d \tilde{Q}_{ab}^{(I)} \, d\tilde{Q}_{ab}^{(Q)}}{(2\pi j)^2}   \quad , \label{eq: 1RSB: Definition of dRPab dIPab}
\end{eqnarray}
where the integration is over $\mathcal{J}=(t-j\infty , t+j\infty)$, for some $t\in \mathbb{R}$ (note that $P_{ab} = P_{ba}^*$).  Substituting in \eqref{eq: 1RSB: Definition of int I-Q} and combining with \eqref{eq: 1RSB: Definition of L-Q}, it follows that
\begin{equation}
\label{eq: 1RSB: First expression for int I-Q L-Q}
\begin{aligned}
\e^{K\I(\Mat{Q}) } \e^{K\L} &= \sum\limits_{\set{\vct{x}_a}} \int_{{\mathcal{J}}^{n^2}} \e^{\sum\limits_{a,b} \tilde{Q}_{ab} (\vct{x}_b^\dag \vct{x}_a - K Q_{ba})} \e^{-\beta \sum\limits_{a=1}^n V(h,\xi,\upsilon,\vct{x}_a,\vct{u}) } \tilde{\D}\Mat{\tilde{Q}} \\
&= \int_{{\mathcal{J}}^{n^2}} \e^{-K\Tr( \Mat{\tilde{Q}} \Mat{Q}) } \left( \sum\limits_{\set{\vct{x}_a}}
\e^{\sum\limits_{a,b} \tilde{Q}_{ab} \vct{x}_b^\dag \vct{x}_a} \e^{-\beta \sum\limits_{a=1}^n V(h,\xi,\upsilon,\vct{x}_a,\vct{u}) }
\right)\, \tilde{\D}\Mat{\tilde{Q}} \quad ,
\end{aligned}
\end{equation}
where
\begin{equation}
\tilde{\D}\Mat{\tilde{Q}} = \prod_{a=1}^n \left( \frac{d\tilde{Q}_{aa}^{(I)}}{2\pi j} \prod_{b=a+1}^n \frac{d\tilde{Q}_{ab}^{(I)} \, d\tilde{Q}_{ab}^{(Q)}}{(2\pi j)^2} \right).
\end{equation}
Considering the inner summation in \eqref{eq: 1RSB: First expression for int I-Q L-Q}, then rearranging terms and using \eqref{defsumV} the expression can be rewritten as
\begin{equation}
\label{eq: 1RSB: Inner summation representation in terms of k}
 \sum_{\set{\vct{x}_a}}
\e^{\sum\limits_{a,b} \tilde{Q}_{ab} \vct{x}_b^\dag \vct{x}_a} \e^{-\beta \sum\limits_{a=1}^n V(h,\xi,\upsilon,\vct{x}_a,\vct{u})} =
\prod_{k=1}^K \sum_{\set{x_a \in \B_{u_k}}} \e^{\left(\sum\limits_{a,b} \tilde{Q}_{ab} x_b^*x_a\right) +h \beta \sum\limits_{a=1}^n   1\left\{(x_a,u_k)=(\xi,\upsilon)\right\} } \quad .
\end{equation}
Defining
\begin{equation}
\label{eq: 1RSB: Definition of M-k-tilde-Q}
M_k(\Mat{\tilde{Q}}) = \sum_{\set{x_a \in \B_{u_k}}} \e^{\left(\sum\limits_{a,b} x_b^*x_a \tilde{Q}_{ab}\right) +h \beta \sum\limits_{a=1}^n  1\left\{(x_a,u_k)=(\xi,\upsilon)\right\}}  \quad,
\end{equation}
one finally gets
\begin{equation}
\label{eq: 1RSB: Expression of int I-Q L-Q with log M-k}
\e^{K\I(\Mat{Q}) } \e^{K\L} = \int_{{\mathcal{J}}^{n^2}} \e^{-K\Tr( \Mat{\tilde{Q}} \Mat{Q}) + \sum\limits_{k=1}^K \log M_k(\Mat{\tilde{Q}})} \, \tilde{\D}\Mat{\tilde{Q}} \quad.
\end{equation}
Now, using the underlying assumption that the coded symbols transmitted by different users are i.i.d., one can apply   the strong law of large numbers for $K\rightarrow\infty$ to get
\begin{align}
\log M(\Mat{\tilde Q}) &\triangleq \frac{1}{K} \sum\limits_{k=1}^K \log M_k(\Mat{\tilde Q}) \\
&\rightarrow \int \log  \sum_{\left\{x_a \in \B_{u}\right\}}   \e^{\left(\sum\limits_{a,b} x_b^*x_a \tilde{Q}_{ab}\right) +h\beta \sum\limits_{a=1}^n 1\left\{(x_a,u)=(\xi,\upsilon)\right\}}  \rd F_U(u) \label{generalM}\\
&=\int \log  \sum_{{\bf x} \in \B_{u}^n}   \e^{ {\bf x}^\dagger\Mat{\tilde{Q}}{\bf x}+h\beta \sum\limits_{a=1}^n 1\left\{(x_a,u)=(\xi,\upsilon)\right\}}  \rd F_U(u) \ , \label{generalM2}
\end{align}
where the convergence is in the \emph{almost sure} sense, for any extended alphabets such that the expectation in \eqref{generalM} exists.
Note that this observation implies that any randomness due to $\vct{u}$ in the RHS of \eqref{eq: 1RSB: Representation of Xi-n as integral with I and G} effectively vanishes at the large system limit, due to the normalization with respect to $K$ outside the logarithm.

The next step in the evaluation of \eqref{eq: 1RSB: Representation of Xi-n as integral with I and G} is the observation that in the limit as $K \rightarrow \infty$, the integrand therein is dominated by the exponential term with maximal exponent. Therefore, only the subshell that corresponds to this extremal value of the correlation between the vectors $\set{\vct{x}_a}$ is relevant for the calculation of the integral. Thus, we have at the saddle point
\begin{equation}
\label{speq1}
%\frac\partial{\partial \Mat Q} \left[ \G(\Mat Q) + \Tr(\Mat{Q\tilde Q})\right] = \Mat 0\quad .
\frac\partial{\partial \Mat Q} \left[ \G(\Mat Q) + \Tr(\Mat{\tilde{Q} Q})\right] = \Mat 0\quad .
\end{equation}
Since the trace is the sum of the eigenvalues, we can write \eqref{eq: 1RSB: Definition of G-Q}
as 
\begin{equation}
\G(\Mat Q)= {\Tr} \int\limits_0^{\beta \Mat{Q}} R(-w)\, \rd w \label{GQintnew}
\end{equation}
and \eqref{speq1} gives
\begin{equation}
\label{generalQtilde}
\Mat{\tilde Q} = -\beta R(-\beta \Mat Q)\quad.
\end{equation}
Furthermore, we observe that also the integrand in \eqref{eq: 1RSB: Expression of int I-Q L-Q with log M-k} is dominated by the exponential term with maximal exponent in the limit $K\to\infty$.
Thus, at the saddle point we have
\begin{equation}
\label{speq2}
%\frac\partial{\partial \Mat{\tilde Q}} \left[ \log M(\Mat{\tilde Q}) - \Tr(\Mat{Q\tilde Q})\right] = \Mat 0 \quad.
\frac\partial{\partial \Mat{\tilde Q}} \left[ \log M(\Mat{\tilde Q}) - \Tr(\Mat{\tilde{Q} Q})\right] = \Mat 0 \quad.
\end{equation}
With \eqref{generalM2}, this gives
\begin{equation}
\label{generalQ}
\Mat Q = \int
\frac{ \sum\limits_{{\bf x} \in \B_{u}^n}  {\bf xx}^\dagger  \e^{{\bf x}^\dagger \Mat{\tilde{Q}}{\bf x} +h\beta \sum\limits_{a=1}^n 1\left\{(x_a,u)=(\xi,\upsilon)\right\}} }
{ \sum\limits_{{\bf x} \in \B_{u}^n}   \e^{{\bf x}^\dagger\Mat{\tilde{Q}}{\bf x} +h\beta \sum\limits_{a=1}^n 1\left\{(x_a,u)=(\xi,\upsilon)\right\}} }
\, \rd F_U(u)\quad.
\end{equation}

We now invoke the \emph{1RSB} assumption \eqref{eq: 1RSB: Definition of the 1RSB structure of Q} regarding the structure of the matrices $\Mat{Q}$ 
%and $\Mat{\tilde{Q}}$ 
at the saddle-point that dominate the integral.
%Accordingly, we require that (c.f., Section \ref{sec: Replica Analysis})
%\begin{equation}
%\label{eq: app: 1RSB: Definition of the 1RSB structure of Q}
%\Mat{Q} = q_1 \Mat{1}_{n \times n} + p_1 \Mat{I}_{\frac{n\beta}{\mu_1} \times \frac{n\beta}{\mu_1}}\otimes  \Mat{1}_{\frac{\mu_1}{\beta} \times \frac{\mu_1}{\beta}} + \frac{\chi_1}{\beta} \Mat{I}_{n \times n} \quad ,
%\end{equation}
%where $q_1$, $p_1$, $\chi_1$, and $\mu_1$ are four macroscopic parameters uniquely determining $\Mat{Q}$.
In a similar manner we set
\begin{equation}
\label{eq: 1RSB: Definition of the 1RSB structure of tilde-Q}
\Mat{\tilde{Q}} = \beta^2f_1^2 \Mat{1}_{n \times n} + \beta^2g_1^2 \Mat{I}_{\frac{n\beta}{\mu_1} \times \frac{n\beta}{\mu_1}}\otimes  \Mat{1}_{\frac{\mu_1}{\beta} \times \frac{\mu_1}{\beta}} - \beta\varepsilon_1 \Mat{I}_{n \times n} \quad ,
\end{equation}
introducing the macroscopic parameters $f_1$, $g_1$, and $\varepsilon_1$.

With these assumptions one can explicitly obtain the eigenvalues of the matrix $\beta \Mat{Q} \,$\footnote{The eigenvalue $(\beta n q_1 + \mu_1 p + \chi_1) $ occurs with multiplicity $1$, the eigenvalue $(\mu_1 p_1 + \chi_1)$ occurs with multiplicity $(\frac{n\beta}{\mu_1}-1)$, and the eigenvalue $\chi_1$ occurs with multiplicity $(n-\frac{n\beta}{\mu_1})$.}, and $\G(\Mat{Q})$ can be rewritten as
\begin{multline}
\label{eq: 1RSB: Representation of G-q-chi}
\G(q_1,p_1,\chi_1,\mu_1) = \left(n-\frac{n\beta}{\mu_1}\right) \int\limits_0^{\chi_1} R(-w) \, \rd w \\ + \left(\frac{n\beta}{\mu_1}-1\right) \int\limits_0^{\chi_1 + \mu_1 p_1} R(-w) \, \rd w +  \int\limits_0^{\chi_1 + \mu_1 p_1 + \beta n q_1} R(-w) \, \rd w \ .
\end{multline}
It also follows from the 1RSB assumption that
\begin{equation}
\label{eq: 1RSB: Expression for the Tr Q-Q}
\Tr(\Mat{\tilde{Q}} \Mat{Q}) =
\begin{bmatrix}
      \beta^2 f_1^2 & \beta^2g_1^2 & -\beta \varepsilon_1
\end{bmatrix}
\begin{bmatrix}
      n^2 & \frac{n\mu_1}{\beta} & n \\
      \frac{n\mu_1}{\beta} & \frac{n\mu_1}{\beta} & n \\
      n & n & n \\
\end{bmatrix}
\begin{bmatrix}
      q_1 \\ p_1 \\ \frac{\chi_1}{\beta}
\end{bmatrix}
\quad ,
\end{equation}
and
\begin{multline}
\label{eq: 1RSB: Expression of M-k-f-g-epsilon-mu}
\log M(f_1,g_1,\varepsilon_1,\mu_1) =\\
\int \log  \sum_{\set{x_a \in \B_{u}}} \e^{\beta^2f_1^2\abs{\sum\limits_{a=1}^n x_a}^2
+ \beta^2 g_1^2 \sum\limits_{l=0}^{\frac{n\beta}{\mu_1}-1}  \abs{\sum\limits_{a=1}^{\frac{\mu}{\beta}} x_{a+\frac{l \mu_1}{\beta}}}^2
- \beta\varepsilon_1  \sum\limits_{a=1}^n  \abs{x_a}^2
+h\beta \sum\limits_{a=1}^n 1\left\{(x_a,u)=(\xi,\upsilon)\right\}} \rd F_U(u) \ .
\end{multline}

Due to \eqref{speq1}, the partial derivatives of
\begin{equation}
\G(q_1,p_1,\chi_1,\mu_1) + \Tr(\Mat{\tilde{Q}} \Mat{Q})
\end{equation}
with respect to $q_1$, $p_1$, and $\chi_1$ must vanish as $K\rightarrow\infty$ by definition of the saddle point.
\if 0
\uwave{It is worth pointing out here that the fixed point equations for the $\Mat{\tilde{Q}}$ can be generated directly from the above equation written as follows:}
\begin{equation}
\label{eq:G+TrQQ}
\frac{\partial}{\partial \Mat{Q}}\left(\G({\bf Q}) + \Tr(\Mat{\tilde{Q}} \Mat{Q})\right) = {\bf 0}
\end{equation}
\uwave{which means that for every $a,b=1,\cdots,n$}
\begin{eqnarray}
\label{eq:G+TrQQ1}
\frac{\partial}{\partial Q_{ab}}\left(\G({\bf Q}) + \Tr(\Mat{\tilde{Q}} \Mat{Q})\right) &=& 0 \\ \nonumber
\frac{\partial \G({\bf Q})}{\partial Q_{ab}} + \tilde{Q}_{ba} &=& 0
\end{eqnarray}
\uwave{However, from the definition of $\G$ this becomes simply }
\begin{equation}
\label{eq:Qtilde=R(Q)}
\Mat{\tilde{Q}}_{ba} = - R(\beta \Mat{Q})_{ba}
\end{equation}
\uwave{This formulation is general and holds for both RS and RSB. In fact, the proof in Appendix A relies on this result. In addition, there is a similar result about the derivative with respect to $\tilde{Q}$,
as follows:}
\begin{equation}
\label{eq: General_saddle_point_eqs2}
\Mat{Q} = \sum_{\set{\vct{x}_a}} \vct{x}\vct{x}^\dag
\e^{\vct{x}^\dag\Mat{\tilde{Q}}\vct{x}} \e^{-\beta \sum_{a=1}^n V(h,\upsilon,\xi,\vct u,\vct x_a)} \quad .
\end{equation}
\fi
Using \eqref{eq: 1RSB: Representation of G-q-chi} and \eqref{eq: 1RSB: Expression for the Tr Q-Q} this yields the following set of equations
\begin{eqnarray}
0 & = & n^2\beta^2 f_1^2 + n \beta  \mu_1 g_1^2 - n\beta \varepsilon_1 + n\beta R(-\chi_1 - \mu_1 p_1 -\beta n q_1) \ , \label{eq: 1RSB: Equation after taking derivative wrt q} \\
0 & = & n \beta \mu_1 f_1^2 + n \beta \mu_1 g_1^2 - n \beta \varepsilon_1 + (n\beta -\mu_1) R(-\chi_1 -\mu_1 p_1) \nonumber \\ & & + \mu_1 R(-\chi_1 -\mu_1 p_1 -n \beta q_1) \ , \label{eq: 1RSB: Equation after taking derivative wrt p} \\
0 & = & n \beta f_1^2 + n \beta g_1^2  -n \varepsilon_1 + \left(n-\frac{n \beta}{\mu_1}\right) R(-\chi_1) + \left(\frac{n \beta}{\mu_1} - 1\right) R(-\chi_1 -\mu_1 p_1) \nonumber \\
& & + R(-\chi_1 -\mu_1 p_1 - n \beta q_1) \ . \label{eq: 1RSB: Equation after taking derivative wrt chi}
\end{eqnarray}
Solving for $\varepsilon_1$, $g_1$, and $f_1$, while focusing on the limit as $n\rightarrow 0$, one gets
\begin{eqnarray}
\varepsilon_1 & = &   R(-\chi_1) \label{eq: app: 1RSB: Expression for epsilion after taking derivative wrt q p and chi} \ ,\\
g_1 & = & \sqrt{\frac{R(-\chi_1) - R(-\chi_1 - \mu_1 p_1)}{\mu_1}} \label{eq: app: 1RSB: Expression for g after taking derivative wrt q p and chi} \ ,\\
f_1 & = & \sqrt{{\frac{R(-\chi_1 - \mu_1 p_1) - R(-\chi_1 - \mu_1 p_1 - n \beta q_1)}{n \beta}}} \xrightarrow{n \rightarrow 0} \sqrt{q_1 R'(-\chi_1-\mu_1 p_1)} \ .
\label{eq: app: 1RSB: Expression for f after taking derivative wrt q p and chi}
\end{eqnarray}

\newcommand{\kernel}{{\cal K}(u,x,y,z)}
\newcommand{\kernelp}{{\cal K}(u,x,\tilde{y},z)}
We now rewrite the expression for $M_k(f_1,g_1,\varepsilon_1,\mu_1)$ in \eqref{eq: 1RSB: Expression of M-k-f-g-epsilon-mu} using the Hubbard-Stratonovich transform and the shortened notation of \eqref{eq: Shortened notation for the Dz integral}
\begin{equation}
\label{eq: 1RSB: Definition of the Hubbard-Stratonovich transform}
%\e^{\abs{x}^2} = \int_\mathbb{C} \e^{2 \Re\set{x z^*}} \underset{\triangleq Dz}{\underbrace{\e^{-\abs{z}^2} \, \frac{dz}{\pi}}} \quad ,
\e^{\abs{x}^2} = \int\limits_\mathbb{C} \e^{2 \Re\set{x z^*}} \, \rD z  \quad ,
\end{equation}
yielding (c.f.\ \cite[(66)-(70)]{Muller-Guo-Moustakas-JSAC-2008})
\begin{align}
\label{eq: 1RSB: Expression of M-k-f-g-epsilon-mu with Hubbard-Stratonovich}
&\log M(f_1,g_1,\varepsilon_1,\mu_1) = \nonumber\\
&\quad = \int \log \sum_{\set{x_a \in \B_{u}}} \int\limits_{\mathbb{C}} \e^{\sum\limits_{a=1}^n \left[2 \beta f_1 \Re\set{x_a z^*} - \beta\varepsilon_1  \abs{x_a}^2  +h\beta  \,1\left\{(x_a,u)=(\xi,\upsilon)\right\}\right]
+ \beta^2 g_1^2 \sum\limits_{l=0}^{\frac{n\beta}{\mu_1}-1}  \abs{\sum\limits_{a=1}^{\frac{\mu_1}{\beta}} x_{a+\frac{l \mu_1}{\beta}}}^2
} \, \rD z \, \rd F_U(u) \\
\label{eq: 1RSB: Definition of M-f-g-epsilon-mu}
&\quad =  \int \log \int\limits_{\mathbb{C}} \left[ \int\limits_{\mathbb{C}} \left(  \sum_{x \in \B_{u}} \kernel
\right)^{\frac{\mu_1}{\beta}} \rD y \right]^{\frac{n \beta}{\mu_1}} \rD z \,  \rd F_U(u) \ ,
\end{align}
with
\begin{equation}
\label{defkernel}
\kernel \triangleq \e^{ 2 \beta \Re\set{x (f_1 z^* + g_1y^*)} - \beta\varepsilon_1  \abs{x}^2  +h\beta \,1\left\{(x,u)=(\xi,\upsilon)\right\} }\quad.
\end{equation}

Due to \eqref{speq2}, the partial derivatives of
\begin{equation}
\log M(f_1,g_1,\varepsilon_1,\mu_1) - \Tr(\Mat{\tilde{Q}}\Mat{Q})
\end{equation}
with respect to $f_1$, $g_1$ and $\varepsilon_1$, must also vanish as $K \rightarrow \infty$.
This produces the following set of equations (while taking the limit as $n\rightarrow 0$)
\begin{multline}\label{eq: app: 1RSB: Equation after partial derivative wrt f}
\chi_1 + p_1\mu_1 =  \frac{1}{f_1} \int \int\limits_{\mathbb{C}^2}  \frac{ \left(  \sum\limits_{x \in \B_{u}}  \kernel  \right)^{\frac{\mu_1}{\beta}-1} }{\int\limits_{\mathbb{C}} \left(  \sum\limits_{x \in \B_{u}}   \kernelp
\right)^{\frac{\mu_1}{\beta}} \, \rD \tilde{y}} %\\ \cdot 
\sum\limits_{x \in \B_{u}}  \Re\set{xz^*} \kernel \, \rD y \, \rD z \, \rd F_U(u) \ ,
\end{multline}
\begin{multline}
\label{eq: app: 1RSB: Equation after partial derivative wrt g}
\chi_1 + (q_1+p_1)\mu_1  =  \frac{1}{g_1} \int \int\limits_{\mathbb{C}^2}  \frac{ \left(  \sum\limits_{x \in \B_{u}}   \kernel  \right)^{\frac{\mu_1}{\beta}-1} }{\int\limits_{\mathbb{C}} \left(  \sum\limits_{x \in \B_{u}}   \kernelp
\right)^{\frac{\mu_1}{\beta}} \, \rD \tilde{y}} \\ \cdot 
\sum\limits_{x \in \B_{u}}  \Re\set{xy^*} \kernel \, \rD y \, \rD z \, \rd F_U(u) \ ,
\end{multline}
\begin{multline}\label{eq: app: 1RSB: Equation after partial derivative wrt epsilon}
q_1+p_1  =   \int \int\limits_{\mathbb{C}^2}  \frac{ \left(  \sum\limits_{x \in \B_{u}}   \kernel  \right)^{\frac{\mu_1}{\beta}-1} }{\int\limits_{\mathbb{C}} \left(  \sum\limits_{x \in \B_{u}}   \kernelp  \right)^{\frac{\mu_1}{\beta}} \, \rD \tilde{y}} %\\ \cdot  
\sum\limits_{x \in \B_{u}}  \abs{x}^2 \kernel \, \rD y \, \rD z \, \rd F_U(u) - \frac{\chi_1}{\beta} \quad .
\end{multline}
The parameter $\mu_1$ should also be chosen such that the partial derivative of
\begin{equation}
\G(q_1,p_1,\chi_1,\mu_1) + \Tr(\Mat{\tilde{Q}} \Mat{Q}) - \log M(f_1,g_1,\varepsilon_1,\mu_1)
\end{equation}
with respect to $\mu_1$ vanishes. This yields at the limit as $n \rightarrow 0$
\begin{multline}\label{eq: app: 1RSB: Equation after partial derivative wrt mu}
0 = - \frac{1}{\mu_1^2} \int\limits_{\chi_1}^{\chi_1+\mu_1 p_1}
 R(-w)\, dw  + \frac{p_1}{\mu_1} R(-\chi_1) + q_1 g_1^2 %\\
+ \int  \int\limits_{\mathbb{C}} \Biggl[ \frac{1}{\mu_1^2} \log\left( \int\limits_{\mathbb{C}} \left(  \sum_{x \in \B_{u}}  \kernel  \right)^{\frac{\mu_1}{\beta}} \, \rD y \right) \\
- \int\limits_{\mathbb{C}} \frac{\left(  \sum\limits_{x \in \B_{u}}   \kernel  \right)^{\frac{\mu_1}{\beta}} }{\beta \mu_1 \int\limits_{\mathbb{C}} \left(  \sum\limits_{x \in \B_{u}}   \kernelp  \right)^{\frac{\mu_1}{\beta}} \, \rD \tilde{y} } %\\ 
\cdot \log\left(\sum\limits_{x \in \B_{u}}   \kernel
\right) \, \rD y \Biggr] \, \rD z \, \rd F_U(u) \quad .
\end{multline}

Incorporating all previous results, we get that the quantity $\Xi_n$ of \eqref{eq: 1RSB: Definition of Xi-n} is equal to
\begin{equation}
\label{eq: 1RSB: Equation for Xi-n with saddle point integration results}
\begin{aligned}
\Xi_n &= \I(\Mat{Q}) + \L - \G(\Mat{Q}) \\
&= \log M(f_1,g_1,\varepsilon_1,\mu_1) -   \beta^2 f_1^2 q_1 n^2 \\
&\phantom{=} - \left(\beta f_1^2 (p_1\mu_1 +\chi_1)+\beta g_1^2 (q_1\mu_1 +p_1\mu_1+ \chi_1 )-\beta\varepsilon_1 (q_1+p_1+\frac{\chi_1}{\beta})\right) n \\
&\phantom{=}
-  \left(n-\frac{n\beta}{\mu_1}\right) \int\limits_0^{\chi_1} R(-w) \, \rd w - \left(\frac{n\beta}{\mu_1}-1\right) \int\limits_0^{\chi_1 + \mu_1 p_1} R(-w) \, \rd w \\
&\phantom{=} -  \int\limits_0^{\chi_1 + \mu_1 p_1 + \beta n q_1} R(-w) \, \rd w  \quad ,
\end{aligned}
\end{equation}
where the macroscopic parameters $\set{f_1,g_1,\varepsilon_1,q_1,p_1,\chi_1,\mu_1}$ are obtained from the saddle point fixed-point equations
\eqref{eq: app: 1RSB: Expression for epsilion after taking derivative wrt q p and chi}--\eqref{eq: app: 1RSB: Expression for f after taking derivative wrt q p and chi}, \eqref{eq: app: 1RSB: Equation after partial derivative wrt f}--\eqref{eq: app: 1RSB: Equation after partial derivative wrt epsilon}, and \eqref{eq: app: 1RSB: Equation after partial derivative wrt mu}. Now in view of \eqref{eq: 1RSB: Identity used for replica analysis}, the next step in the derivation is to take the limit
%of $\frac{1}{n}\Xi_n$ as $n\rightarrow 0$.
%Accordingly, we get
\begin{multline}
\lim_{n\rightarrow 0} \frac{1}{n}\Xi_n =  \frac{\beta}{\mu_1}\int \int\limits_{\mathbb{C}} \log \left(\int\limits_{\mathbb{C}} \left(  \sum\limits_{x \in \B_{u}}   \kernel\right)^{\frac{\mu_1}{\beta}} \, \rD y \right) \, \rD z \,  \rd F_U(u) \\
-\beta f_1^2(\chi_1 + p_1\mu_1) - \beta g_1^2 (\chi_1 + (p_1+q_1)\mu_1) + \beta \varepsilon_1(q_1+p_1+\frac{\chi_1}{\beta}) \\
 - \left(1-\frac{\beta}{\mu_1}\right)\int\limits_0^{\chi_1} R(-w) \, \rd w - \frac{\beta}{\mu_1}\int\limits_0^{\chi_1+\mu_1 p_1} R(-w) \, \rd w - \beta q_1 R(-\chi_1-\mu_1 p_1) \quad . \label{eq: 1RSB: Expression for the derivative of Xi-n at n zero with f g and varepsilon}
\end{multline}
But in fact
\begin{equation}
\lim_{n\rightarrow 0} \frac{1}{n}\Xi_n =  - \beta \sF(\beta,h) \quad,
\label{eq: app: Relation between Xi-n and free energy}
\end{equation}
which is justified by the observation that $\Xi_n$ converges to the same limit for almost every realization of $\vct{u}$, applying the law of large numbers in \eqref{generalM} (see also \eqref{eq: 1RSB: Limit to be considered - Expectation wrt H} and the discussion that follows).

%Before proceeding with the derivation of the empirical joint distribution, 
We note at this point that the energy penalty $\bE_\rsb$ satisfies
\begin{equation}
\label{eq: 1RSB: Definition of the energy penalty as a function of F}
\bE_\rsb = \lim_{\beta \rightarrow \infty} \sF(\beta,h) |_{h={0}} \quad ,
\end{equation}
and \eqref{eq: 1RSB: Limit of the effective energy penalty} can be readily expressed from Proposition \ref{prop: r-RSB solution for energy entropy and free energy} as a function of the macroscopic parameters $\set{q_1,p_1,\chi_1,\mu_1}$.
In order to evaluate the energy penalty, it is thus left to derive the fixed point equations that determine these parameters, as given by \eqref{eq: 1RSB: Equation after partial derivative wrt f - beta inf h zero}--\eqref{eq: 1RSB: Equation after partial derivative wrt mu - beta inf and h zero}, which is obtained by substituting $h=0$ and taking the limit as $\beta\rightarrow \infty$ in \eqref{eq: app: 1RSB: Equation after partial derivative wrt f}--\eqref{eq: app: 1RSB: Equation after partial derivative wrt epsilon}, and \eqref{eq: app: 1RSB: Equation after partial derivative wrt mu} after back-substitution of \eqref{defkernel}. This completes the proof of Proposition \ref{prop: 1RSB: Limiting Energy Penalty}.

%%%%%%%%%%%%%%%%%%%%%%%%%%%
\subsection{Proposition \ref{prop: 1RSB: Marginal limiting conditional distribution of x given u}}
\label{proof42}

%We will now derive the limiting conditional distribution of $x$ given $u$. 
%We start with (\ref{eq317}). We therefore need to evaluate the derivative of the free energy with respect to $h$. 
We derive the limiting conditional distribution of the precoder output $x$ given the input $u$  
starting from (\ref{eq317}). We therefore need to evaluate the derivative of the free energy with respect to $h$. 
This can be done directly given
(\ref{eq: app: Relation between Xi-n and free energy}). Taking the partial derivative in \eqref{eq: 1RSB: Expression for the derivative of Xi-n at n zero with f g and varepsilon} and using \eqref{defkernel}, we get
\begin{align}
\label{eq: 1RSB: Equation for the limiting conditional distribution - finite beta}
P_{X|U}(\xi|\upsilon) &= \frac1{P_U(\upsilon)} \left.\frac{\partial}{\partial h} \sF(\beta,h)\right|_{h={0}} \\
 &= \int \int_{\mathbb{C}^2}  \frac{\left(\sum_{x \in \B_{u}}   \e^{ 2 \beta \Re\set{x (f_1 z^* + g_1y^*)} - \beta\varepsilon_1  \abs{x}^2  }\right)^{\frac{\mu_1}{\beta}-1}  }{\int_{\mathbb{C}} \left(  \sum_{x \in \B_{u}}   \e^{ 2 \beta \Re\set{x (f_1 z^* + g_1\tilde y^*)} - \beta\varepsilon_1  \abs{x}^2  }  \right)^{\frac{\mu_1}{\beta}} \, \rD \tilde y}  \nonumber\\
 & \qquad \cdot  \sum_{x \in \B_{u}}  \e^{ 2 \beta \Re\set{x (f_1 z^* + g_1y^*)} - \beta\varepsilon_1  \abs{{x}}^2   }  \, \rD y  \, \rD z \, \frac{1\{(x,u)=(\xi,\upsilon)\}}{{P_U(\upsilon)}  } \,{\rm d}F_U(u) \\
&=\int_{\mathbb{C}^2}  \frac{\left(\sum_{x \in \B_{\upsilon}}   \e^{ 2 \beta \Re\set{x (f_1 z^* + g_1y^*)} - \beta\varepsilon_1  \abs{x}^2  }\right)^{\frac{\mu_1}{\beta}-1} }{\int_{\mathbb{C}} \left(  \sum_{x \in \B_{\upsilon}}   \e^{ 2 \beta \Re\set{x (f_1 z^* + g_1\tilde y^*)} - \beta\varepsilon_1  \abs{x}^2  }  \right)^{\frac{\mu_1}{\beta}} \, \rD \tilde y} 
 \,\e^{ 2 \beta \Re\set{\xi (f_1 z^* + g_1y^*)} - \beta\varepsilon_1  \abs{{\xi}}^2   } \, \rD y  \, \rD z \ .
\end{align}
After taking the limit $\beta \rightarrow \infty$, while applying the saddle point integration rule, we finally get \eqref{eq: 1RSB: Equation for the limiting conditional distribution - Infinite beta}.  This completes the proof of Proposition \ref{prop: 1RSB: Marginal limiting conditional distribution of x given u}.

%%%%%%%%%%%%%%%%%%%%%%%%%%%%%%%%%%%%%%%%%%%
\subsection{Proposition \ref{prop: 1RSB entropy result}}
\label{proof43}

We start from  \eqref{eq: Thermodynamics definition of entropy}, \eqref{eq: Definition of normalized free energy}, and \eqref{eq: Definition of normalized entropy} which yield
\begin{eqnarray} \label{eq:entropy_definition1}
{\sS(\beta)} &=& \beta \frac{\rm d \left(\beta \sF(\beta)\right)}{\rm d \beta} - \beta \sF(\beta) \\ \nonumber
&=& \beta \frac{\partial \left(\beta \sF(\beta)\right)}{\partial \beta} - \beta \sF(\beta) \quad .
\end{eqnarray}
The partial derivative with respect to $\beta$ above reflects the fact that all implicit dependencies of $\sF(\beta)$ on $\beta$ through its dependence on other parameters, e.g., $f_1, g_1, \chi_1, \mu_1$, have vanishing derivatives since $\sF(\beta)$ is evaluated at a saddle point. Making use of the saddle point equations, we find that
\begin{equation} \label{eq:entropy_definition2}
\beta\frac{\partial \left(\beta \sF(\beta)\right)}{\partial \beta} =
\chi_1\left(1-\frac{\beta}{\mu_1}\right) R(-\chi_1) + \beta\left(\frac{\chi_1}{\mu_1}+p_1+q_1\right)R(-\chi_1-\mu_1p_1) - \beta q_1R'(-\chi_1-\mu_1p_1).
\end{equation}
Next, we need to analyze the behavior of  $\beta\sF(\beta)$ for large $\beta$. This can be seen directly through \eqref{eq: 1RSB: Expression for the derivative of Xi-n at n zero with f g and varepsilon}, \eqref{eq: app: Relation between Xi-n and free energy}. The first term in
\eqref{eq: 1RSB: Expression for the derivative of Xi-n at n zero with f g and varepsilon} {can be shown to be of the form $\beta A + O(\beta^{-1})$, where $A\in\mathbb R$ is some constant. 
%The leading order term $O(\beta)$ originates from the most probable point. 
The reason for this behavior stems from the fact that for a discrete alphabet the corrections to the leading order term are exponential in $\beta$, except for a small $O(\beta^{-1})$ region close to the nearest neighbor points in the lattice. Therefore, to order $\beta^{-1}$, the first term in \eqref{eq: 1RSB: Expression for the derivative of Xi-n at n zero with f g and varepsilon} is simply $\beta$ times its partial derivative with respect to $\beta$. This enables us to evaluate the value of $\beta\sF(\beta)$ for large $\beta$ to the order $\beta^{-1}$ as follows: }
\begin{multline}
\beta\sF(\beta) =  \beta\frac{\partial}{\partial \beta}\left[\frac{\beta}{\mu_1}\int \int_{\mathbb{C}} \log \left(\int_{\mathbb{C}} \left(  \sum_{x \in \B_{u}}   \e^{ 2 \beta \Re\set{x (f_1 z^* + g_1y^*)} - \beta\varepsilon_1  \abs{x}^2 }  \right)^{\frac{\mu_1}{\beta}} \, \rD y \right) \, \rD z \,  \rd F_U(u)\right] \\
-\beta f_1^2(\chi_1 + p_1\mu_1) - \beta g_1^2 (\chi_1 + (p_1+q_1)\mu_1) + \beta \varepsilon_1(q_1+p_1+\frac{\chi_1}{\beta}) \\
 - \left(1-\frac{\beta}{\mu_1}\right)\int_0^{\chi_1} R(-w) \, \rd w - \frac{\beta}{\mu_1}\int_0^{\chi_1+\mu_1 p_1} R(-w) \, \rd w - \beta q_1 R(-\chi_1-\mu_1 p_1) \quad . \label{eq: 1RSB: Expression for the free_energy1}
\end{multline}
{Using the fixed point equations we may re-express the first line as follows:}
\begin{multline}
\frac{\beta}{\mu_1}\int_{\chi_1}^{\chi_1+\mu_1 p_1}  R(-w) {\rm d}w + \beta q_1 R(-\chi_1-\mu_1p_1) +2\beta q_1 R'(-\chi_1-\mu_1p_1) \\
-\chi_1 R(-\chi_1) + \frac{2\beta\chi_1R(-\chi_1)}{\mu_1}-2\beta R(-\chi_1-\mu_1p_1)\left(\frac{\chi_1}{\mu_1}-q_1-p_1\right)
\label{eq: 1RSB: Expression for the free_energy2}\quad.
\end{multline}
{Plugging this into the above equation and using \eqref{eq: app: 1RSB: Expression for epsilion after taking derivative wrt q p and chi}--\eqref{eq: app: 1RSB: Expression for f after taking derivative wrt q p and chi}, we eventually get the following equation for the zero-temperature entropy}
\begin{eqnarray}\label{eq:zero_temp_entropy}
{\bS} &=& \chi_1R(-\chi_1)-\int_0^{\chi_1} R(-w){\rm d}w\quad.
\end{eqnarray}

Remarkably the above equation for the entropy holds also for the RS case. To recover the RS structure of the equations above we start with $\mu_1/\beta=1$ and $\chi_{0}=\chi_1+\mu_1p_1$. Then, we find that $q_{0} = q_1$, $\epsilon_{0} = \epsilon_1-\beta g_1^2$, and $f_{0} = f_1$. After that we find the equations to reduce to the RS case analyzed in \cite{Muller-Guo-Moustakas-JSAC-2008}. 

%%%%%%%%%%%%%%%%%%%%

\section{Proof of Proposition~\ref{prop:general_saddle_point_free_energy}}\label{app: Proof of Proposition on Free Energy}
We will now apply (\ref{eq: Thermodynamics definition of energy}) to express the energy in a compact fashion. We start by considering  the representation of the normalized average free energy in terms of $\Mat{Q}$ (see \eqref{eq: Definition of normalized free energy} and \eqref{eq: Definition of the subshell S-Q}), and let us denote this representation, for the sake of clarity, as $\sF(\Mat Q,\beta)$. In general, the replica crosscorrelation matrix $\Mat Q$ depends on $\beta$.
However, at the saddle point we have (by definition)
\begin{equation}
\frac{\partial \sF(\Mat{Q},\beta)}{\partial \Mat Q} = \Mat 0.
\label{eq:saddle_point_condition_for_Q}
\end{equation}
Thus, the total derivative in \eqref{eq: Thermodynamics definition of energy} becomes a partial derivative at the saddle point, i.e.\
\begin{equation}
\E(\beta)=\frac{\partial}{\partial \beta} \left(\beta \sF(\Mat{Q},\beta)\right) \quad .
\end{equation}
Referring to the proof in Appendix \ref{app: 1RSB: Proof of conditional probabilities proposition}, then with \eqref{eq: 1RSB: Definition of free energy in the replica analysis}, \eqref{eq: 1RSB: Identity used for replica analysis}, \eqref{eq: 1RSB: Definition of Xi-n}, \eqref{eq: 1RSB: Representation of Xi-n as integral with I and G}, and \eqref{GQint}, while substituting $h=0$, this gives
\begin{equation}
\label{eq:general_energy_diff}
\E(\beta)= \lim\limits_{n\to 0}  \frac{1}{n} \frac{\partial}{\partial \beta}  \int\limits_0^\beta \Tr[\Mat QR(-w\Mat Q)]\, {\rm d}w \quad ,
\end{equation}
which is easily shown to be equivalent to \eqref{eq: General expression for energy}.
Furthermore, we get \eqref{prop31b} by plugging \eqref{generalQtilde} into \eqref{generalQ} while substituting $h=0$.

%%%%%%%%%%%%%%%%%%%%%%%%%%%%%%%%%%%%%
\section{Discrete Lattice Relaxation: Small $\chi_1$ Approximation Near Unit Load (1RSB)}
\label{app: chi-zero approximation of the 1RSB equations in the vicinity of unit load}

This appendix provides an approximate derivation of the 1RSB equations for the discrete lattice-based alphabet relaxation scheme of Section \ref{sec: Replica Symmetry Breaking Example}, while assuming a Gaussian $\Mat{H}$, and a ZF front-end.
The approximation is based on the numerical observation that the macroscopic parameter $\chi_1$, employed in the 1RSB ansatz for this setting, approaches zero as the system load gets close to unity. This approximation considerably simplifies the numerical solution of the 1RSB equations in this region of the system load.

%\begin{figure}[t]
%\centering
%\includegraphics[scale=0.35]{Seminar_Energy_Penalty_Comparison_small_chi}
%\caption{Comparison of the energy penalties for the discrete lattice alphabet relaxation (1RSB solution), and the CR-QPSK scheme (RS solution). The lower bound of \cite{Ryan-Collings-Clarkson-Heath-ICC-2008} is provided for comparison as well as finite system simulation results. The small $\chi$ approximations of Appendix \ref{app: chi-zero approximation of the 1RSB equations in the vicinity of unit load} are employed to obtain the energy penalty for $\alpha > 0.978$.}
%\label{fig: Spectral Efficiency Comparison - Optimum alpha}
%\end{figure}

\subsection{Case I: $\alpha=1$}
\label{app: subsec: Case I  - alpha unity}

For $\alpha=1$, the $R$-transform of $\Mat{J}=\Mat{T}^\dag\Mat{T}$ (see Proposition  \ref{prop: 1RSB: Limiting Energy Penalty}) satisfies
\begin{eqnarray}
R(-w) &=& \frac{\alpha - 1 + \sqrt{(1-\alpha)^2+4\alpha w}}{2\alpha w} \underset{\alpha=1}{=}\quad \frac{1}{\sqrt{w}}, \quad \forall w\in \mathbb{R}\label{app: eq: R-transform for Gaussian H - alpha 1} \\
R'(-w) &=& \frac{\left( 1 - \alpha -  \sqrt{(1-\alpha)^2+4\alpha w}    \right)^2}{4\alpha w^2 \sqrt{(1-\alpha)^2+4\alpha w}} \underset{\alpha=1}{=} \frac{1}{2w^{\frac{3}{2}}}, \quad \forall w\in \mathbb{R} \quad . \label{app: eq: Derivative of R-transform for Gaussian H - alpha 1}
\end{eqnarray}
Considering the small $\chi_1$ regime, we get from \eqref{eq: 1RSB: Expression for epsilion after taking derivative wrt q p and chi}--\eqref{eq: 1RSB: Expression for f after taking derivative wrt q p and chi}
\begin{eqnarray}
\varepsilon_1 & = & \frac{1}{\sqrt{\chi_1}} \quad ,\\
g_1 & = & \sqrt{\frac{\frac{1}{\sqrt{\chi_1}} - \frac{1}{\sqrt{\chi_1 + \mu_1 p_1}}}{\mu_1}} \underset{\chi_1 \ll 1}{\approx} \sqrt{\frac{1}{\mu_1\sqrt{\chi_1}}} \quad ,\\
f_1 & \underset{\chi_1 \ll 1}{\approx} & \sqrt{q_1 R'(-\mu_1 p_1)}  = \sqrt{q_1 \frac{1}{2 (\mu_1 p_1)^{\frac{3}{2}}}} \quad . \label{app: eq f for alpha 1 and chi small}
\end{eqnarray}
Particularizing to the two-dimensional discrete lattice-based alphabet relaxation scheme in concern, one gets from \eqref{eq: 1RSB QPSK: Definition of psi-xi}
\begin{equation}\label{app: eq: 1RSB QPSK: Definition of psi-xi  - alpha 1}
\psi_k(\xi) = \frac{\varepsilon_1 v_k - f_1 \xi}{g_1} \underset{\chi_1 \ll 1}{\approx} \frac{v_k}{\sqrt{\chi_1}} \sqrt{\mu_1 \sqrt{\chi_1}}  = \sqrt{\mu_1} \frac{1}{\chi_1^{\frac{1}{4}}} \frac{c_k + c_{k-1}}{2} \quad \forall \abs{\xi} < \infty .
\end{equation}
We now rewrite the function $\Theta_k(\xi)$ of \eqref{eq: 1RSB QPSK: Definition of Theta-xi} as
\begin{equation}
\begin{aligned}
\Theta_k(\xi)
&\triangleq \e^{\mu_1 c_k \left[ (\mu_1 g_1^2 - \varepsilon) c_k+ 2  f_1 \xi\right]} \left[ Q\left( \sqrt{2} (\psi_k(\xi) - \mu_1  g_1 c_k) \right) - Q\left( \sqrt{2} (\psi_{k+1}(\xi) - \mu_1  g_1 c_k) \right)\right]
\end{aligned}
\end{equation}
and observe the following. Starting with exponential argument, we get
\begin{equation}
 (\mu_1 g_1^2 - \varepsilon_1) c_k+ 2  f_1 \xi \underset{\chi_1 \ll 1}{\approx} -R(-\mu_1 p_1)c_k + 2 f_1 \xi = -\frac{c_k}{\sqrt{\mu_1 p_1}} + 2 f_1 \xi \quad ,
\end{equation}
while the arguments of the $Q(\cdot)$ functions satisfiy
\begin{equation}
\psi_k(\xi) - \mu_1 g_1 c_k \underset{\chi_1 \ll 1}{\approx}  \frac{\sqrt{\mu_1}}{\chi_1^{\frac{1}{4}}} \frac{c_k + c_{k-1}}{2} -  \frac{\sqrt{\mu_1}}{\chi_1^{\frac{1}{4}}} c_k = - \frac{\sqrt{\mu_1}}{\chi_1^{\frac{1}{4}}} \frac{c_k - c_{k-1}}{2} \quad ,
\end{equation}
and
\begin{equation}
\psi_{k+1}(\xi) - \mu_1 g_1 c_k \underset{\chi_1 \ll 1}{\approx}  \frac{\sqrt{\mu_1}}{\chi_1^{\frac{1}{4}}} \frac{c_{k+1} + c_{k}}{2} -  \frac{\sqrt{\mu_1}}{\chi_1^{\frac{1}{4}}} c_k =  \frac{\sqrt{\mu_1}}{\chi_1^{\frac{1}{4}}} \frac{c_{k+1} - c_{k}}{2} \quad .
\end{equation}
Now recall that from the underlying definition of the extended alphabet set, it follows that $c_k \ge c_{k-1} \ \forall k$, and it can hence be concluded that
\begin{equation}
\psi_k(\xi) - \mu_1 g_1 c_k \underset{\chi_1 \ll 1}{\approx} \begin{cases}
      -\infty & c_{k-1}=-\infty, \abs{c_k}<\infty \ , \\
      - \frac{\sqrt{\mu_1}}{\chi_1^{\frac{1}{4}}} \frac{c_k - c_{k-1}}{2} \rightarrow -\infty &  \abs{c_{k-1}}, \ \abs{c_k}<\infty \ ,
\end{cases}
\end{equation}
and
\begin{equation}
\psi_{k+1}(\xi) - \mu_1 g_1 c_k \underset{\chi_1 \ll 1}{\approx} \begin{cases}
       \frac{\sqrt{\mu_1}}{\chi_1^{\frac{1}{4}}} \frac{c_{k+1} - c_{k}}{2} \rightarrow \infty &  \abs{c_{k}}, \ \abs{c_{k+1}}<\infty \ , \\
      \infty & c_{k}<\infty, \abs{c_{k+1}}=\infty \ .
\end{cases}
\end{equation}
This implies that
\begin{eqnarray}
Q(\sqrt{2}(\psi_k(\xi) - \mu_1 g_1 c_k)) & \xrightarrow{\chi_1 \rightarrow 0} & 1 \ \forall k \\
Q(\sqrt{2}(\psi_{k+1}(\xi) - \mu_1 g_1 c_k)) & \xrightarrow{\chi_1 \rightarrow 0} & 0 \ \forall k \ .
\end{eqnarray}
We therefore conclude that
\begin{equation}
\begin{aligned}
\Theta_k(\xi)
&\underset{\chi_1 \ll 1}{\approx} \e^{\mu_1 c_k \left[ (\mu_1 g_1^2 - \varepsilon_1) c_k+ 2  f_1 \xi\right]} \underset{\chi_1 \ll 1}{\approx} \e^{\mu_1 c_k \left[ 2  f_1 \xi-R(-\mu_1 p_1) c_k\right]} =  \e^{\mu_1 c_k \left( 2  f \xi- \frac{c_k}{\sqrt{\mu_1 p_1}}\right)} \ .
\end{aligned}
\end{equation}
In a similar manner one can observe that the exponential terms in the RHS of \eqref{eq: 1RSB QPSK: Definition of Psi-xi} vanish as $\chi_1 \rightarrow 0$, and conclude that
\begin{equation}
\Psi_k(\xi) \xrightarrow{\chi_1 \rightarrow 0}  0 \quad \forall k \ .
\end{equation}

In view of the above we can now restate the coupled equations that determine the macroscopic parameters $q_1$, $p_1$, and $\mu_1$ in the following way (cf.\ \eqref{eq: 1RSB QPSK: Fixed point equation for q}--\eqref{eq: 1RSB QPSK: Fixed point equation for mu}, and note that the equation for determining $\chi_1$ can be ignored):
\begin{eqnarray}
q_1 &\underset{\chi_1 \ll 1}{\approx}& 2   \int \frac{ \sum_{m=1}^L c_m^2  \Theta_m(\xi) }{\sum_{m=1}^L \Theta_m(\xi) }  \e^{-\xi^2} \, \frac{\rd \xi}{\sqrt{\pi}} - p_1 \quad , %\label{eq: 1RSB QPSK: Fixed point equation for q}
\\
 p_1 &\underset{\chi_1 \ll 1}{\approx}& \frac{2}{f_1\mu_1}    \int \frac{\sum_{m=1}^L c_m \Theta_m(\xi)}{\sum_{m=1}^L \Theta_m(\xi)}
\xi \e^{-\xi^2} \, \frac{\rd \xi}{\sqrt{\pi}}  \ , %\label{eq: 1RSB QPSK: Fixed point equation for p}
\\
0   &\underset{\chi_1 \ll 1}{\approx}& 2 \int \log \left(\sum_{m=1}^L \Theta_m(\xi) \right) \e^{-\xi^2} \, \frac{\rd \xi}{\sqrt{\pi}}  \ , %\label{eq: 1RSB QPSK: Fixed point equation for mu}
\end{eqnarray}
where we used \eqref{app: eq: R-transform for Gaussian H - alpha 1}--\eqref{app: eq f for alpha 1 and chi small} to obtain
\begin{equation}
\int_0^{\mu_1 p_1} R(-w) \, \rd w = 2 \sqrt{\mu_1 p_1} \quad , \quad  R'(-\mu_1 p_1) = \frac{1}{2(\mu_1 p_1)^{\frac{3}{2}}} \quad .
\end{equation}
The energy penalty in this case is given by (cf.\ \eqref{eq: 1RSB: Limit of the effective energy penalty})
\begin{equation}
\bE_\rsb \underset{\chi_1 \ll 1}{\approx} \frac{q_1+p_1}{\sqrt{\mu_1 p_1}} - \frac{q_1 \mu_1 p_1}{2 (\mu_1 p_1)^{\frac{3}{2}}} = \frac{q_1+2p_1}{2\sqrt{\mu_1 p_1}} \quad .
\end{equation}

\subsection{Case II: $\alpha < 1$, $\alpha \rightarrow 1$}
\label{app: subsec: Case II - alpha lt 1}

In a similar manner to the previous section, we start with the $R$-transform of
$\Mat{J}=\Mat{T}^\dag\Mat{T}$, and rewrite it for small $w$, using the Taylor expansion around $w=0$, as
\begin{equation} %\label{app: eq: R-transform for Gaussian H - alpha 1}
\begin{aligned}
R(-w) &= \frac{\alpha - 1 + \sqrt{(1-\alpha)^2+4\alpha w}}{2\alpha w} \underset{w \ll 1}{\approx} , \frac{-(1-\alpha) + (1-\alpha)\left(1+\frac{2\alpha}{(1-\alpha)^2}w - \frac{2\alpha^2}{(1-\alpha)^4}w^2 + o(w^2) \right)}{2\alpha w}
\\
& \underset{w \ll 1}{\approx} \frac{1}{1-\alpha} - \frac{\alpha}{(1-\alpha)^3}w + O(w^2) \ .
\end{aligned}
\end{equation}
We focus in the following on the regime in which $\alpha \rightarrow 1$, so that $\frac{1}{1-\alpha} \gg 1$, but still $\frac{1}{1-\alpha} \ll \frac{1}{\chi_1}$. It hence follows that
\begin{eqnarray}
\varepsilon_1 & \underset{\chi_1 \ll 1}{\approx} & \frac{1}{1-\alpha} + O(\chi_1) \quad ,\\
g_1 & \underset{\chi_1 \ll 1}{\approx}  & \sqrt{\frac{\frac{1}{1-\alpha} - \frac{1}{\sqrt{\mu_1 p_1}}}{\mu_1}} \underset{\chi_1 \ll 1, \alpha\rightarrow 1}{\approx} \sqrt{\frac{1}{\mu_1(1-\alpha)}} \quad ,\\
f_1 & \underset{\chi_1 \ll 1}{\approx} & \sqrt{q_1 R'(-\mu_1 p_1)}% \underset{\chi \ll 1, \alpha\rightarrow 1}{\approx} \sqrt{q \frac{1}{2 (\mu p)^{\frac{3}{2}}}}
\quad . %\label{app: eq f for alpha 1 and chi small}
\end{eqnarray}
Particularizing again to the two-dimensional discrete lattice alphabet relaxation scheme for QPSK signaling, it follows from \eqref{eq: 1RSB QPSK: Definition of psi-xi} that
\begin{equation}\label{app: eq: 1RSB QPSK: Definition of psi-xi  - alpha lt 1}
\psi_k(\xi) = \frac{\varepsilon_1 v_k - f_1 \xi}{g_1} \underset{\chi_1 \ll 1, \alpha\rightarrow 1}{\approx} \frac{v_k}{1-\alpha} \sqrt{\mu_1 (1-\alpha)}  = \sqrt{\frac{\mu_1}{1-\alpha}} \frac{c_k + c_{k-1}}{2} \quad \forall \abs{\xi} < \infty \quad .
\end{equation}
Considering \eqref{eq: 1RSB QPSK: Definition of Theta-xi} we write
\begin{equation}
 (\mu_1 g_1^2 - \varepsilon_1) c_k+ 2  f_1 \xi \underset{\chi_1 \ll 1}{\approx} -R(-\mu_1 p_1)c_k + 2 f_1 \xi %\underset{\chi \ll 1, \alpha\rightarrow 1}{\approx} -\frac{c_k}{\sqrt{\mu p}} + 2 f \xi
 \quad .
\end{equation}
%where we approximated the R-transform by its $\alpha\rightarrow 1$ limit (see the previous subsection).
Next, the arguments of the $Q(\cdot)$ functions in \eqref{eq: 1RSB QPSK: Definition of Theta-xi} satisfy
\begin{equation}
\psi_k(\xi) - \mu_1 g_1 c_k \underset{\chi_1 \ll 1, \alpha\rightarrow 1}{\approx}  \sqrt{\frac{\mu_1}{1-\alpha}} \frac{c_k + c_{k-1}}{2} -  \sqrt{\frac{\mu_1}{1-\alpha}} c_k = - \sqrt{\frac{\mu_1}{1-\alpha}} \frac{c_k - c_{k-1}}{2} \quad ,
\end{equation}
and
\begin{equation}
\psi_{k+1}(\xi) - \mu_1 g_1 c_k \underset{\chi_1 \ll 1, \alpha\rightarrow 1}{\approx}  \sqrt{\frac{\mu_1}{1-\alpha}} \frac{c_{k+1} + c_{k}}{2} -  \sqrt{\frac{\mu_1}{1-\alpha}} c_k =  \sqrt{\frac{\mu_1}{1-\alpha}} \frac{c_{k+1} - c_{k}}{2} \quad .
\end{equation}
This enables us to conclude that
\begin{eqnarray}
\psi_k(\xi) - \mu_1 g_1 c_k  & \xrightarrow{\chi_1 \ll 1, \alpha\rightarrow 1} & -\infty \\
\psi_{k+1}(\xi) - \mu_1 g_1 c_k & \xrightarrow{\chi_1 \ll 1, \alpha\rightarrow 1} & \infty  \  ,
\end{eqnarray}
and hence
\begin{equation}
\begin{aligned}
\Theta_k(\xi)
&\underset{\chi_1 \ll 1, \alpha \rightarrow 1}{\approx} \e^{\mu_1 c_k \left[ (\mu_1 g_1^2 - \varepsilon_1) c_k+ 2  f_1 \xi\right]} \underset{\chi_1 \ll 1, \alpha \rightarrow 1}{\approx} \e^{\mu_1 c_k \left[ 2  f_1 \xi-R(-\mu_1 p_1) c_k\right]} %=  \e^{\mu c_k \left( 2  f \xi- \frac{c_k}{\sqrt{\mu p}}\right)}
\ , \quad \forall k \quad ,
\end{aligned}
\end{equation}
and
\begin{equation}
\Psi_k(\xi) \xrightarrow{\chi_1 \ll 0, \alpha\rightarrow 1}  0 \quad , \quad \forall k \ .
\end{equation}
Finally, note that $\int_{0}^{\mu_1 p_1} R(-w) \, dw$ exists for $\alpha<1$, and the approximation
\begin{equation}
\mu_1 \chi_1 g_1^2 \xrightarrow{\chi_1 \ll 0, \alpha\rightarrow 1} 0
\end{equation}
is employed to derive the three coupled equation that determine the macroscopic parameters $q_1$, $p_1$, and $\mu_1$. The three equations are thus
\begin{eqnarray}
q_1 &\underset{\chi_1 \ll 1, \alpha \rightarrow 1}{\approx}& 2   \int \frac{ \sum_{m=1}^L c_m^2  \Theta_m(\xi) }{\sum_{m=1}^L \Theta_m(\xi) }  \e^{-\xi^2} \, \frac{\rd \xi}{\sqrt{\pi}} - p_1 \quad , %\label{eq: 1RSB QPSK: Fixed point equation for q}
\\
 p_1 &\underset{\chi_1 \ll 1, \alpha \rightarrow 1}{\approx}& \frac{2}{f_1\mu_1}    \int \frac{\sum_{m=1}^L c_m \Theta_m(\xi)}{\sum_{m=1}^L \Theta_m(\xi)}
\xi \e^{-\xi^2} \, \frac{\rd \xi}{\sqrt{\pi}}  \ , %\label{eq: 1RSB QPSK: Fixed point equation for p}
\\
\mu_1   &\underset{\chi_1 \ll 1, \alpha \rightarrow 1}{\approx}&    \frac{2 \int \log \left(\sum_{m=1}^L \Theta_m(\xi) \right) \e^{-\xi^2} \, \frac{\rd \xi}{\sqrt{\pi}} -\int_{0}^{\mu_1 p_1}
R(-w)\, \rd w + \mu_1 ( q_1  +2 p_1) R(-\mu_1 p_1) }{2 q_1\mu_1 p_1 R'( - \mu_1 p_1)}  \ . \nonumber \\ %\label{eq: 1RSB QPSK: Fixed point equation for mu}
\end{eqnarray}
The expression for the energy penalty is given by
\begin{equation}
\bE_\rsb \underset{\chi_1 \ll 1, \alpha \rightarrow 1}{\approx}  \left(q_1+p_1\right)R(-\mu_1 p_1)
 - q_1 \mu_1 p_1 R'( - \mu_1 p_1)\ Ê.
\end{equation}
We also note that the exact expressions for the $R$-transform and its derivative were employed for the purpose of producing more accurate numerical results, while using this small $\chi_1$ approximation for $\alpha<1$.

%%%%%%%%%%%%%%%%%
\section{Proof of Lemma \ref{incRtrafo}}
\label{prooflemma}

The Stieltjes transform of the probability distribution $F(x)$ is defined by
\begin{equation}
m(s) = \int \frac{{\rm d}F(x)}{x-s}\quad .
\end{equation}
In terms of the Stieltjes transform, the $R$-transform is defined as
\begin{equation}
R(w) = m^{-1}(-w)-\frac1w \quad ,
\end{equation}
where $m^{-1}(s)$ denotes the inverse function of $m(s)$ with respect to composition, i.e., $m(m^{-1}(s))=s$.

We start with the observation that the derivative of the Stieltjes transform is lower bounded by its square 
\begin{equation}
m^\prime(s) = \int \frac{{\rm d}F(x)}{(x-s)^2} \ge [m(s)]^2 \quad ,
\end{equation}
by means of Jensen's inequality, with equality if and only if the distribution $F(x)$ is a single mass point.
Next, we consider the derivative of the $R$-transform. Letting $w=m(s)$, it follows that
\begin{align}
R^\prime (w) &= \frac{{\rm d} m^{-1}(-w)}{{\rm d}w} + \frac1{w^2}\\
& = \frac{-1}{m^\prime(s)} + \frac1{[m(s)]^2} \ge0 \quad ,
\end{align}
with equality if and only if the distribution $F(x)$ is a single mass point. Lemma \ref{incRtrafo} then follows immediately.

\if 0
%%%%%%%%%%%%%%%%%%%%%%%%%%%%%%%%%%%%%
\section{Convex Relaxation}

We have from Prop.~\ref{prop:general_saddle_point_free_energy}
\begin{equation}
\Mat Q = \int \frac{\sum\limits_{{\bf x}\in \B_{u}^n} {\bf xx^\dagger}{\rm e}^{\,\beta{\bf x}^\dagger R(-\beta\Mat Q)\bf x}}
{\sum\limits_{{\bf x}\in \B_{u}^n} {\rm e}^{\,\beta {\bf x}^\dagger R(-\beta\Mat Q)\bf x}}
\,{\rm d} F_u(u) \quad. 
\end{equation}
For convex relaxation, we get due to the symmetry $\B_{1}=-\B_{-1}=[1;\infty)$
\begin{equation}
\Mat Q =  \frac{\sum\limits_{{\bf x}\in \B_{1}^n} {\bf xx^\dagger}{\rm e}^{\,\beta{\bf x}^\dagger R(-\beta\Mat Q)\bf x}}
{\sum\limits_{{\bf x}\in \B_{1}^n} {\rm e}^{\,\beta {\bf x}^\dagger R(-\beta\Mat Q)\bf x}}
 \quad. 
\end{equation}
\fi

%%%%%%%%%%%%%%%%%%%%%%%%%%%%%%%%%%%%%
\section{Spectral Efficiency of Generalized Tomlinson-Harashima Precoding}
\label{App: Spectral Efficiency of Generalized Tomlinson-Harashima Precoding}

For the sake of comparison, we review here the derivation of the spectral efficiency of \emph{generalized Tomlinson-Harashima precoding (GTHP)}, which is another practical alternative to the capacity achieving DPC.
The approach is based on inflated lattice strategies, and borrows ideas from the recent analysis of pulse amplitude modulation (PAM) in \cite{Gariby-Erez-Shamai-ISIT-2007}. The spectral efficiency is derived following \cite{Zamir-Shamai-Erez-2002,Erez-Shamai-Zamir-2005} (see also \cite{Ginis-Cioffi-JSAC-2002, Boccardi-Tosato-Caire-IZS-2006}), while employing successive encoding using the inflated lattice strategy %considered in \cite{Gariby-Erez-Shamai-ISIT-2007} 
at each stage, where the signals of previously encoded users are treated as \emph{causally} known interference. %(more on the notion of causality in the current context in the following). 
We consider here the ``canonical" channel model as in \eqref{eq: Basic System Model}, and note that a comparative analysis of other variants of GTHP can be found, e.g., in \cite{Boccardi-Tosato-Caire-IZS-2006}.

The underlying idea of the scheme considered here is first to induce a ``triangular" channel structure using the $LQ$-factorization of the channel transfer matrix. Assuming $\Mat{H}$ is full rank, we denote
\begin{equation}\label{App eq: LQ-factorization of H}
\Mat{H} = \Mat{L}\bar{\textbf{Q}} \quad ,
\end{equation}
where $\Mat{L}_{[K\times K]}$ is lower triangular with positive diagonal entries and $\bar{\textbf{Q}}_{K\times N}$ has orthonormal rows. The transmitted signal is then given by
\begin{equation}
\vct{t}=\bar{\textbf{Q}}^\dag \vct{x} \quad ,
\end{equation}
where $\vct{x}$ is the nonlinear precoder's output (cf.\ \eqref{eq: Relation of transmitted vector to the vector of coded symbols}). The signal received by the $k$th user is thus given by
\begin{equation}
\label{App eq: Signal receiver by the kth user with LQ factorization}
r_k = L_{kk} x_k + \sum_{j=1}^{k-1} L_{kj} x_j + n_k \quad ,
\end{equation}
where $\set{L_{ij}}$ denote the entries of $\Mat{L}$ and $x_i$ is the nonlinear precoder's output that corresponds to user $i$. Normalizing both sides of the equation by $L_{kk}$, we get the following equivalent channel
\begin{equation}
\label{App eq: Equivalent LQ channel after normalization}
\begin{aligned}
\breve{r}_k &=  x_k + \sum_{j=1}^{k-1} \frac{L_{kj}}{L_{kk}} x_j + \breve{n}_k \\
&= x_k + s_k + \breve{n}_k
\quad ,
\end{aligned}
\end{equation}
where we denote the multiuser interference experienced by user $k$ as $s_k \triangleq \sum_{j=1}^{k-1} \frac{L_{kj}}{L_{kk}} x_j$, and $\breve{n}_k$ is a zero-mean circularly symmetric complex Gaussian noise with variance $\frac{\sigma^2}{L_{kk}^2}$.

In the GTHP setting, instead of DPC (as employed at this point, e.g., by the ``zero-forcing dirty-paper" scheme of \cite{Caire-Shamai-2003}), we follow for each user the THP-type strategy described in \cite{Erez-Shamai-Zamir-2005} for canceling the interference due to previously encoded users. This strategy, which applies for canceling \emph{causally} known interference, leads in the broadcast setting to a considerably reduced complexity as it involves only scalar quantizations (as opposed to vector quantizations in the noncausal case, see therein). To make a fair comparison to the precoding schemes discussed in Sections \ref{sec: Replica Symmetry Breaking Example}-\ref{sec: A Replica Symmetric Example}, we particularize here to the case in which the information bearing signal takes on binary values per each dimension (so that the total spectral efficiency for quadrature modulation, as a function of $\febno$, is twice as much as the one obtained for binary input). The spectral efficiency for continuous input is derived as well for completeness. The basic transmission scheme is reviewed first, while considering real channels. 

The underlying \emph{real} channel model is given by
\begin{equation}
y=x + s + n \quad ,
\end{equation}
where $x$ is subject to an average power constraint $P_x$, $n$ is a zero-mean AWGN with variance $P_n$, and $s$ is an interference signal which is known causally at the transmitter (i.e., at the current time instance), but not at the receiver. This channel model is also referred to in the literature as the ``dirty-tape" model \cite{Erez-Shamai-Zamir-2005}. Consider the one-dimensional lattice
\begin{equation}
\Lambda = \Delta \set { \cdots -3,-1,1,3, \cdots} \ .
\end{equation}
Let $\mathcal{V}=[-\Delta, \Delta)$ denote the basic Voronoi region of $\Lambda$. Let $d$ be a dither signal uniformly distributed over $\mathcal{V}$. Under a common randomness assumption, this dither signal is assumed to be available at the receiver as well. 

Starting with \emph{continuous} information bearing signals, then by the GTHP scheme the transmitter sends the signal
\begin{equation}\label{App: eq: Definition of x for the continuous case}
x = [v - \tilde{\alpha} s - d] \mod \Lambda \quad ,
\end{equation}
where $\tilde{\alpha} \in (0,1]$ is referred to as the ``inflation factor". The receiver scales the received signal by $\tilde{\alpha}$, adds the dither signal, and then performs a modulo-lattice operation, yielding
\begin{equation}
y' = [\tilde{\alpha} y + d] \mod \Lambda \quad .
\end{equation}
Effectively, the induced channel is equivalent to (see \cite{Erez-Shamai-Zamir-2005}, Lemma 6)
\begin{equation}
\label{eq: Effective received signal for continuous input}
y' = [v+n_{\text{eff}}] \mod \Lambda \quad ,
\end{equation}
where the effective noise is given by
\begin{equation}
\label{eq: Effective noise for continuos input}
n_{\text{eff}} \triangleq [(\tilde{\alpha}-1)x  + \tilde{\alpha} n] \mod \Lambda \quad ,
\end{equation}
and we note that the dither signal ensures that $x$ is uniformly distributed over the Voronoi region, and is independent of either the information bearing signal $v$, or the noise $n$.

The capacity of this channel is achieved by a uniform input distribution over the Voronoi region, $v \sim \Unif    \set{\mathcal{V}}$, for which the relation between the lattice constant $\Delta$ and the transmit power $P_x$ is given by $\Delta=\sqrt{3 P_x}$. The corresponding achievable rate is equal to the input-output mutual information of the equivalent channel \eqref{eq: Effective received signal for continuous input}
\begin{equation}
\label{App eq: Achievable rate for SU channel with GTHP and continuous input}
\begin{aligned}
R(P_x) \triangleq \mathrm{I}(v,y') &= \log (2\Delta) - \mathrm{h}(n_{\text{eff}}) \\
&= \frac{1}{2} \log (12 P_x)  - \int_{-\Delta}^\Delta f_{n_{\text{eff}}}(\zeta) \log_2 f_{n_{\text{eff}}}(\zeta) \, \rd \zeta \, \Biggr|_{\Delta= \sqrt{ 3 P_{x}}} \quad .
\end{aligned}
\end{equation}
The entropy of the effective noise is derived via the following observation.
Denoting the ``self-noise" term by
\begin{equation}
\label{eq: Definition of the self noise for continuous input}
Z= (\tilde{\alpha} - 1)x \quad,
\end{equation}
its pdf is given by 
\begin{equation}
\label{eq: PDF of self noise for quantized inflated lattice - Continuous Input}
f_{Z}(\zeta) = \begin{cases}
   \frac{1}{2(1-\tilde{\alpha})\Delta}   & \abs{\zeta} \le (1-\tilde{\alpha}) \Delta \ , \\
   0 & \text{otherwise} \ .
\end{cases}
\end{equation}
The pdf of the effective noise \eqref{eq: Effective noise for continuos input} is thus given by
\begin{equation}
\label{eq: pdf of the effective noise for continuous input}
f_{n_{\text{eff}}}(\zeta) = 
\begin{cases}
     \sum_{i=-\infty}^{\infty} f_{\tilde{Z}}(\zeta-2 i \Delta) & -\Delta \le \zeta < \Delta \ , \\
     %\quad \quad \quad 
     0 & \text{otherwise} \ , 
\end{cases} 
\end{equation}
where $f_{\tilde{Z}}(\zeta)$ denotes the pdf of the pre-modulo noise term, which is given by the convolution of the pdf of the self-noise \eqref{eq: PDF of self noise for quantized inflated lattice - Continuous Input} and the pdf of the scaled AWGN
\begin{equation}
\begin{aligned}
f_{\tilde{Z}}(\zeta) &= f_{Z}(\zeta) \ast f_{\tilde{\alpha}n}(\zeta) \\
&= \frac{1}{2(1-\tilde{\alpha})\Delta}\left[ Q\left( \frac{\sqrt{2}}{\tilde{\alpha}}\left(\zeta-(1-\tilde{\alpha})\Delta\right)\right) -
Q\left( \frac{\sqrt{2}}{\tilde{\alpha}}\left(\zeta+(1-\tilde{\alpha})\Delta\right)\right)\right]
\quad .
\end{aligned}
\end{equation}
We normalized here without loss of generality the spectral level of the AWGN to $\frac{1}{2}$ per dimension (inducing a unit noise spectral level in complex channels \cite{Verdu-paper-low-snr-regime-02}), so that effectively $P_{x}$ specifies the SNR of the original underlying \emph{complex} channel model, corresponding to \eqref{App eq: Equivalent LQ channel after normalization}. 
The rate in \eqref{App eq: Achievable rate for SU channel with GTHP and continuous input} can be optimized with respect to the inflation factor $\tilde{\alpha}$ (which is performed to obtain the numerical results shown in Section \ref{sec: Spectral Efficiency Comparison}). We also note that choosing $\tilde{\alpha}=1$ corresponds to standard THP, while another popular choice is the minimum mean-squared error (MMSE) factor (also referred to as the ``Costa factor'' \cite{Costa-83})
\begin{equation}
\alpha_{\textrm{MMSE}} = \frac{P_x}{P_x + P_n} \quad .
\end{equation}

Turning to discrete input with \emph{M-pulse amplitude modulation (M-PAM)} (representing the information bearing signals), the setting is equivalent to the case in which the continuous information bearing signal considered above is \emph{quantized} (cf.\ \cite{Erez-Zamir-2004}).  Instead of \eqref{App: eq: Definition of x for the continuous case}, the channel input is now given by
\begin{equation}
x  = [\Q(v) - \tilde{\alpha} s - d] \mod \Lambda \triangleq [v_\Q - \tilde{\alpha} s - d] \mod \Lambda \quad ,
\end{equation}
where $\Q(\cdot)$ denotes the nearest-neighbor uniform quantizer with step size $\Delta$ \cite{Gariby-Erez-Shamai-ISIT-2007}, and $v$ is assumed to be uniformly distributed over the Voronoi region. We note here that this transmission scheme differs from the one considered in \cite{Gariby-Erez-Shamai-ISIT-2007}, where the \emph{channel input} is quantized to comply with an M-PAM constellation (see therein). Note also that as in the continuous setting, due to the dither signal, the channel input $x$ is still uniformly distributed over the Voronoi region.
The effective channel can now be represented in the form (cf.\  \eqref{eq: Effective received signal for continuous input})
\begin{equation}
\label{eq: Effective received signal for discrete input}
y'_\Q = [v_\Q+n_{\text{eff}}] \mod \Lambda \quad ,
\end{equation}
where the effective noise is still given by  \eqref{eq: Effective noise for continuos input}.
%\begin{equation}
%\label{eq: Effective noise for discrete input}
%n^\Q_{\text{eff}} = [ (\tilde{\alpha} \Q(x) - x) +   \tilde{\alpha} n ] \mod \Lambda \quad .
%\end{equation}
Restricting this review to the case of binary information bearing signals per dimension, the channel input signal is limited to the interval $\V=[-\Delta,\Delta)$, while the quantized information bearing signal is obtained using
\begin{equation}
\Q(v) = \begin{cases}
      -\frac{\Delta}{2} & -\Delta \le v < 0 \ , \\
     + \frac{\Delta}{2} & 0 \le v < \Delta \ .
\end{cases}
\end{equation}
For consistency we retain the relation $\Delta = \sqrt{3 P_{x_{\Q}}}$ .

The achievable rate for binary input is given again by the mutual information
\begin{equation}
\label{eq: Achievable rate for binary quantized input}
\begin{aligned}
R(P_{x_{\Q}}) \triangleq \textrm{I}(v;y'_\Q) &= \textrm{h}(y'_\Q) - \textrm{h}(n_{\text{eff}}) \quad .
%\\
%&= \log_2(4\sqrt{P_{x_{\Q}}}) + \int_{-\Delta}^\Delta f_{n_{\text{eff}}}(\zeta) \log_2 f_{n_{\text{eff}}}(\zeta) \, d \zeta \\
%&= 2 + \frac{1}{2} \log_2 P_{x_{\Q}} + \int_{-\Delta}^\Delta f_{n_{\text{eff}}}(\zeta) \log_2 f_{n_{\text{eff}}}(\zeta) \, d \zeta
%\quad .
%\\
%&= \log(4\sqrt{P_{x_{\Q}}}) + \sum_{i=-\infty}^{\infty} \int_{-\infty}^\infty f_{\tilde{Z}}(\zeta-2 i \Delta) \log_2\left( \sum_{i=-\infty}^{\infty} f_{\tilde{Z}}(\zeta-2 i \Delta) \right) \, d\zeta \quad .
\end{aligned}
\end{equation}
Note that the pdf of the random quantity inside the modulo function in \eqref{eq: Effective received signal for discrete input} is given by
\begin{equation}\label{App: eq: pdf of y for discrete input - before modulo}
\begin{aligned}
f_{\tilde{Y}}(\zeta) &= f_{v_\Q}(\zeta) \ast f_{\tilde{Z}}(\zeta) \\
&= \left( \frac{1}{2} \delta\left(\zeta-\frac{\Delta}{2}\right) +   \frac{1}{2} \delta\left(\zeta+\frac{\Delta}{2}\right)   \right) \ast f_{\tilde{Z}}(\zeta) \\
&=\frac{1}{2} f_{\tilde{Z}}\left(\zeta-\frac{\Delta}{2}\right) + \frac{1}{2} f_{\tilde{Z}}\left(\zeta+\frac{\Delta}{2}\right) \quad .
\end{aligned}
\end{equation}
Hence, the pdf of the equivalent channel output $y'_\Q$ is equal to 
\begin{equation}
\label{eq: PDF the effective channel output - discrete input}
f_{y'_\Q}(\zeta) = \begin{cases}
   \sum_{i=-\infty}^\infty f_{\tilde{Y}}(\zeta-2 i \Delta)   & \Delta \le \zeta < \Delta \ , \\
   0 & \text{otherwise} \ ,
\end{cases}
\end{equation}
and the achievable rate of \eqref{eq: Achievable rate for binary quantized input} can be rewritten as
\begin{equation}
\label{eq: Achievable rate for binary quantized input - explicit form}
\begin{aligned}
R(P_{x_{\Q}}) &= -\int_{-\Delta}^\Delta f_{y'_\Q}(\zeta) \log_2 f_{y'_\Q}(\zeta) \, \rd \zeta + \int_{-\Delta}^\Delta f_{n_{\text{eff}}}(\zeta) \log_2 f_{n_{\text{eff}}}(\zeta) \, \rd \zeta \\
&= \int_{-\Delta}^\Delta \left[ f_{n_{\text{eff}}}(\zeta) \log_2 f_{n_{\text{eff}}}(\zeta) - f_{y'_\Q}(\zeta) \log_2 f_{y'_\Q}(\zeta)\right] \, \rd \zeta \, \Biggr|_{\Delta= \sqrt{ 3 P_{x_\Q}}} \quad .
%&= 2 + \frac{1}{2} \log_2 P_{x_{\Q}} + \int_{-\Delta}^\Delta f_{n_{\text{eff}}}(\zeta) \log_2 f_{n_{\text{eff}}}(\zeta) \, d \zeta
%\quad .
%\\
%&= \log(4\sqrt{P_{x_{\Q}}}) + \sum_{i=-\infty}^{\infty} \int_{-\infty}^\infty f_{\tilde{Z}}(\zeta-2 i \Delta) \log_2\left( \sum_{i=-\infty}^{\infty} f_{\tilde{Z}}(\zeta-2 i \Delta) \right) \, d\zeta \quad .
\end{aligned}
\end{equation}

The above principles can now be applied to the channel in \eqref{App eq: Signal receiver by the kth user with LQ factorization}, where the transmitter pre-cancells using the GTHP scheme, per each transmitted symbol, the interference due to the \emph{corresponding symbols} of previously encoded users. 
%As described above, only scalar modulo operations are involved in this scheme, as opposed to the ``non-causal" interference cancellation approach in \cite{Gariby-Erez-Shamai-ISIT-2007} that involves \emph{$n$-dimensional modulo lattice operations}, where $n$ here denotes the codeword length of the users (which is significantly more computationally demanding).
Using \eqref{App eq: Achievable rate for SU channel with GTHP and continuous input} and \eqref{eq: Achievable rate for binary quantized input}, the achievable rate of the $k$th user can be obtained by substituting $P_x=L_{kk}^2 \, \snr$ for continuous input, and $P_{x_{\Q}} = L_{kk}^2 \, \snr$ for the binary setting, yielding, respectively, for \emph{real} channels
\begin{equation}
\label{App eq: Achievable rate by user k - GTHP - Continuous}
R^{\gthp}_{C,k}(\snr) =  \frac{1}{2} \log (12 L_{kk}^2 \snr)   + \int_{-\Delta}^\Delta f_{n^C_{\text{eff}}}(\zeta) \log_2 f_{n^C_{\text{eff}}}(\zeta) \, \rd \zeta \Biggr|_{\Delta=\sqrt{3 L_{kk}^2 \snr}} \quad ,
\end{equation}
and
\begin{equation}
\label{App eq: Achievable rate by user k - GTHP - BPSK}
R^{\gthp}_{\Q,k}(\snr) =  
 \int_{-\Delta}^\Delta \left[ f_{n_{\text{eff}}}(\zeta) \log_2 f_{n_{\text{eff}}}(\zeta) - f_{y'_\Q}(\zeta) \log_2 f_{y'_\Q}(\zeta)\right] \, \rd \zeta \, \Biggr|_{\Delta= \sqrt{ 3 L_{kk}^2 \, \snr}} \quad .
%2 + \frac{1}{2} \log_2 (L_{kk}^2 \snr) + \int_{-\Delta}^\Delta f_{n^\Q_{\text{eff}}}(\zeta) \log_2 f_{n^\Q_{\text{eff}}}(\zeta) \, d \zeta \Biggr|_{\Delta=2 \sqrt{ L_{kk}^2 \snr}} \quad .
\end{equation}

To complete the analysis, it is left to derive the normalized spectral efficiency of GTHP in the large system limit. This is obtained using the following observation (see \cite[Lemma 3]{Caire-Shamai-2003}).
\begin{lem}\label{App: Lem: Lemma on the limiting values of the Cholesky diagonals}
Let $\Mat{H}$ be a $K \times N$ random matrix, having i.i.d.\ circularly symmetric zero-mean entries with variance $\frac{1}{N}$ and finite fourth moment, and let $\Mat{H}^{(k)}$, $k<K$, denote the matrix constructed by striking out the \emph{last} $K-k$ rows of $\Mat{H}$. Then
\begin{equation}\label{App: eq: Relation between the Cholesky diagonal and entires of invHH}
L_{kk}^2 = \frac{1}{\left[ (\Mat{H}^{(k)} {\Mat{H}^{(k)}}^\dag)^{-1} \right]_{kk}} \quad ,
\end{equation}
and for $k,K,N\to\infty$, s.t.\ $\frac{K}{N} \to \alpha <\infty$ and $\frac{k}{K}\to \nu\in[0,1)$, it follows that
\begin{equation}\label{App: eq: Limiting Cholesky Square Diagonals}
L_{kk}^2 \xrightarrow[k,K,N\to\infty]{} L^2(\nu)  \triangleq {1-\nu\alpha} \ , \quad \alpha\in(0,1] \quad .
\end{equation}
\end{lem}
Omitting subscripts, the limiting spectral efficiency is thus given for either continuous or binary quantized input by (cf.\ \cite[Eq.\ (41)]{Caire-Shamai-2003})
\begin{align}
C^{\gthp}(\snr) &= \lim_{K,N\to\infty} \frac{1}{N} \sum_{k=1}^K R^{\gthp}_k(\snr) \\
&= \lim_{K,N\to\infty} \frac{K}{N} \frac{1}{K} \sum_{k=1}^K R^{\gthp}_k(\snr) \\
&= \alpha \int_0^1 R^{\gthp}(\nu,\snr) \, \rd \nu \quad , \label{App: eq: GTHP spectral efficiency for with integral over L-nu}
\end{align}
where for the continuous case we substitute
\begin{equation}\label{App: eq: rate function for continuous input and GTHP}
R^{\gthp}(\nu,\snr) = R^{\gthp}_C(\nu,\snr) \triangleq \frac{1}{2} \log_2 (12 L^2(\nu) \snr) + \int_{-\Delta}^\Delta f_{n_{\text{eff}}}(\zeta) \log_2 f_{n_{\text{eff}}}(\zeta) \, \rd \zeta \Biggr|_{\Delta= \sqrt{3 L^2(\nu) \snr}} \quad ,
\end{equation}
and for the case of discrete binary information bearing signals we substitue
\begin{equation}\label{App: eq: rate function for BPSK input and GTHP}
R^{\gthp}(\nu,\snr)=R^{\gthp}_\Q(\nu,\snr) \triangleq 
 \int_{-\Delta}^\Delta \left[ f_{n_{\text{eff}}}(\zeta) \log_2 f_{n_{\text{eff}}}(\zeta) - f_{y'_\Q}(\zeta) \log_2 f_{y'_\Q}(\zeta)\right] \, \rd \zeta \, \Biggr|_{\Delta= \sqrt{ 3 L^2(\nu) \, \snr}} \quad .
%2 + \frac{1}{2} \log_2 (L(\nu)^2 \snr) + \int_{-\Delta}^\Delta f_{n^\Q_{\text{eff}}}(\zeta) \log_2 f_{n^\Q_{\text{eff}}}(\zeta) \, d \zeta \Biggr|_{\Delta= \sqrt{3 L(\nu)^2 \snr}} \quad .
\end{equation}
The spectral efficiency for QPSK modulation satisfies (following the convention in \cite{Verdu-paper-low-snr-regime-02}): 
\begin{equation}\label{App: Relation between spectral efficiency of QPSK and BSPK for GTHP}
C^{\gthp}_\qpsk(\snr)=2C^{\gthp}_\bpsk(\snr/2) \quad ,
\end{equation}
where $C^{\gthp}_\bpsk(\snr)$ is given by \eqref{App: eq: GTHP spectral efficiency for with integral over L-nu} and \eqref{App: eq: rate function for BPSK input and GTHP}, and it can be expressed as a function of $\febno$ through \eqref{eq: Relation between SNR and EbNo}.
An analogous result for the case of continuous input can be readily obtained using \eqref{App: eq: rate function for continuous input and GTHP}. Both spectral efficiencies can be optimized with respect to the choice of the %inflation factor $\tilde{\alpha} \in(0,1]$, as well as the 
system load $\alpha$.

\bibliographystyle{IEEEtran}
\bibliography{IEEEabrv,Main_Reference_List}
% ------------------------------------------------------------------------
%\begin{thebibliography}{1}

%\bibitem{Shamai-Verdu-Zaidel-ISIT-2002}
%S.~{Shamai (Shitz)}, B.~M. Zaidel, and S.~Verd\'{u}.
%\newblock Strongest-users-only detectors for randomly spread {CDMA}.
%\newblock In {\em Proceedings of the 2002 IEEE International Symposium on
%  Information Theory (ISIT'2002)}, page~20, Lausanne, Switzerland, June 30 --
%  July 5, 2002.

%\end{thebibliography}

%-------------------------------------------------------------------------

% ------------------------------------------------------------------------
\end{document}